\documentclass[sigconf, nonacm]{acmart}

\setcopyright{none}
\acmConference[LICS]{41st Annual ACM/IEEE Symposium on Logic in Computer Science}{}{Lisbon, Portugal}
\acmDOI{}
\acmISBN{}

\usepackage{mathtools}
\usepackage{mdframed}
\usepackage{graphicx}

\newcommand{\leftmapsto}{\leftarrow\!\shortmid}
\newcommand{\sem}{[-]}
\newcommand{\interp}[1]{\left\llbracket#1\right\rrbracket}

\usepackage{tikz}
\usepackage{tikz-cd}
\usetikzlibrary{graphs,decorations.pathmorphing,decorations.markings}
\usetikzlibrary{calc}
\usepackage{tikzit}
\usepackage{caption}
\usepackage{subcaption}
\usepackage{relsize}

\usepackage{quiver}
\usepackage{thmtools, thm-restate}

\input{quantum.tikzdefs}
\definecolor{green}{HTML}{ccffcc}
\definecolor{red}{HTML}{ff8888}
\definecolor{grey}{RGB}{211,211,211}
\tikzstyle{blue label}=[text=blue]
\tikzstyle{gate}=[shape=rectangle, text height=1.5ex, text depth=0.25ex, yshift=0.5mm, fill=white, draw=black, minimum height=3mm, yshift=-0.5mm, minimum width=3mm, font={\small}, tikzit category=circuit]
\tikzstyle{big gate}=[shape=rectangle, text height=1.5ex, text depth=0.25ex, yshift=0.5mm, fill=white, draw=black, minimum height=10mm, yshift=-0.5mm, minimum width=5mm, font={\small}, tikzit category=circuit]
\tikzstyle{Z dot}=[inner sep=0mm, minimum size=2mm, shape=circle, draw=black, fill=white, tikzit category=zx]
\tikzstyle{Z phase dot}=[minimum size=5mm, font={\footnotesize\boldmath}, shape=rectangle, rounded corners=2mm, inner sep=0.2mm, outer sep=-2mm, scale=0.8, tikzit shape=circle, draw=black, fill=white, tikzit draw=blue, tikzit category=zx]
\tikzstyle{X dot}=[Z dot, shape=circle, draw=black, fill={gray!40!white}, tikzit category=zx]
\tikzstyle{X phase dot}=[Z phase dot, tikzit shape=circle, tikzit draw=blue, fill={gray!40!white}, font={\footnotesize\boldmath}, tikzit category=zx]
\tikzstyle{mult}=[fill=yellow, inner sep=0mm, minimum size=1.75mm, shape=semicircle, draw=black, rotate=-90, tikzit category=zx]
\tikzstyle{multop}=[mult, rotate=180, tikzit category=zx]
\tikzstyle{umult}=[mult, rotate=90, tikzit category=zx]
\tikzstyle{dmult}=[mult, rotate=270, tikzit category=zx]
\tikzstyle{kscalar}=[draw, fill=white, rounded corners=0.9ex, inner sep=0.25em, scale=1, font={\scriptsize}, rectangle round south west=false, rectangle round north west=false, tikzit fill={rgb,255: red,129; green,253; blue,255}, tikzit shape=rectangle]
\tikzstyle{kscalarop}=[draw, fill=white, rounded corners=0.9ex, inner sep=0.25em, scale=1, font={\scriptsize}, rectangle round north east=false, rectangle round south east=false, tikzit fill={rgb,255: red,116; green,172; blue,255}]
\tikzstyle{ukscalar}=[draw, fill=white, rounded corners=0.9ex, inner sep=0.25em, scale=1, font={\scriptsize}, rectangle round south west=false, rectangle round south east=false]
\tikzstyle{dkscalar}=[draw, fill=white, rounded corners=0.9ex, inner sep=0.25em, scale=1, font={\scriptsize}, rectangle round north west=false, rectangle round north east=false]
\tikzstyle{rexpon}=[fill=white, draw=black, shape=regular polygon, regular polygon sides=3, regular polygon rotate=30, scale=0.7, inner sep=1pt, tikzit category=circuit, tikzit shape=rectangle, tikzit fill=red]
\tikzstyle{lexpon}=[fill=white, draw=black, shape=regular polygon, regular polygon sides=3, regular polygon rotate=-30, scale=0.7, inner sep=1pt, tikzit category=circuit, tikzit shape=rectangle, tikzit fill=green]
\tikzstyle{antip}=[fill=black, draw=black, shape=regular polygon, regular polygon sides=3, regular polygon rotate=30, scale=0.5]
\tikzstyle{recip}=[fill=yellow, draw=white, shape=circle, tikzit category=circuit, label={[rotate=45]center:$\ominus$}, inner sep=0pt, minimum width=2.1mm, tikzit fill={rgb,255: red,102; green,204; blue,255}, tikzit draw=black]
\tikzstyle{hadamard}=[fill=yellow, draw=black, shape=rectangle, inner sep=0.6mm, minimum height=1.5mm, minimum width=1.5mm, tikzit category=zx]
\tikzstyle{paulibox}=[fill={rgb,255: red,221; green,221; blue,255}, draw=black, shape=rectangle, inner sep=0.6mm, minimum height=5mm, minimum width=5mm, font={\footnotesize}, text height=1.5ex, text depth=0.25ex, tikzit category=zx]
\tikzstyle{vertex}=[inner sep=0mm, minimum size=1mm, shape=circle, draw=black, fill=black, tikzit category=misc]
\tikzstyle{vertex set}=[inner sep=0mm, minimum size=1mm, shape=circle, draw=black, fill=white, font={\footnotesize\boldmath}, tikzit category=misc]
\tikzstyle{small black dot}=[fill=black, draw=black, shape=circle, inner sep=0pt, minimum width=1.2mm, tikzit category=circuit]
\tikzstyle{cnot ctrl}=[fill=black, draw=black, shape=circle, inner sep=0pt, minimum width=1.2mm, tikzit category=circuit]
\tikzstyle{cnot targ}=[fill=white, draw=white, shape=circle, tikzit category=circuit, label={center:$\oplus$}, inner sep=0pt, minimum width=2.1mm, tikzit fill={rgb,255: red,102; green,204; blue,255}, tikzit draw=black]
\tikzstyle{ket}=[fill=white, draw=black, shape=regular polygon, regular polygon sides=3, regular polygon rotate=-30, scale=0.7, inner sep=1pt, tikzit category=circuit, tikzit shape=rectangle, tikzit fill=green]
\tikzstyle{bra}=[fill=white, draw=black, shape=regular polygon, regular polygon sides=3, regular polygon rotate=30, scale=0.7, inner sep=1pt, tikzit category=circuit, tikzit shape=rectangle, tikzit fill=red]
\tikzstyle{xket}=[fill={gray!40!white}, draw=black, shape=regular polygon, regular polygon sides=3, regular polygon rotate=60, scale=0.7, inner sep=1pt, tikzit category=circuit, tikzit shape=rectangle, tikzit fill=green]
\tikzstyle{xbra}=[fill={gray!40!white}, draw=black, shape=regular polygon, regular polygon sides=3, regular polygon rotate=0, scale=0.7, inner sep=1pt, tikzit category=circuit, tikzit shape=rectangle, tikzit fill=red]
\tikzstyle{scalar}=[shape=rectangle, text height=1.5ex, text depth=0.25ex, yshift=0.5mm, fill=white, draw=black, minimum height=5mm, yshift=-0.5mm, minimum width=5mm, font={\footnotesize}, scale=0.8]
\tikzstyle{clabel}=[fill=white, draw=none, shape=rectangle, tikzit fill={rgb,255: red,56; green,255; blue,242}, font={\footnotesize}, inner sep=1pt, tikzit category=labels]
\tikzstyle{empty diagram}=[draw={gray!40!white}, dashed, shape=rectangle, minimum width=0.6cm, minimum height=0.6cm, tikzit category=misc]
\tikzstyle{lmat}=[draw, signal, fill=white, signal to=west, signal from=east, minimum width=1mm, minimum height=6pt, inner sep=0.75pt, outer sep=-0.1em, font={\scriptsize}]
\tikzstyle{rmat}=[lmat, signal to=east, signal from=west]
\tikzstyle{umat}=[lmat, signal to=north, signal from=south, minimum height=1mm, minimum width=6pt]
\tikzstyle{dmat}=[umat, signal to=south, signal from=north]
\tikzstyle{idealket}=[fill=white, draw=black, shape=regular polygon, regular polygon sides=3, regular polygon rotate=0, scale=0.8, inner sep=1pt, font={\footnotesize}, tikzit category=circuit, tikzit shape=rectangle, tikzit fill=green]
\tikzstyle{didealket}=[fill=white, draw=black, shape=regular polygon, regular polygon sides=3, regular polygon rotate=180, scale=0.8, inner sep=1pt, font={\footnotesize}, tikzit category=circuit, tikzit shape=rectangle, tikzit fill=green]

\tikzstyle{hadamard edge}=[-, dashed, dash pattern=on 2pt off 0.5pt, thick, draw={rgb,255: red,68; green,136; blue,255}]
\tikzstyle{box edge}=[-, dashed, dash pattern=on 2pt off 0.5pt, thick, draw={rgb,255: red,203; green,192; blue,225}]
\tikzstyle{brace edge}=[-, tikzit draw=blue, decorate, decoration={brace,amplitude=1mm,raise=-1mm}]
\tikzstyle{diredge}=[->]
\tikzstyle{double edge}=[-, double, shorten <=-1mm, shorten >=-1mm, double distance=2pt]
\tikzstyle{gray edge}=[-, {gray!60!white}]
\tikzstyle{pointer edge}=[->, very thick, gray]
\tikzstyle{boldedge}=[-, line width=1.6pt, shorten <=-0.17mm, shorten >=-0.17mm]
\tikzstyle{bidir edge}=[<->, very thick, draw={rgb,255: red,191; green,191; blue,191}]
\tikzstyle{blue edge}=[-, blue]

\newtheorem{theorem}{Theorem}[section]

\newtheorem{corollary}[theorem]{Corollary}
\newtheorem{lemma}[theorem]{Lemma}
\newtheorem{prop}[theorem]{Proposition}

\theoremstyle{definition}
\newtheorem{definition}[theorem]{Definition}

\newtheorem{remark}[theorem]{Remark}

\newtheorem*{notation}{Notation}

\DeclareMathOperator\tr{tr}

\DeclareMathOperator\Kl{Kl}
\DeclareMathOperator\Span{Span}
\DeclareMathOperator\Csp{Cosp}
\DeclareMathOperator\nil{nil}
\DeclareMathOperator\Spec{Spec}
\DeclareMathOperator\Mat{Mat}

\newcommand{\bSpan}{\mathbf{Sp}}
\newcommand{\bCsp}{\mathbf{Csp}}
\newcommand{\FrSpk}{_G \Span(\Affk)}
\newcommand{\FrSpq}{_{G_q} \Span(\Affqq)}
\newcommand{\FrCsp}{_F \Csp(\CAlgfgk)}
\newcommand{\FrCspq}{_F \Csp(\CAlgfgq)}
\newcommand{\kbar}{\overline{k}}

\newcommand{\Fq}{\mathbb F_q}
\newcommand{\Fqbar}{\overline{\mathbb F}_q}
\newcommand{\Affk}{\mathsf{AffVar}_k}

\newcommand{\Affqq}{\mathsf{AffVar}_{\mathbb F_q}^{(\mathbb F_q)}}
\newcommand{\Polyk}{\mathsf{Poly}_k}

\newcommand{\Polyqqred}{\mathsf{Poly}_{\Fq}^{\mathsf{qred.}}}
\newcommand{\CAlgk}{\mathsf{CAlg}_k}
\newcommand{\CAlgfgk}{\mathsf{CAlg}^{\mathsf{f.g.}}_k}
\newcommand{\CAlgfgredk}{\mathsf{CAlg}^{\mathsf{f.g.},\ \mathsf{red.}}_k}

\newcommand{\CAlgfgq}{\mathsf{CAlg}^{\mathsf{f.g.}}_{\Fq}}

\newcommand{\CAlgfgqredq}{\mathsf{CAlg}^{\mathsf{f.g.},\ \mathsf{qred.}}_{\Fq}}
\newcommand{\GagOne}{\LangSpq^{\mathsf{Fourier}}}

\newcommand{\LangSpan}{\mathbb L_{\Span}}
\newcommand{\LangSpk}{\mathsf{GAG}_k}
\newcommand{\LangSpq}{\mathsf{GAG}_q}
\newcommand{\LangCsp}{\mathsf{GCA}_k}

\newcommand{\semqMat}[1][-]{\interp{#1}_q}
\newcommand{\semZOne}[1][-]{\interp{#1}_{q}^{\mathsf{Fourier}}}

\newcommand{\ZH}{\mathsf{ZH}}

\newcommand{\ZHsem}[1][-]{\interp{#1}_{\mathsf{ZH}}}
\newcommand{\ZHtargetMats}{q^{\frac{*}{2}} \Mat_{\mathbb Z[\omega]}}

\newcommand{\zassoc}{{\tiny(\hyperref[fig:lawvere-theory-calg-k]{\textsc{z-assoc}})}}
\newcommand{\zsymm}{{\tiny(\hyperref[fig:lawvere-theory-calg-k]{\textsc{z-symm}})}}
\newcommand{\zunit}{{\tiny(\hyperref[fig:lawvere-theory-calg-k]{\textsc{z-unit}})}}
\newcommand{\zfusion}{{\tiny(\hyperref[fig:lawvere-theory-calg-k]{\textsc{z-fusion}})}}

\newcommand{\xassoc}{{\tiny(\hyperref[fig:lawvere-theory-calg-k]{\textsc{x-assoc}})}}
\newcommand{\xsymm}{{\tiny(\hyperref[fig:lawvere-theory-calg-k]{\textsc{x-symm}})}}
\newcommand{\xunit}{{\tiny(\hyperref[fig:lawvere-theory-calg-k]{\textsc{x-unit}})}}
\newcommand{\xfusion}{{\tiny(\hyperref[fig:lawvere-theory-calg-k]{\textsc{x-fusion}})}}

\newcommand{\massoc}{{\tiny(\hyperref[fig:lawvere-theory-calg-k]{\textsc{m-assoc}})}}
\newcommand{\msymm}{{\tiny(\hyperref[fig:lawvere-theory-calg-k]{\textsc{m-symm}})}}
\newcommand{\munit}{{\tiny(\hyperref[fig:lawvere-theory-calg-k]{\textsc{m-unit}})}}
\newcommand{\mfusion}{{\tiny(\hyperref[fig:lawvere-theory-calg-k]{\textsc{m-fusion}})}}

\newcommand{\addcp}{{\tiny(\hyperref[fig:lawvere-theory-calg-k]{\textsc{add-cp}})}}
\newcommand{\adddel}{{\tiny(\hyperref[fig:lawvere-theory-calg-k]{\textsc{add-del}})}}
\newcommand{\zerocp}{{\tiny(\hyperref[fig:lawvere-theory-calg-k]{\textsc{zero-cp}})}}
\newcommand{\zerodel}{{\tiny(\hyperref[fig:lawvere-theory-calg-k]{\textsc{zero-del}})}}
\newcommand{\zxbialg}{{\tiny(\hyperref[fig:lawvere-theory-calg-k]{\textsc{zx-bialg}})}}

\newcommand{\multcp}{{\tiny(\hyperref[fig:lawvere-theory-calg-k]{\textsc{mult-cp}})}}
\newcommand{\multdel}{{\tiny(\hyperref[fig:lawvere-theory-calg-k]{\textsc{mult-del}})}}
\newcommand{\onecp}{{\tiny(\hyperref[fig:lawvere-theory-calg-k]{\textsc{one-cp}})}}
\newcommand{\onedel}{{\tiny(\hyperref[fig:lawvere-theory-calg-k]{\textsc{one-del}})}}
\newcommand{\zmbialg}{{\tiny(\hyperref[fig:lawvere-theory-calg-k]{\textsc{zm-bialg}})}}

\newcommand{\xmdistr}{{\tiny(\hyperref[fig:lawvere-theory-calg-k]{\textsc{xm-distr}})}}

\newcommand{\kassoc}{{\tiny(\hyperref[fig:lawvere-theory-calg-k]{\textsc{k-assoc}})}}
\newcommand{\oneInK}{{\tiny(\hyperref[fig:lawvere-theory-calg-k]{\textsc{one-in-k}})}}
\newcommand{\zeroInK}{{\tiny(\hyperref[fig:lawvere-theory-calg-k]{\textsc{zero-in-k}})}}
\newcommand{\kcopy}{{\tiny(\hyperref[fig:lawvere-theory-calg-k]{\textsc{k-copy}})}}
\newcommand{\kdistr}{{\tiny(\hyperref[fig:lawvere-theory-calg-k]{\textsc{k-distr}})}}
\newcommand{\kadd}{{\tiny(\hyperref[fig:lawvere-theory-calg-k]{\textsc{k-add}})}}
\newcommand{\kdel}{{\tiny(\hyperref[fig:lawvere-theory-calg-k]{\textsc{k-del}})}}
\newcommand{\kOnZero}{{\tiny(\hyperref[fig:lawvere-theory-calg-k]{\textsc{k-on-zero}})}}
\newcommand{\km}{{\tiny(\hyperref[fig:lawvere-theory-calg-k]{\textsc{k-m}})}}

\newcommand{\polycp}{{\tiny(\hyperref[fig:scalable-rewrites]{\textsc{poly-copy}})}}
\newcommand{\polydel}{{\tiny(\hyperref[fig:scalable-rewrites]{\textsc{poly-del}})}}
\newcommand{\polyadd}{{\tiny(\hyperref[fig:scalable-rewrites]{\textsc{poly-add}})}}
\newcommand{\polymult}{{\tiny(\hyperref[fig:scalable-rewrites]{\textsc{poly-mult}})}}
\newcommand{\polycomp}{{\tiny(\hyperref[fig:scalable-rewrites]{\textsc{poly-comp}})}}

\newcommand{\kinv}{{\tiny(\hyperref[fig:rules-for-csp]{\textsc{k-inv}})}}
\newcommand{\xbone}{{\tiny(\hyperref[fig:rules-for-csp]{\textsc{x-bone}})}}
\newcommand{\zfrob}{{\tiny(\hyperref[fig:rules-for-csp]{\textsc{z-frob}})}}
\newcommand{\xfrob}{{\tiny(\hyperref[fig:rules-for-csp]{\textsc{x-frob}})}}
\newcommand{\cccup}{{\tiny(\hyperref[fig:rules-for-csp]{\textsc{cup}})}}
\newcommand{\cccap}{{\tiny(\hyperref[fig:rules-for-csp]{\textsc{cap}})}}
\newcommand{\ktrp}{{\tiny(\hyperref[fig:rules-for-csp]{\textsc{k-trp}})}}
\newcommand{\mtrp}{{\tiny(\hyperref[fig:rules-for-csp]{\textsc{k-trp}})}}
\newcommand{\onetrp}{{\tiny(\hyperref[fig:rules-for-csp]{\textsc{one-trp}})}}

\newcommand{\ZOneCopy}{{\tiny(\hyperref[fig:gag-one-rules]{\textsc{Z1-copy}})}}
\newcommand{\ZOneDel}{{\tiny(\hyperref[fig:gag-one-rules]{\textsc{Z1-del}})}}
\newcommand{\ZOneTrace}{{\tiny(\hyperref[fig:gag-one-rules]{\textsc{Z1-trace}})}}

\newenvironment{subsubfigure}[2][]{%
  \begin{subfigure}[#1]{#2}%
    \stepcounter{subsubfigure}%
}{%
    \addtocounter{subfigure}{-1}%
  \end{subfigure}%
}

\newcounter{subsubfigure}

\allowdisplaybreaks

\title[Graphical Algebraic Geometry]{Graphical Algebraic Geometry\\ {\smaller \smaller From Ideals and Varieties to Quantum Calculi}}
\author{Dichuan (David) Gao}
\affiliation{%
  \institution{University of Oxford}
  \city{Oxford}
  \country{UK}}
\author{Razin A. Shaikh}
\affiliation{%
  \institution{University of Oxford}
  \city{Oxford}
  \country{UK}}
\author{Aleks Kissinger}
\affiliation{%
  \institution{University of Oxford}
  \city{Oxford}
  \country{UK}}

\begin{document}

\begin{abstract}
  We introduce \emph{Graphical Algebraic Geometry} (GAG), a family of diagrammatic languages extending the Graphical Linear Algebra programme.
  We construct several languages within this family and prove that they are universal and complete for the corresponding (co)span semantics of commutative algebras and affine varieties.
  This framework provides clear graphical representations of algebraic structures --- such as polynomials, ideals, and varieties --- enabling intuitive yet rigorous diagrammatic reasoning.

  We showcase two practical viewpoints on GAG. First, we show that instances of counting constraint satisfaction problem (\#CSP) are recast as rewrite problems of closed diagrams in GAG. This means that deciding rewritability in GAG is \#P-hard, and GAG can be viewed as a complete and compositional rewrite system for networks of polynomial constraints. Second, we characterize the qudit ZH calculus, a diagrammatic language for quantum computation, as an extension of Graphical Algebraic Geometry.
  This establishes the correspondence that \emph{Graphical Algebraic Geometry is to the ZH calculus what Graphical Linear Algebra is to the ZX calculus}.
  Using this construction, we show that computing amplitudes in qudit ZH requires only a constant number of queries to a GAG oracle.
\end{abstract}

\maketitle

\footnotetext{Correspondence to dichuan.gao@cs.ox.ac.uk}

\section{Introduction}
\label{sec:introduction}

\paragraph{Algebraic geometry in computer science}
Algebraic geometry and commutative algebra have been important influences in several fields in computer science, for example in robotics \cite{ahmadi_geometry_2017}, computer vision \cite{hartley_multiple_2004}, and in cryptography \cite{galbraith_mathematics_2012}. In the converse direction, computational techniques have become influential in algebraic geometry with the advent of computer algebra systems \cite{cox_ideals_2015}.

Of particular interest to us is the use of commutative algebra and geometry in two related areas of research. The first, and perhaps more obvious one, is that of Constraint Satisfaction Problems (CSP) \cite{tsang_foundations_1993}. A CSP instance where constraints are given as polynomials is a problem about an algebraic variety. In particular, in the case where the domain of the CSP is a finite field $\Fq$, the corresponding \#CSP instance \cite{jerrum_counting_2017} is to count $\Fq$-rational points in a variety defined by polynomial constraints --- a classical and central problem in 20th century algebraic geometry \cite{hasse_zur_1936,weil_numbers_1949}. More recently, cryptographic proof systems such as SNARK and STARK have relied heavily on algebraic techniques over finite fields, demonstrating the continuing relevance of this approach \cite{ben-sasson_scalable_2018,groth_size_2016}. 

The second area is quantum information theory. It has long been known that algebraic geometry can be used to design classical codes with good asymptotic behavior \cite{reed_polynomial_1960,berlekamp_algebraic_1968,goppa_codes_1981,tsfasman_modular_1982,hoholdt_decoding_1995,guruswami_improved_1999,vladut_number_1983}, and that via the CSS or stabilizer constructions these classical AG codes lift to quantum error correction (QEC) codes with equally good properties \cite{ashikhmin_asymptotically_2001,matsumoto_algebraic_2002,matsumoto_improvement_2002, galindo_quantum_2015,jin_euclidean_2013,jin_quantum_2014,sarvepalli_nonbinary_2006,shaska_quantum_2008}.
Such codes also played a central role in the groundbreaking construction of a magic-state distillation protocol with constant overhead \cite{wills_constant-overhead_2024}.
Aside from error correction, algebraic geometry has also served as a useful language for foundational areas of research such as quantum entanglement and quantum contextuality \cite{coffman_distributed_1999,miyake_classification_2003,landsberg_ideals_2003,gharahi_classifying_2024,frembs_algebraic_2024}.

\paragraph{Graphical linear algebra and its friends}
To systematically understand the role algebraic geometry can play in quantum information theory, we propose that graphical calculi would be useful. Graphical calculi are compositional diagrammatic languages based in category theory, equipped with formal ``rewrite rules'' with which one can manipulate diagrams rigorously as equational terms. For an in-depth exposition of graphical calculi and rewrite theory, see \cite{bonchi_string_2022-1,bonchi_string_2022-2,bonchi_string_2022-3}. Such calculi are usually designed to represent some family of physical phenomena or computational processes. For instance, the ZX calculus and its related calculi have become prominent as languages for representing quantum computational processes \cite{coeckeInteractingQuantumObservables2008,coecke_picturing_2017,debeaudrapZXCalculusLanguage2020,kissinger_picturing_2024, kissingerReducingTcountZXcalculus2020, poorZXcalculusCompleteFiniteDimensional2025}.

When considering the classical algebraic structure inherent in quantum computing, we often encounter fragments of the quantum calculi that are best characterized by languages in the family called Graphical Linear Algebra (GLA) \cite{bonchi_interacting_2014,bonchi_interacting_2017,zanasi_interacting_2018}. Broadly speaking, GLA consists of diagrammatic languages that have denotational semantics built from classical algebraic structure. On the semantic side, the GLA story begins with a category of classical maps: for example, the linear maps between $\mathbb F_2$ spaces. It then applies the (co)span or the (co)relation construction to that category to obtain a dagger-category. For example, $\mathbb F_2$ linear maps give rise to spans of $\mathbb F_2$ linear maps, or to $\mathbb F_2$ linear subsets. On the syntactic side, it builds a simple diagrammatic language with generators that have clear semantic interpretations, and rewrite rules that characterize denotational equivalence.

Languages in the graphical linear algebra family are excellent at picking out fragments of quantum computation that correspond to some given notion of a ``classical-relational backbone''. For example, the CSS codes resulting from classical $\mathbb F_2$-linear codes correspond to the phase-free fragment of the ZX calculus \cite{kissinger_phase-free_2022}, which is picked out by the language of interacting bialgebras \cite{bonchi_interacting_2014}, the language that kickstarted the GLA literature. Similarly, the language of interacting Hopf algebras \cite{bonchi_interacting_2017,zanasi_interacting_2018} picks out the phase-free fragment of the qudit ZX calculus, which models the linear-relational core of qudit computation. Expanding our algebraic structure from linear to affine algebra, the language of graphical affine algebra \cite{bonchi_graphical_2019} picks out the phase-free ZX with a Pauli X gate, which is the affine-relational core of qubit computation \cite{comfort_distributive_2021}. Finally the language of affine Lagrangian relations picks out the stabilizer fragment of the ZX calculus \cite{comfort_graphical_2022}. 

Of course, GLA has applications far beyond quantum computing. For instance, the language for linear relations over a field of rational functions $\mathbb R(x)$ models signal flow graphs and linear dynamical systems \cite{bonchi_full_2015,baez_categories_2015,fong_categorical_2016}, while the language for affine relations models the fundamental structure of concurrency \cite{bonchi_graphical_2019}. The strength of GLA lies in its ability to model a remarkably varied set of physical phenomena using the same formal structure. 

\paragraph{The need for a graphical nonlinear algebra}
However, existing versions of GLA are not yet sufficient to model the fragment of quantum computation with classical algebraic geometry as its ``backbone''; nor can they reason about generic constraint satisfaction problems built on polynomial constraints. This is because commutative algebra involves a multiplication operator (an AND gate in the $\mathbb F_2$ case), which is essentially nonlinear (and non-affine for that matter). A graphical calculus that models classical algebraic geometry would need to supply such an operator, or equivalently, to supply the Toffoli gate.

In the realm of quantum computing, the ZH calculus \cite{backens_zh_2019,backens_completeness_2023} or equivalently the ZX\& calculus \cite{comfort_distributive_2021,comfort_zx-calculus_2021} supplies such a gate for the $\mathbb F_2$ case, and the Galois-qudit ZH calculus \cite{roy_qudit_2023,gao_qudit_2024} serves as the appropriate version in the nonbinary case. The motivating question of this paper is therefore: can we extend GLA in a principled and simple way, to pick out the nonlinear ``classical backbone'' of the ZH calculus, corresponding to classical algebraic geometry?  If so, can we provide a sound and complete set of rewrite rules for that extension?

\paragraph{The solution: GCA and GAG}
The answer is yes. However, unlike in the linear case, in which algebra and geometry coincide, we must deal with a gap between commutative algebra and algebraic geometry, which is bridged by a family of theorems known as the Nullstellensatze \cite{hilbert_ueber_1890,dani_note_2019,hartshorne_algebraic_1977}. Therefore we construct our extension of GLA in two steps.

First, we develop a language $\LangCsp$, which we call graphical commutative algebra (GCA). Diagrams in $\LangCsp$ are built from gates corresponding to copying, addition, scaling, multiplication, and the units of these operators --- the standard generators featured in the theory of commutative algebras. Diagrams can be reduced to a pseudonormal form that looks like
\begin{equation}
    \label{eq:cospan-form-intro}
    \input{./figures/cospan-form-intro.tikz}
\end{equation}
Denotational semantics are given as structured cospans of commutative algebras. In the pseudonormal form (\ref{eq:cospan-form-intro}), the vertical wires denote a commutative algebra quotiented by the ideal $(g_1, \dots, g_n)$, while $f_1$ and $f_2$ denote two algebra morphisms pointing into it. This captures the nonlinear \textit{algebraic} backbone of quantum processes in ZH. Our first main result is a sound and complete rule set for $\LangCsp$ with respect to this cospan semantic.

Second, we develop languages $\LangSpk$ and $\LangSpq$, which we refer to as Graphical Algebraic Geometry (GAG) over algebraically closed fields and finite fields, respectively. Diagrams in these are built from the same generators as diagrams in $\LangCsp$, but they are subject to further rewrite rules that express the content of the Nullstellensatze. Denotational semantics are given respectively, as structured spans of affine varieties, and of $\Fq$-loci of varieties (or, equivalently, of finite sets). We still have the same pseudonormal form (\ref{eq:cospan-form-intro}), but now the vertical wires with polynomials $g_1, \dots, g_n$ denote an affine variety carved out by polynomial constraints, and $f_1, f_2$ are polynomial maps pointing out of it. This captures the nonlinear \textit{geometric} backbone of quantum processes in ZH. Our second main result is that these languages are sound and complete with respect to their span semantics.

Our languages bear a clear relation to CSP instances with polynomial constraints. As will become clear, a diagram of the form $\begin{tikzpicture}
	\begin{pgfonlayer}{nodelayer}
		\node [style=rmat] (0) at (0, 0) {$g$};
		\node [style=X dot] (1) at (1, 0) {};
		\node [style=none] (2) at (-1, 0) {};
	\end{pgfonlayer}
	\begin{pgfonlayer}{edgelayer}
		\draw (1) to (2.center);
	\end{pgfonlayer}
\end{tikzpicture}
$
just \textit{is} the CSP instance with constraints $g = 0$, while the diagram $\begin{tikzpicture}
	\begin{pgfonlayer}{nodelayer}
		\node [style=rmat] (0) at (0, 0) {$g$};
		\node [style=X dot] (1) at (1, 0) {};
		\node [style=Z dot] (2) at (-1, 0) {};
	\end{pgfonlayer}
	\begin{pgfonlayer}{edgelayer}
		\draw (1) to (2);
	\end{pgfonlayer}
\end{tikzpicture}
$ is the \#CSP problem of counting the number of points satisfying $g$. Our language, then, gives a compositional and complete rewrite system for networks of CSP instances with polynomial constraints. Through this view of GAG, we show that determining rewritability of an arbitrary pair of GAG diagrams is \#P-hard.
In the special case where the field is $\mathbb F_2$, this just reduces to a compositional language for SAT instances. This serves as an algebraic elucidation of two ideas already present in the literature: first, that the ZX\& calculus completely models ``qubit multirelations'' \cite{comfort_zx-calculus_2021}; and second, that the ZH calculus can be used to model \#SAT instances, and to reason about optimizations of \#SAT algorithms \cite{laakkonen_picturing_2023,laakkonen_graphical_2024}.

On the other hand, we justify our claim that \emph{GAG is to ZH what GLA is to ZX} by characterizing the Galois-qudit ZH as an extension of GAG over finite fields. Indeed, we show that once we have access to the ``classical backbone'' carved out by GAG, any process in ZH can be built by consuming just a single basis state in the Fourier basis, up to a finite number of global scalars. 

\paragraph{Outline of the paper}
We begin in Section~\ref{sec:background} by covering the requisite theoretical background on category theory and algebraic geometry. In Section~\ref{section:poly-fragment}, we discuss the Lawvere theory of commutative algebras, which serves as a building block (the ``forward fragment'') of our languages. We then construct the (co)span semantic categories of interest, namely, the structured cospan category of commutative algebras, and the structured span category of affine varieties, in Section~\ref{section:semantic-targets}. We prove our first main result in Section~\ref{sec:cospans-of-calg}, where we construct the language $\LangCsp$ and prove its soundness and completeness for cospans of commutative algebras. The second main result is proven in Section~\ref{sec:structured-spans-of-alg-varieties}, where we construct the languages $\LangSpk$ and $\LangSpq$, and prove their soundness and completeness for spans of affine varieties and of rational loci, respectively. In Section~\ref{sec:GAG-to-CSP}, we show that closed diagrams in GAG are equivalent to instances of $\#\mathsf{CSP}$, and therefore rewriting in GAG is \#P-hard. Finally, in Section~\ref{sec:from-GAG-to-ZH} we show how the Galois-qudit ZH can be built as an extension of GAG.

\paragraph{Related Work} 
There have been several previous works which present diagrammatic languages of a similar flavor, although never from an explicitly algebraic-geometric angle. Lafont has presented a graphical calculus \cite{lafont_towards_2003} of Boolean circuits which is equivalent to our ``polynomial fragment'' (Section~\ref{section:poly-fragment}), but modulo the additional rule that encodes the identity $x^2 = x$ for boolean variables. As will become clear in section~\ref{sec:spans-over-ff}, this makes it a ``forward fragment'' of the language $\LangSpan$ constructed in that section, for the finite field~$\mathbb F_2$. Wilson and Zanasi~\cite{wilson_axiomatic_2023} have presented a diagrammatic language which formalizes differentiation of polynomials. Underlying that language is a notion of polynomial circuits, which is again identical to our polynomial fragment in section~\ref{section:poly-fragment}, defined over arbitrary semirings. Finally, the ZX\&-calculus of Comfort \cite{comfort_zx-calculus_2021} is the same language as the language $\LangSpq$ constructed in section~\ref{sec:spans-over-ff} for the field $\mathbb F_2$. However, the completeness proofs for the ZX\&-calculus are tailored for that specific field, so in comparison, our approach serves to elucidate the algebraic nature of the language even in the case of $\mathbb F_2$.

\section{Background}
\label{sec:background}

\subsection{Diagrammatic Languages} 
\label{sec:gla}

In this section we present the theory of PROPs as it relates to the construction of our diagrammatic languages. For a more in-depth and theoretical discussion of PROPs, we recommend the excellent note by Lack \cite{lack_composing_2004}.

\begin{definition}
    \label{def:PROP} 
    \cite{maclane_categorical_1965} A Product-and-Permutation-Category (PROP) is a strict symmetric monoidal category $\mathcal C$ with $Ob(\mathcal C) \cong \mathbb N$, whose monoidal product on objects corresponds to addition of natural numbers.

    A PROP morphism from $\mathcal C$ to $\mathcal D$ is a strict symmetric monoidal functor such that $1_\mathcal C \mapsto 1_\mathcal D$, or in other words, the objects are fixed. A PROP isomorphism is a PROP morphism that is fully faithful. 
\end{definition}

When we speak of a \textbf{presentation} of a PROP, we are speaking about a particular data structure, which we describe here with terminology adopted from the graphical linear algebra literature \cite{bonchi_interacting_2017}. 

\begin{definition}
    A \textbf{signature} is a set $S = \{f_i: n_i \to m_i\}_{i \in \mathcal I}$ where each element has a name, a number of inputs, and a number of outputs, and is pictured as a node in a string diagram: 
    $$\input{./figures/generic-spider.tikz}$$

    A \textbf{diagrammatic term over a signature} $S$ is an inductive structure constructed as in Figure~\ref{fig:diagram-terms}. 

    \begin{figure*}
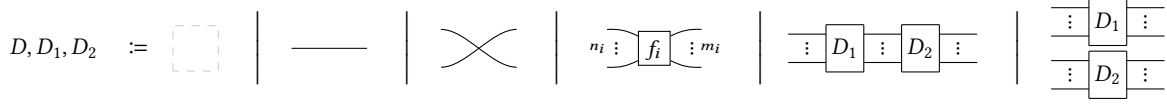

        \caption{Recursive Definition of Diagrammatic Terms}
        \label{fig:diagram-terms}
        \centering 
        $D, D_1, D_2 \quad \coloneq$\quad \input{./figures/empty-diagram.tikz} \quad\Bigg| \quad \input{./figures/id-wire.tikz} \quad\Bigg| \quad \input{./figures/swap.tikz} \quad \Bigg| \quad \input{./figures/generic-spider-simple.tikz} \quad \Bigg|\quad \input{./figures/generic-compose.tikz} \quad \Bigg| \quad \input{./figures/generic-tensor.tikz}
    \end{figure*}
    
    A \textbf{diagram} over $S$ is an equivalence class of diagrammatic terms, where the equivalence relation is the one generated by the coherence axioms of strict symmetric monoidal categories, under the standard interpretation of string diagrams as morphisms in strict symmetric monoidal categories. 
\end{definition}

\begin{definition}
    \label{def:smt} 
        \cite{bonchi_interacting_2017} A symmetric monoidal theory (SMT) is a pair of sets $(S, E)$, where
        \begin{itemize}
            \item The set $S$ is a signature, whose elements are also referred to as the generators or the spiders; 
            \item The set $E = \{D_j \sim D'_j\}_{j \in \mathcal J}$ is a set of pairs of diagrams over $S$, where each pair relates two diagrams with the same numbers of inputs and outputs. These pairs are called the relations or rewrite rules; 
        \end{itemize}
\end{definition}

\begin{definition}
    For any symmetric monoidal theory $(S, E)$, there is a corresponding PROP $\langle S \rangle/E$, where: 
    \begin{itemize}
        \item Morphisms from $n$ to $m$ are equivalence classes of diagrams over $S$ with $n$ input wires and $m$ output wires, where the equivalence relation is the smallest one containing the axioms $E$ while satisfying the laws of transitivity and substitution; 
        \item Composition is plugging together of wires 
        \begin{equation}
        \input{./figures/generic-compose.tikz}
        \end{equation}
        \item Monoidal product is side-by-side juxtaposition 
        \begin{equation}
            \input{./figures/generic-tensor.tikz}
        \end{equation}
    \end{itemize}
\end{definition}

\begin{definition}
    \label{def:presentation} 
    An SMT $(S, E)$ presents a PROP $\mathcal C$ if there is a PROP isomorphism $\langle S \rangle/E \cong \mathcal C$. 
\end{definition}

Because a presentation, in this sense, lends itself to diagrammatic representation so naturally, we will use \textbf{diagrammatic language} for $\mathcal C$ as a synonym for a presentation of $\mathcal C$. 

Often we have a concrete PROP $\mathcal C$ in mind (for example, the PROP $\Mat_R$ of finite matrices over a ring $R$), and we want to find a presentation for it. It is useful, then, to think of this as constructing a ``syntax category'' $\langle S \rangle / E$, and a ``semantics functor'' $\sem: \langle S \rangle / E \to \mathcal C$. If $\sem$ is well-defined as a PROP morphism, we say that $E$ is \textit{sound} for $\mathcal C$. If it is full, we say that $S$ is \textit{universal} for $\mathcal C$. If it is faithful, we say that $E$ is \textit{complete} for $\mathcal C$. If it is all three at the same time, then it is a PROP isomorphism, and so we have a genuine presentation of $\mathcal C$. 

The simplest useful examples of presentations of PROPs are Lawvere theories of well-behaved algebraic theories, which we now define. 

\begin{definition}[Lawvere Theory]
    A Lawvere theory is a PROP $\mathbb L$ whose monoidal product is also its category-theoretical product, admitting a strictly identity-on-objects and product-preserving functor $\aleph_0 \to L$, where $\aleph_0$ is a skeleton of the category of finite sets and functions. 

    Intuitively, a Lawvere theory is a PROP equipped with a cartesian structure (of copying and deleting). 

    A model of a Lawvere theory $\mathbb L$ is a finite product preserving functor $\mathbb L \to \mathsf{Set}$. The models of a fixed Lawvere theory form a category $\mathsf{Mod}(\mathbb L)$. 
\end{definition}

\begin{prop}
    (Standard Fact on Lawvere Theories. See for example \cite{cheng_distributive_2020}) Let $T$ be a monad on $\mathsf{Set}$ that is finitary (i.e. it preserves filtered colimits). Then there is a Lawvere theory $\mathbb L_T$, defined by the equational theory of $T$, such that 
    \begin{equation}
        \mathsf{Mod}(\mathbb L_T) \simeq \mathsf{Alg}(T)
    \end{equation}
    where $\mathsf{Alg}(T)$ denotes the Eilenberg-Moore category of $T$. Moreover, $\mathbb L_T$ is equivalent to the opposite of $\Kl T|_{fin}$, the Kleisli category of $T$ restricted to finite objects. 
\end{prop}

For instance, suppose we want to find a presentation for the category $\mathsf{Mat}_R$ of finite matrices over a ring $R$. Let $T$ be the monad sending a set $X$ to the free $R$-module generated by $X$. Then $\mathsf{Mat}_R$ is equivalent to the opposite of $\Kl T|_{fin}$. The Lawvere theory for finite-dimensional $R$-modules is therefore equivalent to $Mat_R$. But the Lawvere theory is also, in itself, a neatly presented PROP: the generators are the operations in the Lawvere theory 
$$\left\{\input{./figures/copy.tikz} \qquad \begin{tikzpicture}
	\begin{pgfonlayer}{nodelayer}
		\node [style=Z dot] (0) at (0.5, 0) {};
		\node [style=none] (1) at (-0.5, 0) {};
	\end{pgfonlayer}
	\begin{pgfonlayer}{edgelayer}
		\draw (1.center) to (0);
	\end{pgfonlayer}
\end{tikzpicture}
\qquad \input{./figures/add.tikz}\qquad \begin{tikzpicture}
	\begin{pgfonlayer}{nodelayer}
		\node [style=X dot] (0) at (-0.5, 0) {};
		\node [style=none] (3) at (0.5, 0) {};
	\end{pgfonlayer}
	\begin{pgfonlayer}{edgelayer}
		\draw (0) to (3.center);
	\end{pgfonlayer}
\end{tikzpicture}
 \qquad \begin{tikzpicture}
	\begin{pgfonlayer}{nodelayer}
		\node [style=kscalar] (0) at (0, 0) {$r$};
		\node [style=none] (1) at (-1, 0) {};
		\node [style=none] (2) at (1, 0) {};
	\end{pgfonlayer}
	\begin{pgfonlayer}{edgelayer}
		\draw (1.center) to (2.center);
	\end{pgfonlayer}
\end{tikzpicture}
  \quad \forall r \in R \right\}$$
(copy, delete, add, zero, scaling by $r$) and the rewrite rules are the usual axioms of $R$-modules. This yields a diagrammatic presentation of the PROP $\mathsf{Mat}_R$, as we wanted. 

However, in the GLA literature, the target PROP is usually more complicated than the Kleisli category of an algebraic monad. In particular, it usually has a dagger structure. Building presentations of such PROPs requires us to be able to synthesize larger PROPs out of smaller ones. The simplest way to do this is via PROP sum. 

\begin{definition}
    \cite{zanasi_interacting_2018} Let $\mathcal C$ and $\mathcal D$ be PROPs. Their sum $\mathcal C + \mathcal D$ is their coproduct in the category of PROPs. 
\end{definition}

\begin{prop}
    \cite{zanasi_interacting_2018} If $\mathcal C$ has a presentation $(S_1, E_1)$, and $\mathcal D$ has a presentation $(S_2, E_2)$, then $\mathcal C +\mathcal D$ has presentation $(S_1 \sqcup S_2, E_1 \sqcup E_2)$
\end{prop}

In the sum of PROPs, morphisms coming from two different PROPs of course do not interact. To endow interactive behaviour on morphisms of $\mathcal C$ and $\mathcal D$, we need the notion of distributive laws of PROPs. 

\begin{definition}
    A distributive law is a rewrite rule that distributes morphisms from the two different PROPs past each other: 
    $$\overset{f \in \mathcal C}{\longrightarrow}\overset{g \in \mathcal D}{\longrightarrow} \quad \sim \quad \overset{g' \in \mathcal D}{\longrightarrow}\overset{f' \in \mathcal C}{\longrightarrow}$$

    If $E$ is a collection of such rewrites, then $\mathcal C +_E \mathcal D$ denotes the PROP presented by $(S_1 \sqcup S_2, E_1 \sqcup E_2 \sqcup E)$.
\end{definition}
 Note, this is not the traditional way in which distributive laws and their resulting categories are defined, but in the case where $\mathcal C$ and $\mathcal D$ admit presentations (which is the only case we encounter in this paper), this is in fact equivalent to the traditional definition \cite{lack_composing_2004}. 
%  As the more general notion of distributive laws would require us to discuss at length the theory of PROPs, we will not present it here and refer the interested reader to \cite{lack_composing_2004}. 

% The general strategy encountered in graphical linear algebra (especially see \cite{bonchi_interacting_2017}) is to start with a diagrammatic language $\mathbb L$ for the ``forward'' fragment of our target PROP. For example, when $Mat_R$ is our forward fragment, we have the aforementioned Lawvere theory for $R$-modules as $\mathbb L$. On the semantic side, we apply the span or relation construction to $Mat_R$ to obtain a dagger compact semantic category (more on this in section \ref{sec:daggers}). On the syntax side, we look for a set of $\dag$-invariant distributive laws $E$, each of the form 
% $$\overset{f \in \mathbb L}{\longrightarrow} \overset{g \in \mathbb L^{op}}{\longrightarrow} \quad \sim \quad \overset{g' \in \mathbb L^{op}}{\longrightarrow} \overset{f' \in \mathbb L}{\longrightarrow}$$
% such that $\mathbb L +_E \mathbb L^{op}$ is equivalent to our dagger compact semantic category. Since $\mathbb L + \mathbb L^{op}$ is a dagger category, and $E$ is $\dag$-invariant, so the syntax category is guaranteed to be a dagger category. 

\subsection{Algebraic Geometry} 
\label{sec:ag}

In reviewing the basic algebraic geometry tools required for our purposes, we follow standard texts \cite{cox_ideals_2007,hartshorne_algebraic_1977,eisenbud_commutative_1995}. 

We denote by $k$ an arbitrary field, and $\kbar$ its algebraic closure. Likewise we denote by $\Fq$ a finite field of size $q = p^n$ for some prime $p$, and $\Fqbar$ its algebraic closure. The category of finitely generated commutative algebras over $k$ is denoted $\CAlgfgk$.

\subsubsection{Varieties and Polynomial Mappings}

The basic objects we will be concerned with are affine varieties. Intuitively, we think of an affine variety as a place where certain polynomials evaluate to zero. If, say, a set $S$ of polynomials all evaluate to zero at a certain place in space, then it's clear that any linear combination of polynomials in $S$ also evaluate to zero. Thus, we think of affine varieties as being affiliated to \textit{ideals} of polynomials, rather than just to sets of polynomials. 

\begin{definition}
    Let $I \subset k[x_1, \dots, x_n]$ be an ideal. The \textbf{zero-set} of $I$ is: 
    \begin{equation}
        V(I) \coloneq \{\mathbf a \in \kbar^n \mid \forall f \in I, f(\mathbf a) = 0\}
    \end{equation}
    An \textbf{affine variety} over $k$ is a zero-set of any ideal of polynomials. 
\end{definition}

Notice that a variety is a subset in a space over the \textit{algebraic closure} of the base field. This is necessary, because in general the roots of polynomials exist only in the algebraic closure. However, often we need to reason about the restriction of a variety, either to the base field, or to some intermediate field. This is captured in the following notion: 

\begin{definition}
    For an intermediate field $k \subset k' \subset \kbar$, the \textbf{$k'$-rational locus} of a variety $V \subset \kbar^n$ is $V^{(k')} \coloneq V \cap (k')^{n}$; i.e. the set of points in $V$ whose coordinates lie entirely in $k'$. When we say ``a rational locus over $k$'' without naming an intermediate field, we mean the $k$-rational locus of a variety over $k$.
\end{definition}

\begin{definition}
    Let $S$ be any subset of $\kbar^n$. The set of $k$-polynomials vanishing on $S$ is 
    \begin{equation*}
        I_k(S) \coloneq \{f \in k[x_1, \dots, x_n] \mid \forall \mathbf a \in S, f(\mathbf a) = 0\}
    \end{equation*}
    It is easy to verify that this set of polynomials forms a radical ideal. 
\end{definition}

\begin{definition}
    \label{def:regular-map}
    Let $V \subset \kbar^n$ and $W \subset \kbar^m$ be affine varieties over $k$. A function $\phi: V \to W$ is a \textbf{$k$-polynomial mapping} if there exists a polynomial tuple $f = (f_1, \dots, f_m) \in (k[x_1, \dots, x_n])^m$ such that 
    $$\phi(a_1, \dots, a_n) = f(a_1, \dots, a_n)$$
    for all $(a_1, \dots, a_n) \in V$. An \textbf{isomorphism of varieties} is a polynomial mapping $\sigma$ which is invertible as a function, such that the inverse function $\sigma^{-1}$ is itself a polynomial mapping. 

    The polynomial mappings from $V$ to $k$ form a commutative algebra over $k$, called the \textbf{coordinate ring} of $V$, and denoted $k[V]$.
\end{definition}

\begin{prop}
    \label{prop:sheaf-ring} 
    \cite{cox_ideals_2007} Let $V$ be an affine variety over $k$. Then 
    \begin{equation*}
        k[V] \cong k[x_1, \dots, x_n]/I_k(V). 
    \end{equation*}
\end{prop}

\begin{definition}
    \label{def:aff-k} (The Categories of Affine Varieties and Rational Loci)
    \cite{cox_ideals_2007} For each fixed field $k$, the polynomial mappings between affine varieties over $k$ form a category $\Affk$. For each intermediate field $k \subset k' \subset \kbar$, there is a full subcategory $\Affk^{(k')} \subset \Affk$ consisting of affine varieties that are entirely $k'$-rational.
\end{definition}

\begin{prop}
    There is a fully faithful contravariant functor $k[-]: \Affk^{op} \to \CAlgfgk$ sending an affine variety $V$ to its coordinate ring $k[V]$.
\end{prop}

% \begin{proof}
%     Cox et.al. Section 5.4 Proposition 8 \cite{cox_ideals_2007}. 
% \end{proof}

\subsubsection{The Nullstellensatze}

The coordinate ring functor $k[-]$ is a way to associate an algebraic object to each geometric shape (the variety). This association is not an equivalence of categories, but very close to being one. The connection in the opposite direction, as we will now see, is provided by Hilbert's Nullstellensatz for algebraically closed fields, and its alternative formulation for finite fields. The Nullstellensatze can therefore be seen as bridges connecting the algebraic world to the geometric world, and it's not surprising that we will use them extensively in Section~\ref{sec:structured-spans-of-alg-varieties} as we move from graphical commutative algebra to graphical algebraic geometry. 

\begin{restatable}[Hilbert's Nullstellensatz]{prop}{hilbertnullstellensatz}
    \label{prop:hilbert-nullstellensatz}
    \cite{hartshorne_algebraic_1977} Let $k$ be an algebraically closed field, and $J \subset k[x_1, \dots, x_n]$ an ideal. Then  
    \begin{equation}
        I_k(V(J)) = \sqrt{J} 
    \end{equation}
\end{restatable}

% \begin{proof}
%     Hartshorne \cite{hartshorne_algebraic_1977}, Chapter 1, Proposition 1.2. 
% \end{proof}

We want to see this as an equivalence between a category of algebraic objects and one of geometric objects. So let us now translate this into the language of category theory as follows: 

\begin{definition}
    A commutative algebra $A$ over $k$ is \textbf{reduced} if it has no nonzero nilpotents. The full subcategory of $\CAlgfgk$ consisting of reduced commutative algebras is denoted $\CAlgfgredk$. 
    %It is called \textbf{geometrically reduced} if it is reduced and $A \otimes \kbar$ is also reduced. 
\end{definition}

\begin{restatable}{prop}{redisleftadjoint}
    \label{prop:red-is-left-adjoint}
    \cite{eisenbud_commutative_1995} The functor $(-)_{red}: \CAlgfgk \to \CAlgfgredk$ sending $A \mapsto A / \nil A$ is the left-adjoint of the inclusion $\CAlgfgredk \hookrightarrow \CAlgfgk$. 
\end{restatable}

\begin{prop}
    \label{prop:alg-closed-anti-equiv}
    \cite{eisenbud_commutative_1995} If $k$ is an algebraically closed field, then the coordinate ring functor $k[-]$ gives an anti-equivalence of categories 
    \begin{equation}
        k[-]: \Affk^{op} \simeq \CAlgfgredk.
    \end{equation}
    Denote its inverse functor (up to isomorphism) as 
    \begin{equation}
        \Spec: \CAlgfgredk \to \Affk^{op}
    \end{equation}
\end{prop}

% \begin{proof}
%     See Eisenbud Ch. 1 Cor. 1.8 \cite{eisenbud_commutative_1995}. 
% \end{proof}

We now turn our attention to the case of finite fields. While affine varieties over algebraically closed fields are ``rare'' among the subsets of $k^n$, \textit{any subset} of $\Fq^n$ is an affine variety over $\Fq$:

\begin{prop}
    \label{prop:affqq-is-finset} 
    The category $\Affqq$ of rational loci over $\Fq$ is equivalent to the category $\mathsf{FinSet}$ of finite sets and relations. 
\end{prop}
\begin{proof}
    It is well known \cite{glynn_classification_1995,lidl_finite_1997} that any subset of $\Fq^n$ is an affine variety, and any function between them is polynomial. 
\end{proof}

However, our results are most easily read if we think in terms of polynomial constraints and morphisms, rather than subsets and functions. We retain the notation $\Affqq$ as a reminder of this ``attitude''. 

\begin{prop}[Finite Field Nullstellensatz]
    \label{prop:ff-nullstellensatz} 
    Let $J \subset \Fq[x_1, \dots, x_n]$ be any ideal of polynomials over a finite field. Then 
    \begin{equation}
        I_{\Fq} (V^{(\Fq)}(J)) = J + \Gamma_q^n
    \end{equation}
    where 
    $\Gamma_q^n \coloneq (x_1^q - x_1, \dots, x_n^q - x_n)$. \cite{dani_note_2019} 
\end{prop}

As with Hilbert's Nullstellensatz, we now translate the Finite Field Nullstellensatz into a statement about categorical equivalence: 

\begin{definition}
    For a finite field $\Fq$, a commutative algebra $A$ over $\Fq$ is called \textbf{$q$-reduced}\footnote{This is our terminology. In finite field error correction literature this is just called ``reduced'' \cite{dani_note_2019}, but we need to avoid confusion with the aforementioned more standard notion of reduction from algebra.} if for every $f \in A$, $f^q = f$. The $q$-reduced algebras form a full subcategory $\CAlgfgqredq\subset \CAlgfgq$. 
\end{definition}

\begin{definition}
    \label{def:q-radical}
    For a finite field $\Fq$ and a finitely generated commutative algebra $A$ over it, define its $q$-radical as the ideal generated as follows: 
    \begin{equation}
        \Gamma_q(A) \coloneq (f^q - f : f \in A)
    \end{equation}
\end{definition}

\begin{remark}
    \label{rem:basis-of-qradical}
    If $\{x_1, ..., x_n\}$ generate $A$ over $\Fq$ then
    \begin{equation}
        \Gamma_q(A) = \Gamma_q^n = (x_1^q - x_1, \dots, x_n^q - x_n). 
    \end{equation}
\end{remark}

\begin{prop}
    \label{prop:qred-is-left-adjoint}
    The functor $(-)_{qred}: \CAlgfgq \to \CAlgfgqredq$ mapping $A \mapsto A/\Gamma_q(A)$ is left-adjoint to the inclusion. 
\end{prop}

\begin{prop}
    \label{prop:finite-field-anti-equiv}
    For a finite field $\Fq$, the coordinate ring functor $\Fq[-]$ restricted to the $\Fq$-rational loci defines an equivalence of categories 
    \begin{equation}
        \Fq[-]: {\Affqq}^{op} \simeq \CAlgfgqredq
    \end{equation}
    Denote its inverse functor as $\Spec$. 
\end{prop}

\begin{proof}
    Since $\Affqq \simeq \mathsf{FinSet}$, this is really just a nonbinary version of the well known Stone duality between finite boolean algebras and finite sets \cite{johnstone_stone_1982}. The $q$-reduced commutative algebras over $\Fq$ correspond one-to-one with $\Fq$-valued functions out of finite sets, and $\Fq[-]$ sends a set to its algebra of $\Fq$-valued functions \cite{lidl_finite_1997}. For a more algebraic perspective see also \cite{grothendieck_ements_1960}. 
\end{proof}

\begin{notation}
    Denote the $q$-reduced free commutative algebras 
    \begin{equation}
        \mathcal R_q^n \coloneq S_{\Fq}(n)/\Gamma_q^n, 
    \end{equation}
    and denote the full subcategory of $\CAlgfgq$ consisting of such algebras as $\Polyqqred$, and $F_q$ its inclusion functor. 
\end{notation} 

% \begin{remark}
%     These $\mathcal R_q^n$ serve as the $q$-reduced versions of free commutative algebras, in that the following commutes: 
%     % https://q.uiver.app/#q=WzAsNSxbMCwwLCJcXHtcXEZxXm46blxcaW4gXFxtYXRoYmIgTlxcfSJdLFswLDEsIlxcQWZmcXEiXSxbMSwxLCJcXENBbGdmZ3FyZWRxIl0sWzEsMCwiXFx7XFxtYXRoY2FsIFJfcV5uOm4gXFxpbiBcXG1hdGhiYiBOXFx9Il0sWzIsMCwiXFxQb2x5cXFyZWQiXSxbMCwzLCJcXEZxWy1dIl0sWzMsMiwiRl9xIiwwLHsic3R5bGUiOnsidGFpbCI6eyJuYW1lIjoiaG9vayIsInNpZGUiOiJ0b3AifX19XSxbMCwxLCJHX3EiLDIseyJzdHlsZSI6eyJ0YWlsIjp7Im5hbWUiOiJob29rIiwic2lkZSI6InRvcCJ9fX1dLFsxLDIsIlxcRnFbLV0iLDJdLFszLDQsIiIsMix7InN0eWxlIjp7ImhlYWQiOnsibmFtZSI6Im5vbmUifX19XV0=
%     \[\begin{tikzcd}[ampersand replacement=\&]
%         {\{\Fq^n:n\in \mathbb N\}} \& {\{\mathcal R_q^n:n \in \mathbb N\}} \& \Polyqqred \\
%         \Affqq \& \CAlgfgqredq
%         \arrow["{\Fq[-]}", from=1-1, to=1-2]
%         \arrow["{G_q}"', hook, from=1-1, to=2-1]
%         \arrow[equals, from=1-2, to=1-3]
%         \arrow["{F_q}", hook, from=1-2, to=2-2]
%         \arrow["{\Fq[-]}"', from=2-1, to=2-2]
%     \end{tikzcd}\]
% \end{remark}

\section{The Polynomial Fragment} 
\label{section:poly-fragment} 

Consider the free-forgetful adjunction 
% https://q.uiver.app/#q=WzAsMixbMCwwLCJcXG1hdGhzZntTZXR9Il0sWzEsMCwiXFxDQWxnayJdLFswLDEsIlNfayIsMCx7ImN1cnZlIjotMn1dLFsxLDAsIlUiLDAseyJjdXJ2ZSI6LTJ9XSxbMiwzLCIiLDAseyJsZXZlbCI6MSwic3R5bGUiOnsibmFtZSI6ImFkanVuY3Rpb24ifX1dXQ==
\[\begin{tikzcd}
	{\mathsf{Set}} & \CAlgk
	\arrow[""{name=0, anchor=center, inner sep=0}, "{S_k}", curve={height=-12pt}, from=1-1, to=1-2]
	\arrow[""{name=1, anchor=center, inner sep=0}, "U", curve={height=-12pt}, from=1-2, to=1-1]
	\arrow["\dashv"{anchor=center, rotate=-90}, draw=none, from=0, to=1]
\end{tikzcd}\]
On finite sets $n$, $S_k(n)$ is just the polynomial algebra $k[x_1, \dots, x_n]$. Then let $T_k = U \circ S_k$ be the associated monad. 

\begin{definition}
    \label{def:polyk} 
    Let $\Polyk$ be the full subcategory of $\CAlgfgk$ given by the image of $S_k$ restricted to finite sets. In other words, let $\Polyk$ be the Kleisli category of $T_k$ restricted to finite sets, corresponding to the polynomial algebras $k[x_1, \dots, x_n]$. 
\end{definition}

\begin{restatable}[Sharps and Flats]{notation}{sharpsAndFlats}
    It is worth being careful about the equivalence between $\Polyk$ and $\Kl T_k |_{fin}$, the Kleisli category of $T_k$ restricted to finite objects. If $f \in (k[x_1, \dots, x_n])^m$ is a tuple of polynomials, denote by $f^\#$ its unique extension to an algebra morphism $S_k(m) \to S_k(n)$. This acts on a polynomial $g(y_1, \dots, y_m)$ by substituting $y_j$ with $f_j(x_1, \dots, x_n)$. Notice the contravariance implicit in this: algebra morphisms point and compose in the opposite direction than polynomials. Conversely, if $\varphi: S_k(m) \to S_k(n)$ is an algebra morphism, denote by $\varphi^\flat: m \to S_k(n)$ the corresponding tuple of polynomials defined by $\varphi^\flat_j = \varphi(x_j)$.  
\end{restatable}

Recall from Section~\ref{sec:gla} that for any finitary monad $T$, there is an equivalence $\Kl T|_{fin} \simeq \mathbb L_T^{op}$, where $\mathbb L_T$ denotes the Lawvere theory corresponding to $T$. In our case, the Lawvere theory is that of commutative algebras over $k$, which we denote $\mathbb L_{CAlg_k}$. Then: 

\begin{prop}
    \label{prop:forward-fragment-universal-complete}
    The generators and rewrite rules depicted in Figure~\ref{fig:lawvere-theory-calg-k} form a presentation of the PROP $\mathbb L_{CAlg_k} \simeq \Polyk^{op}$. 
\end{prop}

\begin{figure*}
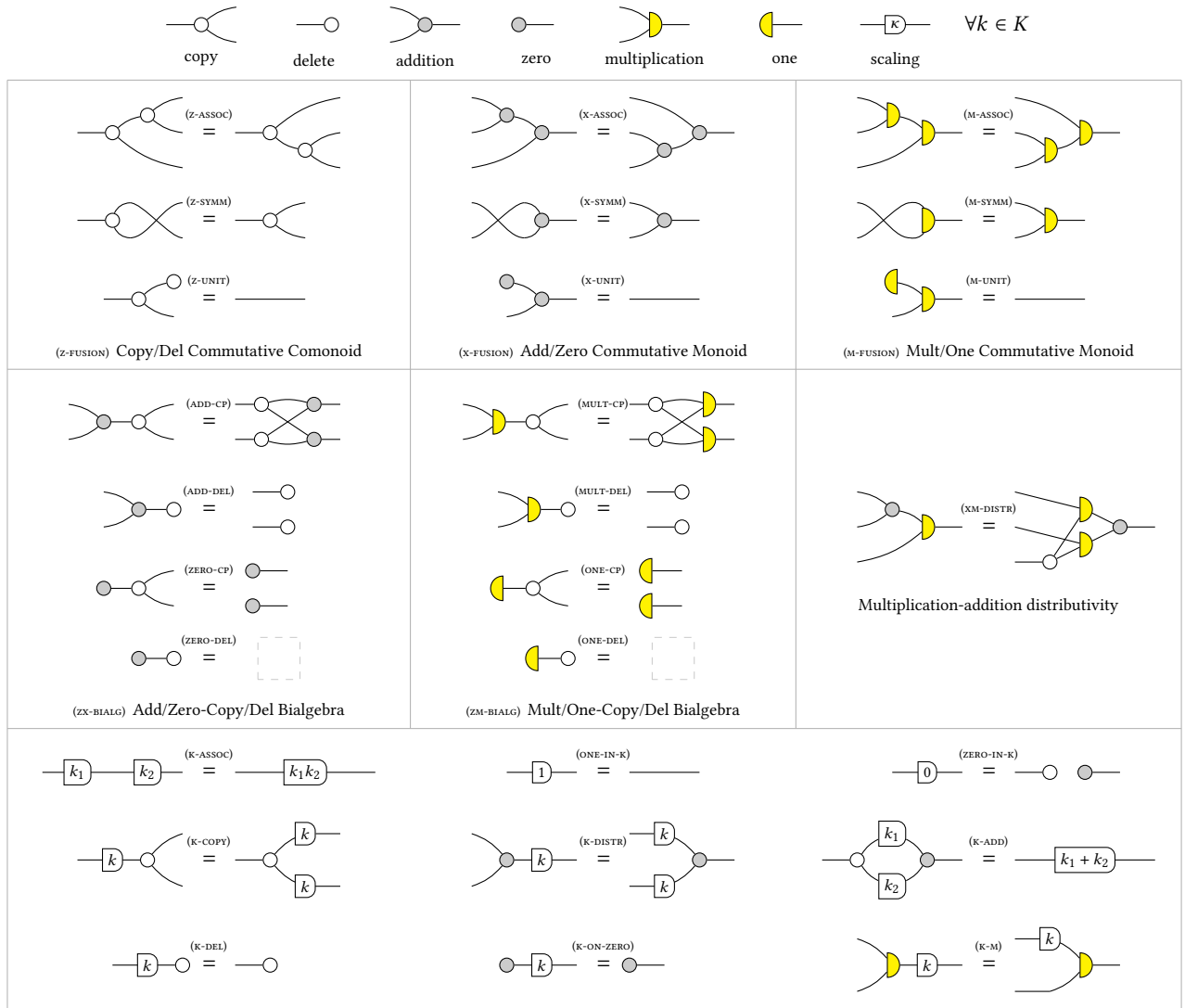

\caption{The Theory $\mathbb L_{CAlg_K}$ of Commutative Algebras over $k$} 
\label{fig:lawvere-theory-calg-k}
    % generators 
    \input{./figures/copy-text.tikz} \qquad \input{./figures/delete-text.tikz}\qquad \input{./figures/add-text.tikz}\qquad \input{./figures/zero-text.tikz}\qquad \input{./figures/mult-text.tikz} \qquad\input{./figures/one-text.tikz} \qquad \input{./figures/scalar-text.tikz}  \quad $\forall k \in K$
    % \rule{\textwidth}{0.4pt}
    % rewrites 
    % ---------------------------------------------
    \input{./figures/L-CAlg-Rules.tikz}
\end{figure*}

\begin{proof}
    This is the standard presentation of the Lawvere theory of commutative algebras over $k$, with the following interpretation for the diagrams: the white dots 
    $\input{./figures/copy.tikz}$ and $\begin{tikzpicture}
	\begin{pgfonlayer}{nodelayer}
		\node [style=Z dot] (0) at (0.5, 0) {};
		\node [style=none] (1) at (-0.5, 0) {};
	\end{pgfonlayer}
	\begin{pgfonlayer}{edgelayer}
		\draw (1.center) to (0);
	\end{pgfonlayer}
\end{tikzpicture}
$ 
    are read as copying and its counit deleting, the grey dots 
    $\input{./figures/add.tikz}$ and $\begin{tikzpicture}
	\begin{pgfonlayer}{nodelayer}
		\node [style=X dot] (0) at (-0.5, 0) {};
		\node [style=none] (3) at (0.5, 0) {};
	\end{pgfonlayer}
	\begin{pgfonlayer}{edgelayer}
		\draw (0) to (3.center);
	\end{pgfonlayer}
\end{tikzpicture}
$ 
    are read as adding and its unit zero, the yellow half-moons 
    $\input{./figures/mult.tikz}$ and $\begin{tikzpicture}
	\begin{pgfonlayer}{nodelayer}
		\node [style=multop] (0) at (-0.5, 0) {};
		\node [style=none] (3) at (0.5, 0) {};
	\end{pgfonlayer}
	\begin{pgfonlayer}{edgelayer}
		\draw (0) to (3.center);
	\end{pgfonlayer}
\end{tikzpicture}
$
    are read as multiplication and its unit one, and the ``squeeze'' gadgets $\begin{tikzpicture}
	\begin{pgfonlayer}{nodelayer}
		\node [style=kscalar] (1) at (0, 0) {$\kappa$};
		\node [style=none] (2) at (-1, 0) {};
		\node [style=none] (3) at (1, 0) {};
	\end{pgfonlayer}
	\begin{pgfonlayer}{edgelayer}
		\draw (2.center) to (3.center);
	\end{pgfonlayer}
\end{tikzpicture}
$ labelled by elements of $k$  
    are read as the scaling operation by the field $k$ on a $k$-commutative algebra. The rewrite rules of the Lawvere theory guarantee that each (finite variable) polynomial with coefficients in $k$ can be expressed as a diagram in this language uniquely up to rewrites. 
\end{proof} 

\begin{restatable}[Scalable Notation]{notation}{scalable}
    We can now unambiguously introduce scalable notation for diagrams in $\mathbb L_{CAlg_k}$. For $n \in \mathbb N$ we define $n$-labelled wire and gates, as depicted in Figure~\ref{fig:scalable-notation}. 
    
    For every polynomial tuple $f \in (k[x_1, \dots, x_n])^m$, we abbreviate the diagram (unique up to rewrites) that corresponds to $f^\#$ in $\mathbb L_{CAlg_k}$ as $\begin{tikzpicture}
	\begin{pgfonlayer}{nodelayer}
		\node [style=rmat] (0) at (0, 0) {$f$};
		\node [style=none] (1) at (-1, 0) {};
		\node [style=none] (2) at (1, 0) {};
		\node [style=none] (3) at (-1.5, 0) {$n$};
		\node [style=none] (4) at (1.5, 0) {$m$};
	\end{pgfonlayer}
	\begin{pgfonlayer}{edgelayer}
		\draw (1.center) to (2.center);
	\end{pgfonlayer}
\end{tikzpicture}
$. 
    
    For example, for the polynomial tuple $f = (x_1 + x_2, x_1 x_3 + x_1^2) \in (k[x_1, x_2,x_3])^2$, we have 
    \begin{equation}
        \begin{tikzpicture}
	\begin{pgfonlayer}{nodelayer}
		\node [style=rmat] (0) at (0, 0) {$f$};
		\node [style=none] (1) at (-1, 0) {};
		\node [style=none] (2) at (1, 0) {};
		\node [style=none] (3) at (-1, 0.5) {$3$};
		\node [style=none] (4) at (1, 0.5) {$2$};
	\end{pgfonlayer}
	\begin{pgfonlayer}{edgelayer}
		\draw (1.center) to (2.center);
	\end{pgfonlayer}
\end{tikzpicture}
 = \input{./figures/example-poly-constr.tikz}
    \end{equation}
\end{restatable}

\begin{figure*}
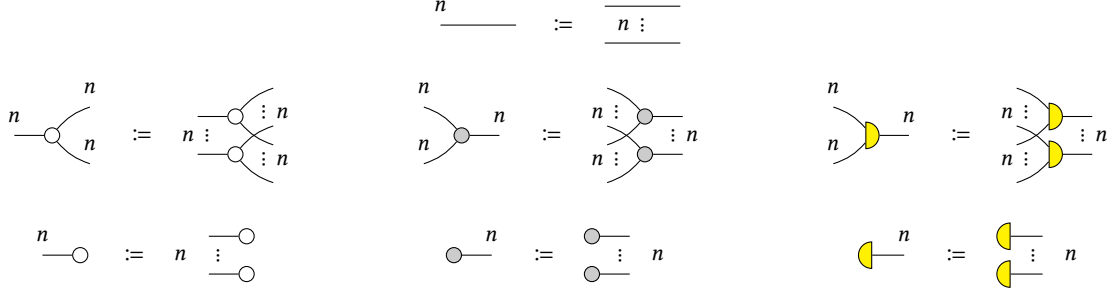

    \caption{Scalable Notation for $\mathbb L_{CAlg_k}$} 
    \label{fig:scalable-notation} 
    \centering 
    \begin{subfigure}{0.3\textwidth}
        \centering 
        \input{./figures/scalable-wire.tikz} \quad $\coloneq$ \quad \input{./figures/scalable-wire-constr.tikz}
        % \caption{Scalable Wires} 
        % \label{fig:scalable-wires}
    \end{subfigure}\\[1.5em]

    \begin{subfigure}{0.3\textwidth}
        \centering 
        \input{./figures/scalable-copy.tikz} \quad $\coloneq$ \quad \input{./figures/scalable-copy-constr.tikz} 
        % \caption{Scalable Copy} 
        % \label{fig:scalable-copy}
    \end{subfigure}
    \begin{subfigure}{0.3\textwidth}
        \centering 
        \input{./figures/scalable-add.tikz} \quad $\coloneq$ \quad \input{./figures/scalable-add-constr.tikz} 
        % \caption{Scalable Addition} 
        % \label{fig:scalable-add}
    \end{subfigure}
    \begin{subfigure}{0.3\textwidth}
        \centering 
        \input{./figures/scalable-mult.tikz} \quad $\coloneq$ \quad \input{./figures/scalable-mult-constr.tikz} 
        % \caption{Scalable Multiplication}
        % \label{fig:scalable-mult}
    \end{subfigure}\\[1.5em]

    \begin{subfigure}{0.3\textwidth}
        \centering 
        \input{./figures/scalable-delete.tikz} \quad $\coloneq$ \quad \input{./figures/scalable-delete-constr.tikz} 
        % \caption{Scalable Delete}
        \label{fig:scalable-delete}
    \end{subfigure}
    \begin{subfigure}{0.3\textwidth}
        \centering 
        \input{./figures/scalable-zero.tikz} \quad $\coloneq$ \quad \input{./figures/scalable-zero-constr.tikz} 
        % \caption{Scalable Zero}
        \label{fig:scalable-zero}
    \end{subfigure}
    \begin{subfigure}{0.3\textwidth}
        \centering 
        \input{./figures/scalable-one.tikz} \quad $\coloneq$ \quad \input{./figures/scalable-one-constr.tikz} 
        % \caption{Scalable One}
        % \label{fig:scalable-one}
    \end{subfigure}
\end{figure*}

Notice that the contravariance between morphism composition and polynomial composition has been inherited by our diagrammatic language. Diagrammatically, everything is drawn in the direction of the Lawvere theory $\mathbb L_{CAlg_k}$. So $\begin{tikzpicture}
	\begin{pgfonlayer}{nodelayer}
		\node [style=rmat] (0) at (0, 0) {$f$};
		\node [style=none] (1) at (-1, 0) {};
		\node [style=none] (2) at (1, 0) {};
		\node [style=none] (3) at (-1.5, 0) {$n$};
		\node [style=none] (4) at (1.5, 0) {$m$};
	\end{pgfonlayer}
	\begin{pgfonlayer}{edgelayer}
		\draw (1.center) to (2.center);
	\end{pgfonlayer}
\end{tikzpicture}
$ corresponds to a tuple of polynomials $(k[x_1, \dots, x_n])^m$, which corresponds to a morphism $f^\# \in \Polyk(m, n)$. The necessity for this will become clear in Section~\ref{sec:structured-spans-of-alg-varieties} when we construct the language for spans of affine varieties.

\begin{restatable}{prop}{scalableRewrites}
    From the rewrite rules of $\mathbb L_{CAlg_k}$ it is possible to derive the rewrite rules about polynomials in the scalable notation, as depicted in Figure~\ref{fig:scalable-rewrites}. 
\end{restatable}

\begin{proof}
    See Appendix~\ref{append:scalable-rules}. 
\end{proof}

\begin{figure*}
    \caption{Rules for Scalable Notation} 
    \label{fig:scalable-rewrites} 
    \centering 
    \input{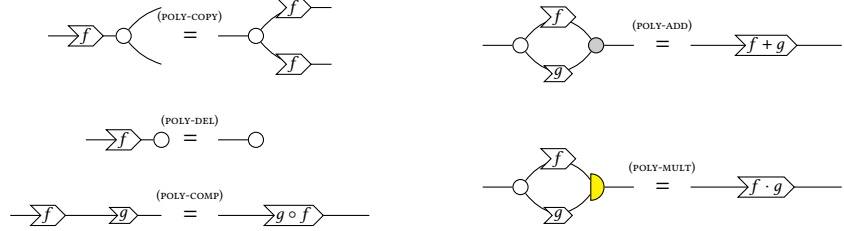}
    % \begin{subfigure}{0.4\textwidth}
    %     \centering
    %     \tikzfig{p-copy-lhs} = \tikzfig{p-copy-rhs}
    %     \caption{Polynomial copies} 
    %     \label{eq:poly-copy} 
    % \end{subfigure}
    % \begin{subfigure}{0.4\textwidth}
    %     \centering
    %     \tikzfig{p-del-lhs} = \tikzfig{delete}
    %     \caption{Polynomial terminates} 
    %     \label{eq:poly-del} 
    % \end{subfigure}
    % \begin{subfigure}{0.4\textwidth}
    %     \centering
    %     \tikzfig{poly-add-lhs} = \tikzfig{poly-add-rhs}
    %     \caption{Polynomial addition} 
    %     \label{eq:poly-add} 
    % \end{subfigure}
    % \begin{subfigure}{0.4\textwidth}
    %     \centering
    %     \tikzfig{poly-mult-lhs} = \tikzfig{poly-mult-rhs}
    %     \caption{Polynomial multiplication} 
    %     \label{eq:poly-mult} 
    % \end{subfigure}
    % \begin{subfigure}{0.4\textwidth}
    %     \centering
    %     \tikzfig{poly-comp-lhs} = \tikzfig{poly-comp-rhs}
    %     \caption{Polynomial composition} 
    %     \label{eq:poly-comp} 
    % \end{subfigure}
\end{figure*}
\section{(Co)span Semantics}
\label{section:semantic-targets}

In this section we describe the target semantic categories of graphical commutative algebra and graphical algebraic geometry. We begin by recalling the definition of structured (co)spans, which is of particular importance here: 
\begin{definition}
    \label{def:structured-span}
    \cite{baez_structured_2020} Let $F: \mathcal A \to \mathcal X$ be a functor where $\mathcal X$ is a category with pullbacks. Define the bicategory of \textbf{structured spans} $_F \bSpan (\mathcal X)$: its objects are those of $\mathcal A$; its 1-cells are pairs of $\mathcal X$ morphisms of the form $FA \leftarrow S \rightarrow FB$ with $A, B \in \mathcal A$, which compose via pullback 
    % https://q.uiver.app/#q=WzAsNixbMSwxLCJTIl0sWzAsMiwiQSJdLFsyLDIsIkIiXSxbMywxLCJUIl0sWzQsMiwiQyJdLFsyLDAsIlNcXHRpbWVzX0JUIl0sWzAsMV0sWzAsMl0sWzMsMl0sWzMsNF0sWzUsMF0sWzUsM10sWzUsMiwiIiwyLHsic3R5bGUiOnsibmFtZSI6ImNvcm5lciJ9fV1d
    
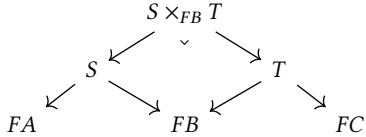
\[\begin{tikzcd}[row sep=small, column sep=small]
        && {S\times_{FB} T} \\
        & S && T \\
        FA && FB && FC
        \arrow[from=1-3, to=2-2]
        \arrow[from=1-3, to=2-4]
        \arrow["\lrcorner"{anchor=center, pos=0.125, rotate=-45}, draw=none, from=1-3, to=3-3]
        \arrow[from=2-2, to=3-1]
        \arrow[from=2-2, to=3-3]
        \arrow[from=2-4, to=3-3]
        \arrow[from=2-4, to=3-5]
    \end{tikzcd}\]
    and its 2-cells are commuting diagrams of the form 
    % https://q.uiver.app/#q=WzAsNCxbMSwwLCJTIl0sWzAsMSwiQSJdLFsyLDEsIkIiXSxbMSwyLCJTJyJdLFswLDFdLFswLDJdLFszLDFdLFszLDJdLFswLDNdXQ==
    
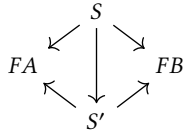
\[\begin{tikzcd}[row sep=small, column sep=small]
        & S \\
        FA && FB \\
        & {S'}
        \arrow[from=1-2, to=2-1]
        \arrow[from=1-2, to=2-3]
        \arrow[from=1-2, to=3-2]
        \arrow[from=3-2, to=2-1]
        \arrow[from=3-2, to=2-3]
    \end{tikzcd}\]
    Quotienting the categories of 1-cells of $_F \bSpan (\mathcal X)$ by the 2-isomorphisms, we obtain a 1-category $_F \Span(\mathcal X)$, whose objects are those of $\mathcal A$, and morphisms are isomorphism classes of spans $FA \leftarrow S \rightarrow FB$. 

    For a functor $F: \mathcal A \to \mathcal X$ where $\mathcal X$ has pushouts, the bicategory of \textbf{structured cospans} $_F \bCsp(\mathcal X)$ and the 1-category $_F \Csp(\mathcal X)$ are defined dually to the above, with 1-cells of the form $FA \rightarrow S \leftarrow FB$ and composition given by pushout instead of pullback. 
\end{definition}

The following is an immediate well-known consequence of the dual nature of the definitions of structured cospans and spans: 
\begin{prop}
    \label{prop:duality-of-co-span}
    If $F: \mathcal A \to \mathcal X$ is a functor and $\mathcal X$ has pullbacks, then $\mathcal X^{op}$ has pushouts, and there is an equivalence 
    \begin{equation}
        _F \Span(\mathcal X) \simeq \,\, _{F^{op}}\Csp(\mathcal X^{op})
    \end{equation}
    where $F^{op}: \mathcal A^{op} \to \mathcal X^{op}$ is the dual functor of $F$. 
\end{prop}

\begin{definition}
    In the structured span construct, the category $\mathcal A$ embeds into $_F \Span(\mathcal X)$ both covariantly by the graph embedding:
    $$G_{\Span}: \quad \left(A \overset{f}{\rightarrow} B\right) \quad \mapsto \quad \left(FA = FA \overset{Ff}{\rightarrow} FB\right)$$
    and contravariantly by the cograph embedding:
    $$G_{\Span}^{op}: \quad \left(A \overset{f}{\rightarrow} B\right) \quad \mapsto \quad \left(FB \overset{Ff}{\leftarrow} FA = FA\right)$$
    Similarly in the structured cospan construct, the category $\mathcal A$ embeds into $_F \Csp(\mathcal X)$ covariantly by the graph embedding:
    $$G_{\Csp}: \quad \left(A \overset{f}{\rightarrow} B\right) \quad \mapsto \quad \left(FA \overset{Ff}{\rightarrow} FB =  FB\right)$$
    and contravariantly by the cograph embedding:
    $$G_{\Csp}^{op}: \quad \left(A \overset{f}{\rightarrow} B\right) \quad \mapsto \quad \left(FB = FB \overset{Ff}{\leftarrow}  FA\right)$$
\end{definition}

The following are not new results, but are recalled here as we shall need it immediately. 
\begin{prop}
    \label{prop:calg-pushout}
    For any field $k$, the category $\CAlgfgk$ of finitely-generated commutative algebras over $k$ has pushouts.
\end{prop}

\begin{proof}
    To construct the pushout: suppose $A = k[x_1, \dots, x_n]/I_A$, $B = k[y_1, \dots, y_m]/I_B$, $C = k[z_1, \dots, z_r]/I_C$ with maps $A \overset{\phi}{\longleftarrow} C \overset{\psi}{\longrightarrow} B$. Then 
    \begin{equation}
        \label{eq:explicit-pushout-of-calg}
        A \otimes_C B = \frac{k[x_1, \dots, x_n, y_1, \dots, y_m]}{I_A + I_B + \sum_{i = 1}^r(\phi(z_i) - \psi(z_i))}
    \end{equation}
    Where, by a slight abuse of notation, $I_A, I_B$ denote the images of $I_A, I_B$ via the embeddings
    $$k[x_1, \dots, x_n] \hookrightarrow k[x_1, \dots, x_n, y_1, \dots, y_m] \hookleftarrow k[y_1, \dots, y_m]$$
\end{proof}

\begin{corollary}
    \label{prop:affine-pullback}
    If $k$ is either algebraically closed or finite, the category $\Affk^{(k)}$ has pullbacks. 
\end{corollary}

We are now ready to define the three target semantic categories we meet in this paper: one for graphical commutative algebra, and two for graphical algebraic geometry. As we shall see later in proposition \ref{prop:semantic-props}, structuring the (co)spans with respect to the appropriate inclusion functors ensures that each of these semantic categories is a PROP. 

\begin{definition}[Semantic Category of Graphical Commutative Algebra]
    Consider the structured cospan category $\FrCsp$, where $F$ is the inclusion of the free commutative algebras into the finitely generated commutative algebras. The objects of $\FrCsp$ are polynomial algebras $k[x_1, \dots, x_n]$, and its morphisms are structured cospans
        $$k[x_1, \dots, x_n] \overset{f}{\rightarrow} A \overset{g}{\leftarrow} k[y_1, \dots, y_m]$$
    where $A$ is some finitely generated commutative algebra over $k$. 
\end{definition}

\begin{definition}[Semantic Category of Graphical Algebraic Geometry]
    Consider the structured span categories $\FrSpk$ and $\FrSpq$, where $G$ and $G_q$ are the inclusions of the free spaces $k^n$, $\Fq^n$ in the categories of affine varieties. The objects are free spaces $k^n$ and $\Fq^n$ respectively, and morphisms in $\FrSpk$ are spans of polynomial maps 
    $$k^n \leftarrow V \rightarrow k^m \quad , \quad \Fq^n \leftarrow W \rightarrow \Fq^m$$ 
    respectively, where $V$ is an affine variety over $k$ and $W$ is a rational locus over $\Fq$. 
\end{definition}

\begin{prop}
    \label{prop:semantic-props}
    The three semantic categories are PROPs. 
\end{prop}

\begin{proof}
    The objects of $\FrCsp$ are the polynomial algebras $k[x_1, \dots, x_n]$, the monoidal product on objects is $k[x_1, \dots, x_n] \otimes k[y_1, \dots, y_m] = k[x_1, \dots, x_n, y_1, \dots, y_m]$; and on morphisms by the tensor-product of algebra morphisms on both legs of the cospans. 
    % \begin{equation}
    %     \begin{split}
    %         &\left(S_{n_1}(k) \overset{f_1}{\rightarrow} V \overset{g_1}{\leftarrow} S_{m_1}(k)\right) \otimes \left(S_{n_2}(k) \overset{f_2}{\rightarrow} W \overset{g_2}{\leftarrow} S_{m_2}(k)\right) \\ &= S_{n_1 + n_2}(k) \xrightarrow{f_1\otimes f_2} V\otimes W \xleftarrow{g_1 \otimes g_2} S_{m_1 + m_2}(k)
    %     \end{split}
    % \end{equation}
    % where $f_1 \otimes f_2$ acts on the first $n_1$ variables in $S_{n_1 + n_2}(k)$ according to $f_1$, and on the last $n_2$ variables according to $f_2$, and similarly for $g_1 \otimes g_2$. 

    The braidings are given by the graphs of 
    \begin{equation}
        \sigma: k[x_1, \dots, x_n, y_1, \dots, y_m] \to k[y_1, \dots, y_m, x_1, \dots, x_n]
    \end{equation}
    The proofs for the other two categories are similar. 
\end{proof}

The semantic categories for GCA and GAG are related, and we now explain this relation. 
First, notice that the structured (co)span construction is ``functorial'': 

\begin{prop}
    \label{prop:induced-functor-of-spans}
    \cite{baez_structured_2020} Given $F: \mathcal A \to \mathcal X$ and $F': \mathcal B \to \mathcal Y$, with $\mathcal X, \mathcal Y$ both having pullbacks, if there is a commutative square of functors 
    % https://q.uiver.app/#q=WzAsNCxbMCwwLCJcXG1hdGhjYWwgQSJdLFsxLDAsIlxcbWF0aGNhbCBYIl0sWzAsMSwiXFxtYXRoY2FsIEIiXSxbMSwxLCJcXG1hdGhjYWwgWSJdLFswLDEsIkYiXSxbMSwzLCJcXFBzaSJdLFswLDIsIlxcUGhpIiwyXSxbMiwzLCJGJyIsMl1d
    \[\begin{tikzcd}
        {\mathcal A} & {\mathcal X} \\
        {\mathcal B} & {\mathcal Y}
        \arrow["F", from=1-1, to=1-2]
        \arrow["\Phi"', from=1-1, to=2-1]
        \arrow["\Psi", from=1-2, to=2-2]
        \arrow["{F'}"', from=2-1, to=2-2]
    \end{tikzcd}\]
    and $\Psi$ preserves pullbacks, then there is a functor 
    \begin{equation}
        \begin{split}
            \Span(\Psi, \Phi): &_F \Span(\mathcal X) \to \,\, _{F'} \Span(\mathcal Y) 
        \end{split}
    \end{equation}
    which sends the structured span $(FA_1 \xleftarrow{f} X \xrightarrow{g} FA_2)$ to a structured span in $_F \Span(\mathcal Y)$ by hitting everything with $\Psi$. The dual statement is true for structured cospans. 
\end{prop}

\begin{corollary}
    In particular, in the situation of proposition \ref{prop:induced-functor-of-spans}, if $\Phi, \Psi$ are both equivalences of categories, then the induced functor $_F \Span(\mathcal X) \to \,\, _{F'} \Span(\mathcal Y) $ is also an equivalence of categories. 
\end{corollary}

Now to explain how our three semantic categories are related: recall from proposition \ref{prop:qred-is-left-adjoint} that there are functors 
\begin{equation}
    \begin{split}
        (-)_{red}: \CAlgfgk &\to \CAlgfgredk \\ 
        (-)_{qred}: \CAlgfgq &\to \CAlgfgqredq
    \end{split}
\end{equation}
each of which are left-adjoint to the inclusion in the opposite direction, and therefore they preserve pushouts. Moreover, they commute with the inclusions of the ``free'' algebras $F: \Polyk \hookrightarrow \CAlgfgk$ and $F_q: \Polyqqred \hookrightarrow \CAlgfgq$. The functoriality of the structured cospan construction (proposition \ref{prop:induced-functor-of-spans}) implies there are functors 
\begin{equation}
    \label{eq:functors-btx-cospan-semantic-targets}
    \begin{split}
        &\FrCsp \xrightarrow{(-)_{red}} \,\, _F \Csp \left( \CAlgfgredk \right) \\ 
        &\FrCspq \xrightarrow{(-)_{qred}} \,\, _{F_q} \Csp \left( \CAlgfgqredq \right)
    \end{split}
\end{equation}

Now recall that if $k$ is algebraically closed, then there is an anti-equivalence (\ref{prop:alg-closed-anti-equiv})
\begin{equation}
    \Spec: \CAlgfgredk \simeq \Affk^{op} 
\end{equation}
and if $\Fq$ is finite of size $q$, then there is an anti-equivalence (\ref{prop:finite-field-anti-equiv})
\begin{equation}
    \Spec: \CAlgfgqredq \simeq {\Affqq}^{op}
\end{equation}
Moreover, it is easy to check that these commute with the inclusions $F, F_q, G, G_q$ in the obvious way. So using the duality between span and cospan (\ref{prop:duality-of-co-span}) we obtain maps 
\begin{equation}
    \label{eq:semantic-target-equivalences}
    \begin{split}
        \FrCsp \xrightarrow{(-)_{red}} _F \Csp \left( \CAlgfgredk \right) &\simeq \,\,  \FrSpk \\ 
        \FrCspq \xrightarrow{(-)_{qred}} _{F_q} \Csp \left( \CAlgfgqredq \right) &\simeq \,\, \FrSpq
    \end{split}
\end{equation}

\section{Graphical Commutative Algebra}
\label{sec:cospans-of-calg} 

We now work towards a diagrammatic presentation, which we refer to as Graphical Commutative Algebra (GCA), of the structured cospan category $\FrCsp$. All lemmas, propositions, and theorems in this section are proven in Appendix~\ref{append:csp-proofs}, unless otherwise specified. 

\subsection{Defining the Calculus $\LangCsp$}

Our language $\LangCsp$ will be defined essentially as an amalgamation, modulo certain additional rewrite rules, of the Lawvere theory $\mathbb L_{CAlg_k}$ and its own opposite, so as to become a dagger-category: 

\begin{definition}
    Construct the PROP $\LangCsp = \mathbb L_{CAlg_k} +_E \mathbb L_{CAlg_k}^{op}$ where $E$ is the set of rewrite rules displayed in Figure~\ref{fig:rules-for-csp}. 
    
    We supply this PROP with an intended semantic functor 
    \begin{equation}
        \sem: \LangCsp \to \FrCsp
    \end{equation}
    by mapping each generator in $\mathbb L_{CAlg_k}$ via the cograph embedding, and each generator in $\mathbb L_{CAlg_k}^{op}$ via the graph embedding. 
    % https://q.uiver.app/#q=WzAsNyxbMCwwLCJcXG1hdGhiYiBMX3tDQWxnX2t9Il0sWzAsMiwiXFxtYXRoYmIgTF97Q0FsZ19rfV57b3B9Il0sWzAsMSwiXFxtYXRoYmIgTF97Q0FsZ19rfStcXG1hdGhiYiBMX3tDQWxnX2t9XntvcH0iXSxbMSwwLCJcXG1hdGhzZntQb2x5fV9rXntvcH0iXSxbMSwyLCJcXG1hdGhzZntQb2x5fV9rIl0sWzIsMSwiXFxGckNzcCJdLFsxLDEsIlxcbWF0aGJiIExfe0NBbGdfa30rX0UgXFxtYXRoYmIgTF97Q0FsZ19rfV57b3B9Il0sWzMsNSwiR197Q3NwfSJdLFs0LDUsIkdfe0NzcH1ee29wfSIsMl0sWzAsM10sWzEsNF0sWzAsMl0sWzEsMl0sWzIsNiwiIiwxLHsic3R5bGUiOnsiaGVhZCI6eyJuYW1lIjoiZXBpIn19fV0sWzYsNSwiXFxzZW0iLDFdXQ==
    \[\begin{tikzcd}[ampersand replacement=\&,cramped,row sep=small]
        {\mathbb L_{CAlg_k}} \& {\mathsf{Poly}_k^{op}} \\
        {\mathbb L_{CAlg_k}+\mathbb L_{CAlg_k}^{op}} \& {\LangCsp} \& \FrCsp \\
        {\mathbb L_{CAlg_k}^{op}} \& {\mathsf{Poly}_k}
        \arrow[from=1-1, to=1-2]
        \arrow[from=1-1, to=2-1]
        \arrow["{G_{Csp}}", from=1-2, to=2-3]
        \arrow[two heads, from=2-1, to=2-2]
        \arrow["\sem"{description}, from=2-2, to=2-3]
        \arrow[from=3-1, to=2-1]
        \arrow[from=3-1, to=3-2]
        \arrow["{G_{Csp}^{op}}"', from=3-2, to=2-3]
    \end{tikzcd}\]
\end{definition}

\begin{figure*}[h]
    \centering 
    \caption{Additional Rewrite Rules for $\LangCsp$} 
    \label{fig:rules-for-csp}
    \input{./figures/L-Csp-Rules.tikz}
\end{figure*}

We must first prove that the functor $\sem$ is well-defined, which is equivalent to the following statement: 

\begin{restatable}[Cospan Soundness]{prop}{cspsound}
    \label{prop:csp-sound}
    Every rewrite rule in $E$ is sound with respect to the intended semantic $\sem$. 
\end{restatable}

Once well-definedness has been established, we can derive the following useful rewrite rules from the given rules: 

\begin{restatable}{lemma}{derivedRulesOfCsp}
    \label{lem:derived-rules}
    The following are some useful derived rules of $\LangCsp$: 
    \begin{align}
        \label{eq:z-special}
        \tag{special}
        \input{./figures/z-special-lhs.tikz} &= \begin{tikzpicture}
	\begin{pgfonlayer}{nodelayer}
		\node [style=none] (0) at (-1, 0) {};
		\node [style=none] (1) at (1, 0) {};
	\end{pgfonlayer}
	\begin{pgfonlayer}{edgelayer}
		\draw (0.center) to (1.center);
	\end{pgfonlayer}
\end{tikzpicture}
 \\ 
        \label{eq:poly-trans} 
        \tag{poly-trp}
        \input{./figures/poly-trans-lhs.tikz} &= \begin{tikzpicture}
	\begin{pgfonlayer}{nodelayer}
		\node [style=lmat] (0) at (0, 0) {$f$};
		\node [style=none] (1) at (-1.5, 0) {};
		\node [style=none] (2) at (1.5, 0) {};
	\end{pgfonlayer}
	\begin{pgfonlayer}{edgelayer}
		\draw (1.center) to (2.center);
	\end{pgfonlayer}
\end{tikzpicture}
 \qquad \forall f \in (k[x_1, \dots, x_n])^m
    \end{align}
\end{restatable}

\subsection{Universality and Completeness} 

To consider the nature of the (essential) image of $\sem$ in $\FrCsp$, we first prove the following sequence of lemmas in the language $\LangCsp$, which will allow us to introduce an essential piece of syntactic sugar for summarizing a diagram. 

\begin{restatable}{lemma}{principleidealflex}
    \label{lem:principle-ideal-flex}
    For any $f \in k[x_1, \dots, x_n]$, the following is derivable from the rewrite rules of $\LangCsp$: 
    $$\input{./figures/polyzero-going-left.tikz} = \input{./figures/polyzero-going-right.tikz}$$
\end{restatable}
So we just speak of gadgets of the form 
$\input{./figures/polyzero.tikz}$
\begin{restatable}{lemma}{principleidealscommute}
    \label{lem:principle-ideals-commute}
    For any $f, g \in k[x_1, \dots, x_n]$, the following is derivable from the rewrite rules of $\LangCsp$: 
    $$\input{./figures/polyzero-commute-lhs.tikz} = \input{./figures/polyzero-commute-rhs.tikz}$$
\end{restatable}
\begin{restatable}{lemma}{idealbasis}
    \label{lem:ideal-basis}
    Let $f_1, \dots, f_m \in k[x_1, \dots, x_n]$, and pick some $f$ inside the ideal $(f_1, \dots, f_m)$. Then the following is derivable from the rewrite rules of $\LangCsp$: 
    \begin{equation}
        \label{eq:ideal-basis-with-redundancy}
        \input{./figures/polyzero-ideal-basis.tikz} = \input{./figures/polyzero-ideal-basis-with-redundancy.tikz}. 
    \end{equation}
\end{restatable}
Lemmas \ref{lem:principle-ideals-commute} and \ref{lem:ideal-basis} mean that we can coherently define the following gadget: 
\begin{definition}
    \label{def:ring-quot-gadget}
    Let $I \subset k[x_1, \dots, x_n]$ be an ideal. Since $k[x_1, \dots, x_n]$ is Noetherian \cite{eisenbud_commutative_1995}, there exists a finite set of polynomials $f_1, \dots, f_m \in k[x_1, \dots, x_n]$ such that $I = (f_1, \dots, f_m)$. Then define 
    \begin{equation}
        \begin{tikzpicture}
	\begin{pgfonlayer}{nodelayer}
		\node [style=none] (0) at (-1, -0.5) {};
		\node [style=none] (1) at (1, -0.5) {};
		\node [style=Z dot] (2) at (0, -0.5) {};
		\node [style=idealket] (3) at (0, 0.5) {$I$};
	\end{pgfonlayer}
	\begin{pgfonlayer}{edgelayer}
		\draw (0.center) to (1.center);
		\draw (3) to (2);
	\end{pgfonlayer}
\end{tikzpicture}
 \coloneq \input{./figures/polyzero-ideal-basis.tikz}
    \end{equation}
    Note that neither the choice of the basis $f_1, \dots, f_m$, nor the ordering of them, matters in the definition, since any two such choices give diagrams which are rewritable to each other per lemmas \ref{lem:principle-ideals-commute} and \ref{lem:ideal-basis}. For concreteness, we may always choose the reduced Grobner basis with respect to some chosen monomial ordering. 
\end{definition}

\begin{corollary} \label{cor:idem}
    Idempotence $\input{./figures/ring-quot-twice.tikz} = \begin{tikzpicture}
	\begin{pgfonlayer}{nodelayer}
		\node [style=none] (0) at (-1, -0.5) {};
		\node [style=none] (1) at (1, -0.5) {};
		\node [style=Z dot] (2) at (0, -0.5) {};
		\node [style=idealket] (3) at (0, 0.5) {$I$};
	\end{pgfonlayer}
	\begin{pgfonlayer}{edgelayer}
		\draw (0.center) to (1.center);
		\draw (3) to (2);
	\end{pgfonlayer}
\end{tikzpicture}
$ is provable in $\LangCsp$.
\end{corollary}

We now explain the semantics of these gadgets. 

\begin{restatable}[Quotient Gadgets]{prop}{quotientgadgets}
    \label{prop:quotient-gadget}
    The semantic of the gadget $\begin{tikzpicture}
	\begin{pgfonlayer}{nodelayer}
		\node [style=none] (0) at (-1, -0.5) {};
		\node [style=none] (1) at (1, -0.5) {};
		\node [style=Z dot] (2) at (0, -0.5) {};
		\node [style=idealket] (3) at (0, 0.5) {$I$};
	\end{pgfonlayer}
	\begin{pgfonlayer}{edgelayer}
		\draw (0.center) to (1.center);
		\draw (3) to (2);
	\end{pgfonlayer}
\end{tikzpicture}
$ is (any finite presentation of) the finitely generated commutative algebra defined by $I$. More precisely, if $I \subset k[x_1, \dots, x_n]$ is any ideal, then 
    $$\left[\begin{tikzpicture}
	\begin{pgfonlayer}{nodelayer}
		\node [style=none] (0) at (-1, -0.5) {};
		\node [style=none] (1) at (1, -0.5) {};
		\node [style=Z dot] (2) at (0, -0.5) {};
		\node [style=idealket] (3) at (0, 0.5) {$I$};
	\end{pgfonlayer}
	\begin{pgfonlayer}{edgelayer}
		\draw (0.center) to (1.center);
		\draw (3) to (2);
	\end{pgfonlayer}
\end{tikzpicture}
\right] = \left(S_k(n) \twoheadrightarrow S_k(n)/I \twoheadleftarrow S_k(n)\right)$$
    where both legs are the canonical projections $\pi_I$ against $I$.     
\end{restatable}

This is essentially what brings us from a language for $\Polyk$ to a language that has full information about $\CAlgfgk$, which forces us to introduce the following notation: 

\begin{notation}
    If $A = S_k(n)/I$, and $f \in (k[x_1, \dots, x_n])^m$, denote by $f^\#_I: S_k(m) \to A$ the algebra morphism one gets by passing to the quotient: 
    \begin{equation}
        S_k(m) \xrightarrow{f^\#} S_k(n) \twoheadrightarrow S_k(n)/I = A
    \end{equation}
    Or equivalently, the algebra morphism one gets by reducing $f$ against $I$ component-wise to obtain $f_I \in A^m$, and then taking its extension to an algebra morphism $f_I^\#: S_k(m) \to A$. 
    % Note that since the operation of taking the extension is bijective (by virtue of $S_k$ being the free construction for commutative algebras), so the mapping $f \mapsto f^\#_I$ is surjective with kernel $I^m$. 
\end{notation}

\begin{restatable}[Cospan Universality]{prop}{cspUniversal}
    \label{prop:csp-universal}
    For any ideal $I \subset k[x_1, \dots, x_n]$ and tuples of polynomials $f \in (k[x_1, \dots, x_n])^m$ and $g \in (k[x_1, \dots, x_n])^{m'}$, 
    \begin{equation}
        \left[\input{./figures/span-form.tikz}\right] = \left(S_k(m) \xrightarrow{f^\#_I} A  \xleftarrow{g^\#_I} S_k(m')\right)
    \end{equation}
    where $A = S_k(n)/I$. Since every algebra morphism $S_k(m) \to A$ is equal to $f^\#_I$ for some $f$, $\sem: \LangCsp \to \FrCsp$ is full. 
\end{restatable}

It remains now to prove the completeness of the language $\LangCsp$. We do this by reducing diagrams to a pseudo-normal form as follows: 

\begin{definition}
    A diagram $D$ in $\LangCsp$ is said to be in \textbf{cospan form} if it is drawn as 
    \begin{equation}
        \input{./figures/span-form.tikz} 
    \end{equation}
    for some ideal $I$ and some polynomial tuples $f, g$. 
\end{definition}

\begin{restatable}[Cospan Normal Form]{prop}{cospanForm}
    \label{prop:cospan-form}
    Every diagram in $\LangCsp$ can be rewritten into cospan form. 
\end{restatable}

Thus, to show completeness of $\LangCsp$, it suffices to show that, for any two diagrams in cospan form, if they have the same semantic in $\FrCsp$, then they can be rewritten into each other. We do this via the following lemmas: 

\begin{restatable}{lemma}{diagramSheafRing}
    \label{lem:diagram-sheaf-ring}
    Let $I \subset k[x_1, \dots, x_n]$ be an ideal, and let $f, g \in (k[x_1, \dots, x_n])^m$ be any two polynomial tuples such that $f-g\in I^m$. That is, let $f, g$ be such that $f^\#_I = g^\#_I$. Then the following diagram rewrite is derivable in $\LangCsp$:  
    \begin{equation}
        \input{./figures/ring-on-variety-lhs.tikz} = \input{./figures/ring-on-variety-rhs.tikz} 
    \end{equation}
\end{restatable}

\begin{restatable}{lemma}{diagramRespectsKernel}
    \label{lem:diagram-respects-kernel} 
    Suppose $I \subset S_k(n)$ and $J \subset S_k(m)$ are ideals of polynomials and $f \in (S_k(n))^m$. Suppose $f^\#_I: S_k(m) \to S_k(n)/I$ satisfies $J \subset \ker f^\#_I$. 
    % % https://q.uiver.app/#q=WzAsMyxbMCwwLCJTX2sobSkiXSxbMCwxLCJTX2sobSkvSiJdLFsxLDEsIlNfayhuKS9JIl0sWzAsMiwiXFxzaWdtYV5cXCNfSSJdLFswLDEsIlxccGlfSiIsMl0sWzEsMiwiXFxhbHBoYSIsMl1d
    % \[\begin{tikzcd}[ampersand replacement=\&]
    %     {S_k(m)} \\
    %     {S_k(m)/J} \& {S_k(n)/I}
    %     \arrow["{\pi_J}"', from=1-1, to=2-1]
    %     \arrow["{\sigma^\#_I}", from=1-1, to=2-2]
    %     \arrow["\alpha"', from=2-1, to=2-2]
    % \end{tikzcd}\]
    Then the following rewrite holds:
    $$\input{./figures/regular-map.tikz} = \input{./figures/regular-map-without-im.tikz}$$
\end{restatable}

\begin{restatable}{lemma}{isoInverseIsOp}
    \label{lem:iso-inverse-is-op}
    Suppose $I \subset S_k(n)$ and $J \subset S_k(m)$ are ideals, $\alpha: S_k(m)/J \to S_k(n)/I$ is an isomorphism of algebras, and $\sigma, \tau$ are a pair of polynomial tuples such that they represent $\alpha, \alpha^{-1}$ respectively, i.e. such that the following commutes: 
    % https://q.uiver.app/#q=WzAsNixbMCwwLCJTX2sobSkiXSxbMCwxLCJTX2sobikiXSxbMSwwLCJTX2sobSkvSiJdLFsxLDEsIlNfayhuKS9JIl0sWzIsMCwiU19rKG0pIl0sWzIsMSwiU19rKG4pIl0sWzAsMiwiXFxwaV9KIiwwLHsic3R5bGUiOnsiaGVhZCI6eyJuYW1lIjoiZXBpIn19fV0sWzQsMiwiXFxwaV9KIiwyLHsic3R5bGUiOnsiaGVhZCI6eyJuYW1lIjoiZXBpIn19fV0sWzEsMywiXFxwaV9JIiwyLHsic3R5bGUiOnsiaGVhZCI6eyJuYW1lIjoiZXBpIn19fV0sWzUsMywiXFxwaV9JIiwwLHsic3R5bGUiOnsiaGVhZCI6eyJuYW1lIjoiZXBpIn19fV0sWzAsMSwiXFxzaWdtYV5cXCMiLDJdLFs1LDQsIlxcdGF1XlxcIyIsMl0sWzIsMywiXFxhbHBoYSJdLFszLDIsIlxcYWxwaGFeey0xfSIsMCx7Im9mZnNldCI6LTN9XV0=
    \[\begin{tikzcd}[ampersand replacement=\&]
        {S_k(m)} \& {S_k(m)/J} \& {S_k(m)} \\
        {S_k(n)} \& {S_k(n)/I} \& {S_k(n)}
        \arrow["{\pi_J}", two heads, from=1-1, to=1-2]
        \arrow["{\sigma^\#}"', from=1-1, to=2-1]
        \arrow["\alpha", from=1-2, to=2-2]
        \arrow["{\pi_J}"', two heads, from=1-3, to=1-2]
        \arrow["{\pi_I}"', two heads, from=2-1, to=2-2]
        \arrow["{\alpha^{-1}}", shift left=3, from=2-2, to=1-2]
        \arrow["{\tau^\#}"', from=2-3, to=1-3]
        \arrow["{\pi_I}", two heads, from=2-3, to=2-2]
    \end{tikzcd}\]
    Then the following rewrite holds:
    \begin{equation}
        \input{./figures/sigma-inverse.tikz} = \input{./figures/sigma-op.tikz}
    \end{equation}
\end{restatable}

By composing $\begin{tikzpicture}
	\begin{pgfonlayer}{nodelayer}
		\node [style=rmat] (0) at (0, 0) {$\sigma$};
		\node [style=none] (1) at (-1, 0) {};
		\node [style=none] (2) at (1, 0) {};
	\end{pgfonlayer}
	\begin{pgfonlayer}{edgelayer}
		\draw (1.center) to (2.center);
	\end{pgfonlayer}
\end{tikzpicture}
$ on the right to both sides of the equation in \ref{lem:iso-inverse-is-op}, we obtain: 

\begin{restatable}{corollary}{conjByIso}
    \label{cor:conj-by-iso}
    If $\alpha: S_k(m)/J \to S_k(n)/I$ is an isomorphism of algebras, and $\sigma \in (k[x_1, \dots, x_n])^m$ represents $\alpha$ as in lemma \ref{lem:iso-inverse-is-op}, then 
    $$\begin{tikzpicture}
	\begin{pgfonlayer}{nodelayer}
		\node [style=Z dot] (1) at (0, -0.5) {};
		\node [style=none] (2) at (1, -0.5) {};
		\node [style=none] (3) at (-1, -0.5) {};
		\node [style=idealket] (5) at (0, 0.5) {$J$};
	\end{pgfonlayer}
	\begin{pgfonlayer}{edgelayer}
		\draw (3.center) to (2.center);
		\draw (5) to (1);
	\end{pgfonlayer}
\end{tikzpicture}
 = \input{./figures/ring-quot-conj-sigma.tikz}$$
\end{restatable}

Lemmas \ref{lem:diagram-sheaf-ring}, \ref{lem:diagram-respects-kernel}, and \ref{lem:iso-inverse-is-op} suffice to show that any two diagrams in cospan normal form which have the same semantic in $\FrCsp$ can be rewritten into each other. And therefore: 

\begin{restatable}{prop}{cspComplete}
    \label{prop:csp-complete}
    The functor $\sem$ is faithful. That is, $\LangCsp$ is complete as a language for $\FrCsp$. 
\end{restatable}

We have now completed the demonstration of our first main theorem: 

\begin{theorem}
    \label{thm:csp-iso}
    The functor $\sem: \LangCsp \to\,\, \FrCsp$ is an isomorphism of PROPs. Thus $\LangCsp$ is a sound, universal and complete language for $\FrCsp$. 
\end{theorem}

\begin{proof}
    Propositions \ref{prop:csp-sound}, \ref{prop:csp-universal}, and \ref{prop:csp-complete} together imply theorem \ref{thm:csp-iso}. 
\end{proof}

\section{Graphical Algebraic Geometry}
\label{sec:structured-spans-of-alg-varieties} 

Our focus now turns from algebra to geometry. 
As mentioned in Section~\ref{sec:background}, this is where we will make heavy use of the Nullstellensatze to bridge the gap from algebra to geometry. 

In this section $k$ will always denote an algebraically closed field, and $\Fq$ will always denote a finite field of size $q$. All proofs are contained in Appendix~\ref{sec:append-sp} unless otherwise specified. 

Recall from equation \ref{eq:semantic-target-equivalences} that as a consequence of the Nullstellensatze, there are two functors 
\begin{equation}
    \begin{split}
        &\FrCsp \xrightarrow{\Spec (-)_{red}} \,\, \FrSpk \\
        &\FrCspq \xrightarrow{\Spec (-)_{qred}} \,\, \FrSpq
    \end{split}
\end{equation}
We now seek diagrammatic presentations for $\FrSpk$ and $\FrSpq$. Since we already have a presentation 
$$\LangCsp \cong \FrCsp,$$ 
so we seek to find $\LangSpk$ and $\LangSpq$ that fit into the following commutative diagrams: 
% https://q.uiver.app/#q=WzAsOCxbMCwwLCJcXExhbmdDc3AiXSxbMSwwLCJcXEZyQ3NwIl0sWzEsMSwiXFxGclNwayJdLFswLDEsIlxcTGFuZ1NwayJdLFsyLDAsIlxcTGFuZ0NzcCJdLFszLDAsIlxcRnJDc3AiXSxbMywxLCJcXEZyU3BxIl0sWzIsMSwiXFxMYW5nU3BxIl0sWzAsMSwiXFxzZW0iXSxbMSwyLCJcXFNwZWMgKC0pX3tyZWR9IiwwLHsic3R5bGUiOnsiaGVhZCI6eyJuYW1lIjoiZXBpIn19fV0sWzMsMiwiXFxzZW1fayIsMix7InN0eWxlIjp7ImJvZHkiOnsibmFtZSI6ImRhc2hlZCJ9fX1dLFswLDMsIiIsMix7InN0eWxlIjp7ImJvZHkiOnsibmFtZSI6ImRhc2hlZCJ9LCJoZWFkIjp7Im5hbWUiOiJlcGkifX19XSxbNCw1LCJcXHNlbSJdLFs0LDcsIiIsMCx7InN0eWxlIjp7ImJvZHkiOnsibmFtZSI6ImRhc2hlZCJ9LCJoZWFkIjp7Im5hbWUiOiJlcGkifX19XSxbNyw2LCJcXHNlbV9xIiwyLHsic3R5bGUiOnsiYm9keSI6eyJuYW1lIjoiZGFzaGVkIn19fV0sWzUsNiwiXFxTcGVjKC0pX3txcmVkfSIsMCx7InN0eWxlIjp7ImhlYWQiOnsibmFtZSI6ImVwaSJ9fX1dXQ==
\[\begin{tikzcd}[ampersand replacement=\&,cramped,column sep=small]
	\LangCsp \& \FrCsp \& \LangCsp \& \FrCsp \\
	\LangSpk \& \FrSpk \& \LangSpq \& \FrSpq
	\arrow["\sem", from=1-1, to=1-2]
	\arrow[dashed, two heads, from=1-1, to=2-1]
	\arrow["{\Spec (-)_{red}}", two heads, from=1-2, to=2-2]
	\arrow["\sem", from=1-3, to=1-4]
	\arrow[dashed, two heads, from=1-3, to=2-3]
	\arrow["{\Spec(-)_{qred}}", two heads, from=1-4, to=2-4]
	\arrow["{\sem_k}"', dashed, from=2-1, to=2-2]
	\arrow["{\sem_q}"', dashed, from=2-3, to=2-4]
\end{tikzcd}\]
such that $\sem_k$ and $\sem_q$ are both isomorphisms of PROPs. 

\subsection{GAG over Algebraically Closed Fields} 
\label{sec:alg-closed} 

For algebraically closed fields, Hilbert's Nullstellensatz tells us that affine varieties correspond to reduced coordinate rings. So we introduce the following rewrite rule $\eqref{ax:reduced}$, which intuitively should be thought of as ``reducing'' all rings. This intuition is made formal in lemma \ref{lem:diagram-reduction}. 
\begin{definition}
    Let $\LangSpk$ be the quotient of $\LangCsp$ defined by the following additional rewrite rule: 
    \begin{equation}
        \label{ax:reduced}
        \input{./figures/nilpotent-lhs.tikz} = \begin{tikzpicture}
	\begin{pgfonlayer}{nodelayer}
		\node [style=X dot] (0) at (0.5, 0) {};
		\node [style=none] (1) at (-0.5, 0) {};
	\end{pgfonlayer}
	\begin{pgfonlayer}{edgelayer}
		\draw (1.center) to (0);
	\end{pgfonlayer}
\end{tikzpicture}

        \tag{red}
    \end{equation}
\end{definition}

\begin{restatable}{prop}{reductionSound}
    \label{prop:reductionSound}
    The PROP morphism 
    \begin{equation}
        \LangCsp \xrightarrow{\Spec \sem_{red}} \FrSpk
    \end{equation}
    factors through $\LangSpk$. In other words, rule $\eqref{ax:reduced}$ is sound for affine varieties over an algebraically closed field. 
\end{restatable}

Thus we obtain a PROP morphism 
\begin{equation}
    \LangSpk \xrightarrow{\sem_k} \FrSpk 
\end{equation}
It remains to show that this morphism is faithful: in other words, that our rewrite rules are complete. 

\begin{restatable}{lemma}{diagramReduction}
    \label{lem:diagram-reduction}
    Let $I \subset k[x_1, \dots, x_n]$ be an ideal. The following is derivable from the rewrite rules of $\LangSpk$: 
    \begin{equation}
        \label{eq:reduction-rewrite} 
        \begin{tikzpicture}
	\begin{pgfonlayer}{nodelayer}
		\node [style=none] (0) at (-1, -0.5) {};
		\node [style=none] (1) at (1, -0.5) {};
		\node [style=Z dot] (2) at (0, -0.5) {};
		\node [style=idealket] (3) at (0, 0.5) {$I$};
	\end{pgfonlayer}
	\begin{pgfonlayer}{edgelayer}
		\draw (0.center) to (1.center);
		\draw (3) to (2);
	\end{pgfonlayer}
\end{tikzpicture}
 = \begin{tikzpicture}
	\begin{pgfonlayer}{nodelayer}
		\node [style=none] (0) at (-1, -0.5) {};
		\node [style=none] (1) at (1, -0.5) {};
		\node [style=Z dot] (2) at (0, -0.5) {};
		\node [style=idealket] (3) at (0, 0.5) {$\sqrt{I}$};
	\end{pgfonlayer}
	\begin{pgfonlayer}{edgelayer}
		\draw (0.center) to (1.center);
		\draw (3) to (2);
	\end{pgfonlayer}
\end{tikzpicture}

    \end{equation}
\end{restatable}

Leveraging Hilbert's Nullstellensatz, we see that lemma \ref{lem:diagram-reduction} alone is sufficient to prove completeness for $\LangSpk$ (See Appendix~\ref{sec:append-sp} for details in proof): 

\begin{restatable}{prop}{spkComplete}
    \label{prop:spk-complete}
    The semantic map $\LangSpk \xrightarrow{\sem_k} \FrSpk$ is faithful. In other words, $\LangSpk$ is complete for $\FrSpk$. 
\end{restatable}

Fullness, on the other hand, is inherited from $\LangCsp$. Thus: 

\begin{theorem}
    The semantic map $\LangSpk \xrightarrow{\sem_k} \FrSpk$ is an isomorphism of PROPs; $\LangSpk$ is a universal, sound, and complete language for $\FrSpk$. 
\end{theorem}

\subsection{GAG over Finite Fields} 
\label{sec:spans-over-ff}

The recipe for obtaining a language for spans of rational loci of affine varieties over finite fields is almost identical to the one used for spans of affine varieties over closed fields. We begin by introducing a rewrite rule \eqref{ax:q-reduced} which intuitively $q$-reduces all rings: 

\begin{definition}
    Let $\LangSpq$ be the quotient of $\LangCsp$ by the following additional rewrite rule: 
    \begin{equation}
        \label{ax:q-reduced} 
        \input{./figures/qred-lhs.tikz} = \begin{tikzpicture}
	\begin{pgfonlayer}{nodelayer}
		\node [style=none] (0) at (-1, 0) {};
		\node [style=none] (1) at (1, 0) {};
	\end{pgfonlayer}
	\begin{pgfonlayer}{edgelayer}
		\draw (0.center) to (1.center);
	\end{pgfonlayer}
\end{tikzpicture}
 
        \tag{qred}
    \end{equation}
\end{definition}

\begin{restatable}{prop}{qredSound}
    \label{prop:qred-sound}
    The PROP morphism 
    \begin{equation}
        \LangCsp \xrightarrow{\Spec [-]_{qred}} \,\, \FrSpq
    \end{equation}
    factors through $\LangSpq$. In other words, the rule \ref{ax:q-reduced} is sound with respect to $\Fq$-rational loci. 
\end{restatable}

Thus we obtain a PROP morphism 
\begin{equation}
    \LangSpq \xrightarrow{\sem_q} \,\, \FrSpq
\end{equation}
Now to show that this morphism is faithful: again, we start by showing that the extra rule \eqref{ax:q-reduced} implies that ``rings are $q$-reduced'' in our language:

\begin{restatable}{lemma}{diagramQReduction}
    \label{lem:diagram-qreduction}
    The following is derivable from the rewrite rules of $\LangSpq$: 
    $\begin{tikzpicture}
	\begin{pgfonlayer}{nodelayer}
		\node [style=Z dot] (0) at (0, -0.5) {};
		\node [style=idealket] (1) at (0, 0.5) {$\Gamma_q^n$};
		\node [style=none] (2) at (-1, -0.5) {};
		\node [style=none] (3) at (1, -0.5) {};
	\end{pgfonlayer}
	\begin{pgfonlayer}{edgelayer}
		\draw (2.center) to (3.center);
		\draw (0) to (1);
	\end{pgfonlayer}
\end{tikzpicture}
 = \begin{tikzpicture}
	\begin{pgfonlayer}{nodelayer}
		\node [style=none] (0) at (-1, 0) {};
		\node [style=none] (1) at (1, 0) {};
	\end{pgfonlayer}
	\begin{pgfonlayer}{edgelayer}
		\draw (0.center) to (1.center);
	\end{pgfonlayer}
\end{tikzpicture}
 $
\end{restatable}

Together with the finite field Nullstellensatz, lemma \ref{lem:diagram-qreduction} is sufficient for proving completeness (see Appendix~\ref{sec:append-sp} for details in proof): 

\begin{restatable}{prop}{spqComplete} 
    \label{prop:spq-complete}
    The semantic map $\sem_q$ is faithful. In other words, the rules of $\LangSpq$ are complete. 
\end{restatable}

\begin{restatable}{theorem}{spq-iso}
    The semantic map $\sem_q$ is an isomorphism of PROPs. 
\end{restatable}

\subsection{Matrix Semantics for GAG}

Given a structured span of affine varieties over a field $k$ (not necessarily algebraically closed) of the form
   $k^n \xleftarrow{\varphi} V \xrightarrow{\psi} k^m$
we can ``repackage'' the data of this structured span as an indexed family of affine varieties
\begin{equation}
    V_{\mathbf x \mathbf y} := \varphi^{-1}(\mathbf x) \cap \psi^{-1}(\mathbf y)
\end{equation}
each of which is a subvariety of $V$.

In the particular case where $k = \Fq$ is a finite field, and $V$ is an $\Fq$-rational locus, the family of affine varieties $V_{\mathbf x \mathbf y}$ consists of finite sets, which we can view as entries in an $\mathbb N$-matrix: 

\begin{definition}
    Let $M_q: \,\, \FrSpq \to \Mat_\mathbb N$ be the functor which sends a structured span of $\Fq$-rational loci
    \begin{equation}
        \Fq^n \xleftarrow{\varphi} V \xrightarrow{\psi} \Fq^m
    \end{equation}
    to the $q^n \times q^m$-matrix $M_q(\xleftarrow{\varphi} \xrightarrow{\psi})$ with entries
    \begin{equation}
        M_q(\xleftarrow{\varphi} \xrightarrow{\psi})_{\mathbf x \mathbf y} = |\varphi^{-1}(\mathbf x) \cap \psi^{-1}(\mathbf y)|.
    \end{equation}
\end{definition}

The following is immediate (proofs are given in \ref{append:spq-matn}):

\begin{restatable}{prop}{MSound}
    The functor $M_q$ is well-defined and fully faithful. 
\end{restatable}

Since $\Affqq \simeq \mathsf{FinSet}$, this is just a restriction of the equivalence, well-known in folklore, between $\Span(\mathsf{FinSet}) \simeq \Mat_\mathbb N$ in disguise. 

But since we already know that $\LangSpq$ is a sound, universal, and complete language for $\FrSpq$ from Section~\ref{sec:spans-over-ff}, we have:

\begin{corollary}
    The diagrammatic language $\LangSpq$ serves as a sound, universal, and complete language for a full subcategory of $\Mat_\mathbb N$ consisting of the objects $q^n$ with $n \in \mathbb N$.
\end{corollary}
We use $\semqMat: \LangSpq \to \Mat_\mathbb N$ to denote this new semantic.

\section{From GAG to Counting CSP} 
\label{sec:GAG-to-CSP}
We now briefly illustrate the direct connection between GAG and CSP. Recall that a constraint satisfaction problem over a domain $D$ and a constraint language $\Gamma$ (typically a set of finite arity relations over $D$) is a pair $\langle X, C\rangle$, where: 
\begin{itemize}
    \item $X$ is a set of variables $x_1, \dots, x_n$, 
    \item $C = \{C_1, \dots, C_m\}$ is a set of constraints, each of which is a pair $C_j = \langle s_j, g_j \rangle$ where $s_j \subset X$ is the scope of the constraint, and $g_j \in \Gamma$ is the content of the constraint. 
\end{itemize}
An assignment $\sigma: X \to D$ is said to satisfy the CSP if for every $j$, the tuple $(\sigma(x_i))_{i \in s_j}$ is in the relation $g_j$. The problem of determining whether there exists an assignment $\sigma$ that satisfies $\langle X, C\rangle$ is called the decision problem $\mathsf{CSP}(\Gamma)$; while the problem of counting the number of assignments that satisfy $\langle X, C\rangle$ is called the counting problem $\#\mathsf{CSP}(\Gamma)$. 

Now suppose that our domain of values is a finite field $\Fq$, and our constraint language $\Gamma$ consists of polynomial maps $g: \Fq^n \to \Fq$, read as the requirement $g = 0$. Of course, every finite-arity relation over $\Fq$ can be expressed as an affine variety, and therefore our $\Gamma$ is just the set of all relations over $D$. The finite field GAG gives an easy representation of instances of $\#\mathsf{CSP}(\Gamma)$ in the following way: 

\begin{restatable}{prop}{GAGasCSP}
    \label{prop:gag-as-csp}
    Let $\langle X, C\rangle$ be a CSP instance over the domain $\Fq$, with set of variables $X = \{x_1, \dots, x_n\}$, and polynomial constraints $g_1, \dots, g_m \in k[x_1, \dots, x_n]$. Let $I$ be the ideal generated by $g_1, \dots, g_m$. Then the span semantic of the ideal gadget is
    \begin{equation}
        \label{eq:ideal-gadget-as-csp-instance}
    \input{./figures/ring-quot-gadget-numbered.tikz} \overset{\sem_q}\mapsto \left(\Fq^n \hookleftarrow V(I)^{\Fq} \hookrightarrow \Fq^n \right)
    \end{equation}
    where, recall, $V(I)^{\Fq} := \{\mathbf x \in \Fq^n: g_1(\mathbf x) = \dots = g_j(\mathbf x) = 0\}$. In other words, it is the set of assignments satisfying $\langle X, C\rangle$. 

    The same diagram composed with no wires on either side has the following matrix semantic: 
    \begin{equation}
        \label{eq:count-csp}
        \begin{tikzpicture}
	\begin{pgfonlayer}{nodelayer}
		\node [style=Z dot] (2) at (0, -0.5) {};
		\node [style=idealket] (3) at (0, 0.5) {$I$};
	\end{pgfonlayer}
	\begin{pgfonlayer}{edgelayer}
		\draw (3) to (2);
	\end{pgfonlayer}
\end{tikzpicture}
 \overset{\semqMat[-]}\mapsto \left|V(I)^{(\Fq)}\right|
    \end{equation}
    In other words, it is the number of assignments that satisfy $\langle X, C\rangle$. 
\end{restatable}

\begin{theorem}
    The problem of rewriting arbitrary GAG diagrams into each other is \#P-hard.
\end{theorem}

\begin{proof}
    This follows from equation \eqref{eq:count-csp} in proposition \ref{prop:gag-as-csp}.
    We can fix a family of polynomials $\{g_j\}_{j \in \mathbb N}$ such that each $g_j$ is a polynomial of some number $n_j$ of variables with exactly $j$ roots in $\Fq^{n_j}$.

    The problem $\#\mathsf{CSP}(\Gamma)$ reduces to rewriting a given diagram $\begin{tikzpicture}
	\begin{pgfonlayer}{nodelayer}
		\node [style=Z dot] (0) at (-0.5, 0) {};
		\node [style=bra] (1) at (0.5, 0) {$I$};
	\end{pgfonlayer}
	\begin{pgfonlayer}{edgelayer}
		\draw (1) to (0);
	\end{pgfonlayer}
\end{tikzpicture}
$ into the form $\begin{tikzpicture}
	\begin{pgfonlayer}{nodelayer}
		\node [style=rmat] (0) at (0, 0) {$g_j$};
		\node [style=X dot] (1) at (1, 0) {};
		\node [style=Z dot] (2) at (-1, 0) {};
	\end{pgfonlayer}
	\begin{pgfonlayer}{edgelayer}
		\draw (1) to (2);
	\end{pgfonlayer}
\end{tikzpicture}
$ for some $j \in \mathbb N$. Since $\#\mathsf{CSP}(\Gamma)$ is \#P-complete \cite{jerrum_counting_2017}, so is the problem of rewriting arbitrary GAG diagrams into each other.
\end{proof}

\begin{remark}
    In the case $q = 2$, polynomial constraints are equivalent to boolean constraints, and so $\#\mathsf{CSP}(\Gamma)$ is $\#\mathsf{SAT}$.
\end{remark}

\begin{remark}
    Conversely, since every diagram in GAG can be rewritten into cospan form,
    so every \textit{closed} diagram (a diagram without input or output wires) can be rewritten into: 
    \begin{equation}
        \input{./figures/closed-span-form.tikz} \overset{\polydel}= \begin{tikzpicture}
	\begin{pgfonlayer}{nodelayer}
		\node [style=Z dot] (0) at (-1, -0.5) {};
		\node [style=Z dot] (1) at (1, -0.5) {};
		\node [style=Z dot] (2) at (0, -0.5) {};
		\node [style=idealket] (3) at (0, 0.5) {$I$};
	\end{pgfonlayer}
	\begin{pgfonlayer}{edgelayer}
		\draw [in=180, out=0] (0) to (1);
		\draw (3) to (2);
	\end{pgfonlayer}
\end{tikzpicture}
 \overset{\zfusion}= \begin{tikzpicture}
	\begin{pgfonlayer}{nodelayer}
		\node [style=Z dot] (2) at (0, -0.5) {};
		\node [style=idealket] (3) at (0, 0.5) {$I$};
	\end{pgfonlayer}
	\begin{pgfonlayer}{edgelayer}
		\draw (3) to (2);
	\end{pgfonlayer}
\end{tikzpicture}

    \end{equation} 
    So determining the semantic of an arbitrary \textit{closed} diagram in GAG is equivalent to $\#\mathsf{CSP}(\Gamma)$. 
\end{remark}

\section{From GAG to Qudit ZH Calculus}
\label{sec:from-GAG-to-ZH}
The ZH calculus is a graphical language for quantum computation, and it has a natural correspondence with the Toffoli-Hadamard gate set.
In this section, we elucidate the connection between algebraic geometry and the ZH calculus.
We review the Galois-Qudit ZH calculus introduced in \cite{gao_qudit_2024} and show that the ZH calculus can be obtained as a fairly minimal extension of finite-field GAG:
adding a single state and one scalar suffices to recover the ZH calculus in arbitrary dimension.

\subsection{Qudit ZH Calculus}
Here we give a brief overview of the qudit ZH Calculus as introduced in \cite{gao_qudit_2024}.

\begin{definition}
    \cite{ketkar_nonbinary_2005}
    A Galois-Qudit is a quantum system represented by a Hilbert space of dimension $q = p^t$ for some prime $p$ and some natural number $t$, with a computational basis labelled as $|j\rangle$ with $j \in \Fq$.
\end{definition}

For a finite field $\Fq$, the \textbf{Galois-Qudit ZH Calculus} over $\Fq$ is the free\footnote{No complete set of rewrite rules is known for the qudit ZH. For the sake of full generality, we work with the free dagger-PROP generated by the ZH generators.} dagger-PROP generated by three families of generators:
\begin{equation*}
    \input{./figures/zh-generators.tikz}
\end{equation*}
We will denote this free dagger-PROP as $\ZH$. 
It admits a semantics functor to the PROP $\ZHtargetMats$.
A morphism from $n$ to $m$ in $\ZHtargetMats$ is a matrix of the form $q^{\frac{k}{2}} A$, where $A$ is a $q^m \times q^n$ matrix with entries in the ring $\mathbb Z[\omega]$, $\omega$ is a primitive $p$th root of $1$, and $k$ is an integer.  
\begin{equation}
    \ZH \xrightarrow{\ZHsem} \ZHtargetMats
\end{equation}
This functor is full~\cite{gao_qudit_2024} and is defined on the generators as follows:
\begin{align*}
    \input{./figures/Z-spider.tikz} &\quad \mapsto\quad  \sum_{i \in \mathbb F_q} |i \rangle^{\otimes m} \langle i |^{\otimes n} \\
    \input{./figures/H-spider.tikz} &\quad \mapsto\quad  \frac{1}{\sqrt{q}}\sum_{\substack{j_1, ..., j_m \in \mathbb F_q \\ i_1, ..., i_n \in \mathbb F_q}} \!\!\! \omega^{\tr(j_1...j_m i_1...i_n)} | j_1 \rangle ... | j_m \rangle \langle i_1 | ... \langle i_n | \\
    \begin{tikzpicture}
	\begin{pgfonlayer}{nodelayer}
		\node [style=X phase dot] (0) at (-0.75, 0) {$j$};
		\node [style=none] (1) at (0.5, 0) {};
	\end{pgfonlayer}
	\begin{pgfonlayer}{edgelayer}
		\draw (0) to (1.center);
	\end{pgfonlayer}
\end{tikzpicture}
 &\quad \mapsto\quad  \sqrt{q}|j\rangle \quad j \in \Fq
\end{align*}
where $\tr = \tr_{q/p}$ is the field trace function.

\subsection{ZH Calculus as an Extension of GAG}
We now extend the language $\LangSpq$ for finite field GAG to cover the entirety of the Galois-Qudit ZH Calculus over $\Fq$.
\begin{definition}
    The dagger-PROP $\GagOne$ is generated by $\LangSpq$ together with the following additional generators.
    \begin{equation}
        \begin{tikzpicture}
	\begin{pgfonlayer}{nodelayer}
		\node [style=Z phase dot] (0) at (-0.75, 0) {$1$};
		\node [style=none] (1) at (0.5, 0) {};
	\end{pgfonlayer}
	\begin{pgfonlayer}{edgelayer}
		\draw (0) to (1.center);
	\end{pgfonlayer}
\end{tikzpicture}
 \qquad \qquad \begin{tikzpicture}
	\begin{pgfonlayer}{nodelayer}
		\node [style=scalar] (0) at (0.5, 0) {$q^{-\frac{1}{2}}$};
	\end{pgfonlayer}
\end{tikzpicture}

    \end{equation}
    Here, $\begin{tikzpicture}
	\begin{pgfonlayer}{nodelayer}
		\node [style=Z phase dot] (0) at (-0.75, 0) {$1$};
		\node [style=none] (1) at (0.5, 0) {};
	\end{pgfonlayer}
	\begin{pgfonlayer}{edgelayer}
		\draw (0) to (1.center);
	\end{pgfonlayer}
\end{tikzpicture}
$ is a state in the Fourier basis as we define below.
\end{definition}
\begin{definition}
    We define the semantic functor $\semZOne: \GagOne \to \ZHtargetMats$.
    For each generator of $\LangSpq$, we set $\semZOne$ to be the composition of $\semqMat$ with the canonical inclusion $\Mat_\mathbb N \hookrightarrow \ZHtargetMats$.
    For the additional generators, we define:
    \begin{align}
        \semZOne[\ \begin{tikzpicture}
	\begin{pgfonlayer}{nodelayer}
		\node [style=Z phase dot] (0) at (-0.75, 0) {$1$};
		\node [style=none] (1) at (0.5, 0) {};
	\end{pgfonlayer}
	\begin{pgfonlayer}{edgelayer}
		\draw (0) to (1.center);
	\end{pgfonlayer}
\end{tikzpicture}
 \ ] &= \quad \sum_{i \in \Fq} \omega^{\tr(i)} |i\rangle \\
        \semZOne[\ \begin{tikzpicture}
	\begin{pgfonlayer}{nodelayer}
		\node [style=scalar] (0) at (0.5, 0) {$q^{-\frac{1}{2}}$};
	\end{pgfonlayer}
\end{tikzpicture}
 \ ] &= \quad q^{-\frac{1}{2}}
    \end{align}
    Hence, the following commutes:
    \[\begin{tikzcd}
        {\GagOne} & \ZHtargetMats \\
        \LangSpq & {\Mat_{\mathbb N}}
        \arrow["\semZOne", from=1-1, to=1-2]
        \arrow[hook, from=2-1, to=1-1]
        \arrow["\semqMat"', from=2-1, to=2-2]
        \arrow[hook, from=2-2, to=1-2]
    \end{tikzcd}\]
\end{definition}

The dagger-PROP $\GagOne$ covers the entirety of the Galois-Qudit ZH Calculus.
To see this, we describe how its generators can be constructed in $\GagOne$ and show that the semantics agree.
\begin{definition}
    We have the functor $\iota_{\ZH}: \ZH \to \GagOne$ defined on the generators as follows:
    \begin{align*}
    \input{./figures/z-to-gag.tikz}\\[0.6em]
    \input{./figures/h-to-gag.tikz}\\[0.6em]
    \input{./figures/basis-to-gag.tikz}
    \end{align*}
\end{definition}
\begin{restatable}{prop}{iotaZHSound}\label{prop:iota-zh-sound}
    The functor $\iota_{\ZH}$ is consistent with the semantics of the ZH Calculus.
    In other words, the following diagram commutes:
    \[\begin{tikzcd}[ampersand replacement=\&]
        \GagOne      \& \ZHtargetMats \\
        \ZH \&
        \arrow["\semZOne", from=1-1, to=1-2]
        \arrow["\iota_{\ZH}", from=2-1, to=1-1]
        \arrow["\ZHsem"', from=2-1, to=1-2]
    \end{tikzcd}\]
\end{restatable}

We will now show that by adding two rules to $\GagOne$, we can simulate any amplitude computable in the ZH calculus using the rewrite rules of GAG and a constant overhead.
Consider the following rewrite rules.
\begin{equation}\label{eq:gag-one-copy-delete}
    \input{./figures/gag-one-copy-delete.tikz}
\end{equation}
\begin{restatable}{prop}{copyDeleteZOneSound}\label{prop:copy-delete-z1-sound}
    The rewrite rules in equation \eqref{eq:gag-one-copy-delete} are sound with respect to the semantics $\semZOne$.
\end{restatable}
By working in the quotient $\GagOne/_{\eqref{eq:gag-one-copy-delete}}$,
we can reduce any diagram in $\GagOne/_{\eqref{eq:gag-one-copy-delete}}$ to a diagram in $\LangSpq$ and a single $\begin{tikzpicture}
	\begin{pgfonlayer}{nodelayer}
		\node [style=Z phase dot] (0) at (-0.75, 0) {$1$};
		\node [style=none] (1) at (0.5, 0) {};
	\end{pgfonlayer}
	\begin{pgfonlayer}{edgelayer}
		\draw (0) to (1.center);
	\end{pgfonlayer}
\end{tikzpicture}
$, up to a finite scalar.
\begin{restatable}{prop}{pullingZOneOut}\label{lem:pulling-z1-out}
    For any diagram $D$ in $\GagOne/_{\eqref{eq:gag-one-copy-delete}}$, there exists $k\in \mathbb{N}$ and a diagram $D'$ containing only the generators of $\LangSpq$ such that
    \begin{equation}\label{fig:gag-one-rules}
        \input{./figures/gag-pulling-z1-out.tikz}
    \end{equation}
\end{restatable}

\begin{proof}
    The idea is to pull all the occurrences of $\begin{tikzpicture}
	\begin{pgfonlayer}{nodelayer}
		\node [style=Z phase dot] (0) at (-0.75, 0) {$1$};
		\node [style=none] (1) at (0.5, 0) {};
	\end{pgfonlayer}
	\begin{pgfonlayer}{edgelayer}
		\draw (0) to (1.center);
	\end{pgfonlayer}
\end{tikzpicture}
$ to the bottom of the diagram, and then apply rule $\ZOneCopy$ repeatedly to merge all the occurrences of $\begin{tikzpicture}
	\begin{pgfonlayer}{nodelayer}
		\node [style=Z phase dot] (0) at (-0.75, 0) {$1$};
		\node [style=none] (1) at (0.5, 0) {};
	\end{pgfonlayer}
	\begin{pgfonlayer}{edgelayer}
		\draw (0) to (1.center);
	\end{pgfonlayer}
\end{tikzpicture}
$ into one. The detailed proof is in Appendix~\ref{append:zh}.
\end{proof}

This shows that just a single copy of the state in the Fourier basis is sufficient to extend the expressive power of finite-field GAG to cover the entirety of the Galois-Qudit ZH Calculus.
Using this, we can compute amplitudes in the ZH calculus using finite-field GAG.
Computing an amplitude in the ZH calculus corresponds to evaluating the semantics of a scalar diagram.
We have the following result.
\begin{restatable}{theorem}{zhAmplitudeViaGAG}
    Let $D$ be a scalar diagram in the ZH calculus over $\Fq$.
    The scalar $\ZHsem[D]$ can be computed using exactly $q$ queries to an oracle $\mathcal{O}$ that evaluates the semantics of scalar diagrams in $\LangSpq$.
\end{restatable}

\section{Conclusion and Future Directions} 

We have introduced a family of diagrammatic languages, which we refer to collectively as graphical algebraic geometry (GAG). This family of languages extends the existing literature on graphical linear algebra (GLA) by introducing the nonlinear multiplication operator. We construct three examples of languages within this family, and show that they are sound, universal, and complete, for their respective semantics in (co)spans of commutative algebras and affine varieties. 

We showed that closed diagrams in GAG correspond exactly to counting constraint satisfaction problems, thereby showing that rewriting GAG diagrams into each other is \#P-hard. We show also that GAG forms the nonlinear classical ``backbone'' of the Galois-qudit ZH calculus. Indeed, we show that all quantum processes in the qudit ZH calculus can be represented by consuming a single Fourier basis state, up to a finite number of scalars, meaning that a GAG oracle would be able to simulate the ZH processes with a constant number of queries. These connections serve to illustrate the potential versatility of GAG. 

There are several directions of future work. On the theoretic side, there is an open question of constructing graphical algebraic geometry over $\mathbb R$. Indeed, for a real-closed field, one would be able to express inequality constraints of the form $g(x) \geq 0$ using the generators of GAG; but finding a complete set of rewrite rules for such languages may prove difficult, and will involve the Positivstellensatze. Such a language would have exciting applications in the verification of hybrid systems, where the safety space is often given as a semialgebraic set in $\mathbb R^n$, and dynamics are given by polynomial mappings \cite{prajna_safety_2004}. On the practical side, there are two obvious directions. The first is to explore the connection between GAG and CSP: in particular, are there existing computational techniques that can optimize GAG rewriting, or bound the complexity of rewriting between diagrams of a certain class? The second is to use the characterization of ZH as an extension of GAG to ``grok'' quantum algebraic geometry codes, in the same way graphical linear algebra has been used to ``grok'' CSS codes \cite{kissinger_phase-free_2022}.

\section*{Acknowledgements} 

We would like to thank Soichiro Fujii for pointing out initially the possibility of building the graphical language separately for the algebraic setting and for the geometric setting. We would like to thank Lia Yeh, Cole Comfort, Robert Booth, Nathan Corbyn, and Owen Lynch for many helpful discussions. DG thanks Simon Harrison for his generous support for the Wolfson Harrison UK Research Council Quantum Foundation Scholarship. RS is supported by the Clarendon Fund Scholarship. This work is additionally supported by the Engineering and Physical Sciences Research Council grant number EP/Z002230/1, ``(De)constructing quantum software (DeQS)''. 

\bibliographystyle{ACM-Reference-Format}
\bibliography{references}

\appendix

\section{Derived Rules for Scalable Notation in $\mathbb L_{CAlg_k}$} 
\label{append:scalable-rules}

\scalableRewrites* 

\begin{proof}
    We begin by proving that rewrite rule \polycp is derivable in $\mathbb L_{CAlg_k}$, that is,
    \begin{equation}
        \input{./figures/p-copy-lhs.tikz} = \input{./figures/p-copy-rhs.tikz}
    \end{equation}

    Since, by theorem \ref{prop:forward-fragment-universal-complete}, the polynomials are in correspondence with diagrams in $\mathbb L_{CAlg_k}$, we may prove this proposition by structural induction on diagrams.

    First we check that the equation holds the generators. The equation is immediate if $\begin{tikzpicture}
	\begin{pgfonlayer}{nodelayer}
		\node [style=rmat] (0) at (0, 0) {$f$};
		\node [style=none] (1) at (-1, 0) {};
		\node [style=none] (2) at (1, 0) {};
	\end{pgfonlayer}
	\begin{pgfonlayer}{edgelayer}
		\draw (1.center) to (2.center);
	\end{pgfonlayer}
\end{tikzpicture}
 = \begin{tikzpicture}
	\begin{pgfonlayer}{nodelayer}
		\node [style=none] (0) at (-1, 0) {};
		\node [style=none] (1) at (1, 0) {};
	\end{pgfonlayer}
	\begin{pgfonlayer}{edgelayer}
		\draw (0.center) to (1.center);
	\end{pgfonlayer}
\end{tikzpicture}
$. 
    
    If $\begin{tikzpicture}
	\begin{pgfonlayer}{nodelayer}
		\node [style=rmat] (0) at (0, 0) {$f$};
		\node [style=none] (1) at (-1, 0) {};
		\node [style=none] (2) at (1, 0) {};
	\end{pgfonlayer}
	\begin{pgfonlayer}{edgelayer}
		\draw (1.center) to (2.center);
	\end{pgfonlayer}
\end{tikzpicture}
 = \input{./figures/copy.tikz}$, then we have: 
    \begin{equation}
    \begin{split}
        LHS &= \input{./figures/copy-copies-1.tikz} \\
        &\overset{\zassoc}{=} \input{./figures/copy-copies-2.tikz} \\
        &\overset{\zassoc}{=} \input{./figures/copy-copies-3.tikz}  \\
        &\overset{\zsymm}{=} \input{./figures/copy-copies-4.tikz} \\
        &\overset{\zassoc}{=} \input{./figures/copy-copies-5.tikz} = RHS
    \end{split}
    \end{equation}
    If $\begin{tikzpicture}
	\begin{pgfonlayer}{nodelayer}
		\node [style=rmat] (0) at (0, 0) {$f$};
		\node [style=none] (1) at (-1, 0) {};
		\node [style=none] (2) at (1, 0) {};
	\end{pgfonlayer}
	\begin{pgfonlayer}{edgelayer}
		\draw (1.center) to (2.center);
	\end{pgfonlayer}
\end{tikzpicture}
 = \begin{tikzpicture}
	\begin{pgfonlayer}{nodelayer}
		\node [style=Z dot] (0) at (0.5, 0) {};
		\node [style=none] (1) at (-0.5, 0) {};
	\end{pgfonlayer}
	\begin{pgfonlayer}{edgelayer}
		\draw (1.center) to (0);
	\end{pgfonlayer}
\end{tikzpicture}
$, then we have (recalling that the spider with legs labelled by $0$ is the empty diagram): 
    \begin{equation}
    \begin{split}
        LHS = \begin{tikzpicture}
	\begin{pgfonlayer}{nodelayer}
		\node [style=Z dot] (0) at (0.5, 0) {};
		\node [style=none] (1) at (-0.5, 0) {};
	\end{pgfonlayer}
	\begin{pgfonlayer}{edgelayer}
		\draw (1.center) to (0);
	\end{pgfonlayer}
\end{tikzpicture}
 \overset{\zunit}{=} \begin{tikzpicture}
	\begin{pgfonlayer}{nodelayer}
		\node [style=Z dot] (0) at (0, 0) {};
		\node [style=Z dot] (1) at (1, 0.5) {};
		\node [style=Z dot] (2) at (1, -0.5) {};
		\node [style=none] (3) at (-1, 0) {};
	\end{pgfonlayer}
	\begin{pgfonlayer}{edgelayer}
		\draw (3.center) to (0);
		\draw [bend left=15] (0) to (1);
		\draw [bend right=15] (0) to (2);
	\end{pgfonlayer}
\end{tikzpicture}
 = RHS
    \end{split}
    \end{equation}
    If $\begin{tikzpicture}
	\begin{pgfonlayer}{nodelayer}
		\node [style=rmat] (0) at (0, 0) {$f$};
		\node [style=none] (1) at (-1, 0) {};
		\node [style=none] (2) at (1, 0) {};
	\end{pgfonlayer}
	\begin{pgfonlayer}{edgelayer}
		\draw (1.center) to (2.center);
	\end{pgfonlayer}
\end{tikzpicture}
 = \input{./figures/add.tikz}$, then we have: 
    \begin{equation}
        \begin{split}
            LHS = \input{./figures/add-copy.tikz} \overset{\addcp}{=} \input{./figures/copy-add.tikz} = RHS
        \end{split}
    \end{equation}
    If $\begin{tikzpicture}
	\begin{pgfonlayer}{nodelayer}
		\node [style=rmat] (0) at (0, 0) {$f$};
		\node [style=none] (1) at (-1, 0) {};
		\node [style=none] (2) at (1, 0) {};
	\end{pgfonlayer}
	\begin{pgfonlayer}{edgelayer}
		\draw (1.center) to (2.center);
	\end{pgfonlayer}
\end{tikzpicture}
 = \begin{tikzpicture}
	\begin{pgfonlayer}{nodelayer}
		\node [style=X dot] (0) at (-0.5, 0) {};
		\node [style=none] (3) at (0.5, 0) {};
	\end{pgfonlayer}
	\begin{pgfonlayer}{edgelayer}
		\draw (0) to (3.center);
	\end{pgfonlayer}
\end{tikzpicture}
$, then we have: 
    \begin{equation}
        LHS = \input{./figures/zero-copy.tikz} \overset{\zerocp}{=} \begin{tikzpicture}
	\begin{pgfonlayer}{nodelayer}
		\node [style=X dot] (0) at (-0.5, 0.5) {};
		\node [style=X dot] (1) at (-0.5, -0.5) {};
		\node [style=none] (2) at (0.5, 0.5) {};
		\node [style=none] (3) at (0.5, -0.5) {};
	\end{pgfonlayer}
	\begin{pgfonlayer}{edgelayer}
		\draw (0) to (2.center);
		\draw (1) to (3.center);
	\end{pgfonlayer}
\end{tikzpicture}
 = RHS
    \end{equation}
    If $\begin{tikzpicture}
	\begin{pgfonlayer}{nodelayer}
		\node [style=rmat] (0) at (0, 0) {$f$};
		\node [style=none] (1) at (-1, 0) {};
		\node [style=none] (2) at (1, 0) {};
	\end{pgfonlayer}
	\begin{pgfonlayer}{edgelayer}
		\draw (1.center) to (2.center);
	\end{pgfonlayer}
\end{tikzpicture}
 = \input{./figures/mult.tikz}$, then we have: 
    \begin{equation}
        LHS = \input{./figures/mult-copy.tikz} \overset{\multcp}{=} \input{./figures/copy-mult.tikz} = RHS
    \end{equation}
    If $\begin{tikzpicture}
	\begin{pgfonlayer}{nodelayer}
		\node [style=rmat] (0) at (0, 0) {$f$};
		\node [style=none] (1) at (-1, 0) {};
		\node [style=none] (2) at (1, 0) {};
	\end{pgfonlayer}
	\begin{pgfonlayer}{edgelayer}
		\draw (1.center) to (2.center);
	\end{pgfonlayer}
\end{tikzpicture}
 = \begin{tikzpicture}
	\begin{pgfonlayer}{nodelayer}
		\node [style=multop] (0) at (-0.5, 0) {};
		\node [style=none] (3) at (0.5, 0) {};
	\end{pgfonlayer}
	\begin{pgfonlayer}{edgelayer}
		\draw (0) to (3.center);
	\end{pgfonlayer}
\end{tikzpicture}
$, then we have: 
    \begin{equation}
        LHS = \input{./figures/one-copy.tikz} \overset{\onecp}{=} \begin{tikzpicture}
	\begin{pgfonlayer}{nodelayer}
		\node [style=multop] (0) at (-0.5, 0.5) {};
		\node [style=multop] (1) at (-0.5, -0.5) {};
		\node [style=none] (2) at (0.5, 0.5) {};
		\node [style=none] (3) at (0.5, -0.5) {};
	\end{pgfonlayer}
	\begin{pgfonlayer}{edgelayer}
		\draw (0) to (2.center);
		\draw (1) to (3.center);
	\end{pgfonlayer}
\end{tikzpicture}
 = RHS
    \end{equation}
    If $\begin{tikzpicture}
	\begin{pgfonlayer}{nodelayer}
		\node [style=rmat] (0) at (0, 0) {$f$};
		\node [style=none] (1) at (-1, 0) {};
		\node [style=none] (2) at (1, 0) {};
	\end{pgfonlayer}
	\begin{pgfonlayer}{edgelayer}
		\draw (1.center) to (2.center);
	\end{pgfonlayer}
\end{tikzpicture}
 = \begin{tikzpicture}
	\begin{pgfonlayer}{nodelayer}
		\node [style=kscalar] (1) at (0, 0) {$\kappa$};
		\node [style=none] (2) at (-1, 0) {};
		\node [style=none] (3) at (1, 0) {};
	\end{pgfonlayer}
	\begin{pgfonlayer}{edgelayer}
		\draw (2.center) to (3.center);
	\end{pgfonlayer}
\end{tikzpicture}
$, then we have: 
    \begin{equation}
        LHS = \input{./figures/k-copy-lhs.tikz} \overset{\kcopy}{=} \input{./figures/k-copy-rhs.tikz} = RHS
    \end{equation}
    Thus the proposition holds for all the generators of $\mathbb L_{CAlg_k}$. Now for the inductive step: suppose the proposition holds for $\begin{tikzpicture}
	\begin{pgfonlayer}{nodelayer}
		\node [style=rmat] (0) at (0, 0) {$f$};
		\node [style=none] (1) at (-1, 0) {};
		\node [style=none] (2) at (1, 0) {};
		\node [style=none] (3) at (-1.5, 0) {$n$};
		\node [style=none] (4) at (1.5, 0) {$m$};
	\end{pgfonlayer}
	\begin{pgfonlayer}{edgelayer}
		\draw (1.center) to (2.center);
	\end{pgfonlayer}
\end{tikzpicture}
$ and $\begin{tikzpicture}
	\begin{pgfonlayer}{nodelayer}
		\node [style=rmat] (0) at (0, 0) {$g$};
		\node [style=none] (1) at (-1, 0) {};
		\node [style=none] (2) at (1, 0) {};
		\node [style=none] (3) at (-1, 0.5) {$n'$};
		\node [style=none] (4) at (1, 0.5) {$m'$};
	\end{pgfonlayer}
	\begin{pgfonlayer}{edgelayer}
		\draw (1.center) to (2.center);
	\end{pgfonlayer}
\end{tikzpicture}
$. Then, in the case $m = n'$, consider the diagram $\begin{tikzpicture}
	\begin{pgfonlayer}{nodelayer}
		\node [style=rmat] (0) at (-1, 0) {$f$};
		\node [style=rmat] (1) at (1, 0) {$g$};
		\node [style=none] (2) at (-2, 0) {};
		\node [style=none] (3) at (2, 0) {};
	\end{pgfonlayer}
	\begin{pgfonlayer}{edgelayer}
		\draw (2.center) to (3.center);
	\end{pgfonlayer}
\end{tikzpicture}
$. we have 
    \begin{equation}
        \begin{split}
            LHS &= \input{./figures/poly-comp-copy-1.tikz} \\
            &\overset{\text{I.H.}}{=} \input{./figures/poly-comp-copy-2.tikz} \\
            &\overset{\text{I.H.}}{=} \input{./figures/poly-comp-copy-3.tikz} = RHS
        \end{split}
    \end{equation}
    Finally, for the diagram $\input{./figures/poly-tensor.tikz}$ we have 
    \begin{equation}
        LHS = \input{./figures/poly-tensor-copy-1.tikz} \overset{\text{I.H.}}{=} \input{./figures/poly-tensor-copy-2.tikz} = RHS
    \end{equation}
    Here, I.H. denotes the inductive hypothesis. This completes the induction. 

    We next prove that rule \polydel is derivable in $\mathbb L_{CAlg_k}$, that is,
    \begin{equation}
        \begin{tikzpicture}
	\begin{pgfonlayer}{nodelayer}
		\node [style=rmat] (0) at (0, 0) {$f$};
		\node [style=Z dot] (1) at (1, 0) {};
		\node [style=none] (2) at (-1, 0) {};
	\end{pgfonlayer}
	\begin{pgfonlayer}{edgelayer}
		\draw (2.center) to (1);
	\end{pgfonlayer}
\end{tikzpicture}
 = \begin{tikzpicture}
	\begin{pgfonlayer}{nodelayer}
		\node [style=Z dot] (0) at (0.5, 0) {};
		\node [style=none] (1) at (-0.5, 0) {};
	\end{pgfonlayer}
	\begin{pgfonlayer}{edgelayer}
		\draw (1.center) to (0);
	\end{pgfonlayer}
\end{tikzpicture}

    \end{equation}

    We again use structural induction on the diagrams of $\mathbb L_{CAlg_k}$. The equation is tautological if either $\begin{tikzpicture}
	\begin{pgfonlayer}{nodelayer}
		\node [style=rmat] (0) at (0, 0) {$f$};
		\node [style=none] (1) at (-1, 0) {};
		\node [style=none] (2) at (1, 0) {};
	\end{pgfonlayer}
	\begin{pgfonlayer}{edgelayer}
		\draw (1.center) to (2.center);
	\end{pgfonlayer}
\end{tikzpicture}
 = \begin{tikzpicture}
	\begin{pgfonlayer}{nodelayer}
		\node [style=none] (0) at (-1, 0) {};
		\node [style=none] (1) at (1, 0) {};
	\end{pgfonlayer}
	\begin{pgfonlayer}{edgelayer}
		\draw (0.center) to (1.center);
	\end{pgfonlayer}
\end{tikzpicture}
$ or $\begin{tikzpicture}
	\begin{pgfonlayer}{nodelayer}
		\node [style=rmat] (0) at (0, 0) {$f$};
		\node [style=none] (1) at (-1, 0) {};
		\node [style=none] (2) at (1, 0) {};
	\end{pgfonlayer}
	\begin{pgfonlayer}{edgelayer}
		\draw (1.center) to (2.center);
	\end{pgfonlayer}
\end{tikzpicture}
 = \begin{tikzpicture}
	\begin{pgfonlayer}{nodelayer}
		\node [style=Z dot] (0) at (0.5, 0) {};
		\node [style=none] (1) at (-0.5, 0) {};
	\end{pgfonlayer}
	\begin{pgfonlayer}{edgelayer}
		\draw (1.center) to (0);
	\end{pgfonlayer}
\end{tikzpicture}
$ (recalling that the spider with legs labelled by $0$ is the empty diagram). 

    If $\begin{tikzpicture}
	\begin{pgfonlayer}{nodelayer}
		\node [style=rmat] (0) at (0, 0) {$f$};
		\node [style=none] (1) at (-1, 0) {};
		\node [style=none] (2) at (1, 0) {};
	\end{pgfonlayer}
	\begin{pgfonlayer}{edgelayer}
		\draw (1.center) to (2.center);
	\end{pgfonlayer}
\end{tikzpicture}
 = \input{./figures/copy.tikz}$, we have 
    \begin{equation}
        LHS = \begin{tikzpicture}
	\begin{pgfonlayer}{nodelayer}
		\node [style=Z dot] (0) at (0, 0) {};
		\node [style=Z dot] (1) at (1, 0.5) {};
		\node [style=Z dot] (2) at (1, -0.5) {};
		\node [style=none] (3) at (-1, 0) {};
	\end{pgfonlayer}
	\begin{pgfonlayer}{edgelayer}
		\draw (3.center) to (0);
		\draw [bend left=15] (0) to (1);
		\draw [bend right=15] (0) to (2);
	\end{pgfonlayer}
\end{tikzpicture}
 \overset{\zunit}{=} \begin{tikzpicture}
	\begin{pgfonlayer}{nodelayer}
		\node [style=Z dot] (0) at (0.5, 0) {};
		\node [style=none] (1) at (-0.5, 0) {};
	\end{pgfonlayer}
	\begin{pgfonlayer}{edgelayer}
		\draw (1.center) to (0);
	\end{pgfonlayer}
\end{tikzpicture}
 = RHS
    \end{equation}
    If $\begin{tikzpicture}
	\begin{pgfonlayer}{nodelayer}
		\node [style=rmat] (0) at (0, 0) {$f$};
		\node [style=none] (1) at (-1, 0) {};
		\node [style=none] (2) at (1, 0) {};
	\end{pgfonlayer}
	\begin{pgfonlayer}{edgelayer}
		\draw (1.center) to (2.center);
	\end{pgfonlayer}
\end{tikzpicture}
 = \input{./figures/add.tikz}$ or $\begin{tikzpicture}
	\begin{pgfonlayer}{nodelayer}
		\node [style=X dot] (0) at (-0.5, 0) {};
		\node [style=none] (3) at (0.5, 0) {};
	\end{pgfonlayer}
	\begin{pgfonlayer}{edgelayer}
		\draw (0) to (3.center);
	\end{pgfonlayer}
\end{tikzpicture}
$ or $\input{./figures/mult.tikz}$ or $\begin{tikzpicture}
	\begin{pgfonlayer}{nodelayer}
		\node [style=multop] (0) at (-0.5, 0) {};
		\node [style=none] (3) at (0.5, 0) {};
	\end{pgfonlayer}
	\begin{pgfonlayer}{edgelayer}
		\draw (0) to (3.center);
	\end{pgfonlayer}
\end{tikzpicture}
$ or $\begin{tikzpicture}
	\begin{pgfonlayer}{nodelayer}
		\node [style=kscalar] (1) at (0, 0) {$\kappa$};
		\node [style=none] (2) at (-1, 0) {};
		\node [style=none] (3) at (1, 0) {};
	\end{pgfonlayer}
	\begin{pgfonlayer}{edgelayer}
		\draw (2.center) to (3.center);
	\end{pgfonlayer}
\end{tikzpicture}
$, the desired proposition is just a straightforward application of axioms \adddel, \zerodel, \multdel, \onedel, \kdel respectively. So the proposition holds for all generators of $\mathbb L_{CAlg_k}$. Now for the inductive step: suppose the proposition holds for $\begin{tikzpicture}
	\begin{pgfonlayer}{nodelayer}
		\node [style=rmat] (0) at (0, 0) {$f$};
		\node [style=none] (1) at (-1, 0) {};
		\node [style=none] (2) at (1, 0) {};
		\node [style=none] (3) at (-1.5, 0) {$n$};
		\node [style=none] (4) at (1.5, 0) {$m$};
	\end{pgfonlayer}
	\begin{pgfonlayer}{edgelayer}
		\draw (1.center) to (2.center);
	\end{pgfonlayer}
\end{tikzpicture}
$ and $\begin{tikzpicture}
	\begin{pgfonlayer}{nodelayer}
		\node [style=rmat] (0) at (0, 0) {$g$};
		\node [style=none] (1) at (-1, 0) {};
		\node [style=none] (2) at (1, 0) {};
		\node [style=none] (3) at (-1, 0.5) {$n'$};
		\node [style=none] (4) at (1, 0.5) {$m'$};
	\end{pgfonlayer}
	\begin{pgfonlayer}{edgelayer}
		\draw (1.center) to (2.center);
	\end{pgfonlayer}
\end{tikzpicture}
$. Then, in the case $m = n'$, consider the diagram $\begin{tikzpicture}
	\begin{pgfonlayer}{nodelayer}
		\node [style=rmat] (0) at (-1, 0) {$f$};
		\node [style=rmat] (1) at (1, 0) {$g$};
		\node [style=none] (2) at (-2, 0) {};
		\node [style=none] (3) at (2, 0) {};
	\end{pgfonlayer}
	\begin{pgfonlayer}{edgelayer}
		\draw (2.center) to (3.center);
	\end{pgfonlayer}
\end{tikzpicture}
$. we have 
    \begin{equation}
        LHS = \begin{tikzpicture}
	\begin{pgfonlayer}{nodelayer}
		\node [style=rmat] (0) at (-0.75, 0) {$f$};
		\node [style=rmat] (1) at (0.5, 0) {$g$};
		\node [style=Z dot] (2) at (1.5, 0) {};
		\node [style=none] (3) at (-1.5, 0) {};
	\end{pgfonlayer}
	\begin{pgfonlayer}{edgelayer}
		\draw (3.center) to (2);
	\end{pgfonlayer}
\end{tikzpicture}
 \overset{\text{I.H.}}{=} \begin{tikzpicture}
	\begin{pgfonlayer}{nodelayer}
		\node [style=rmat] (0) at (-0.5, 0) {$f$};
		\node [style=Z dot] (2) at (0.5, 0) {};
		\node [style=none] (3) at (-1.25, 0) {};
	\end{pgfonlayer}
	\begin{pgfonlayer}{edgelayer}
		\draw (3.center) to (2);
	\end{pgfonlayer}
\end{tikzpicture}
 \overset{\text{I.H.}}{=} \begin{tikzpicture}
	\begin{pgfonlayer}{nodelayer}
		\node [style=Z dot] (0) at (0.5, 0) {};
		\node [style=none] (1) at (-0.5, 0) {};
	\end{pgfonlayer}
	\begin{pgfonlayer}{edgelayer}
		\draw (1.center) to (0);
	\end{pgfonlayer}
\end{tikzpicture}
 = RHS
    \end{equation}
    Finally, for the diagram $\input{./figures/poly-tensor.tikz}$ we have 
    \begin{equation}
        LHS = \input{./figures/poly-tensor-delete-1.tikz} \overset{\text{I.H.}}{=} \begin{tikzpicture}
	\begin{pgfonlayer}{nodelayer}
		\node [style=Z dot] (0) at (0.5, 0.5) {};
		\node [style=Z dot] (1) at (0.5, -0.5) {};
		\node [style=none] (2) at (-0.5, 0.5) {};
		\node [style=none] (3) at (-0.5, -0.5) {};
	\end{pgfonlayer}
	\begin{pgfonlayer}{edgelayer}
		\draw (2.center) to (0);
		\draw (3.center) to (1);
	\end{pgfonlayer}
\end{tikzpicture}
 = RHS
    \end{equation}
    This completes the induction. 

    Finally, the rules \polyadd, \polymult, and \polycomp are direct consequences of the fact that $\mathbb L_{CAlg_k}$ is the Lawvere theory of commutative algebras over $k$. 
\end{proof}

\section{Proofs about $\LangCsp$}
\label{append:csp-proofs}

\subsection{Proofs about Soundness}
\cspsound*

\begin{proof}
    First to prove \kinv: for $l \in k^{\times}$,
    \begin{equation}
        \begin{split}
            \left[\begin{tikzpicture}
	\begin{pgfonlayer}{nodelayer}
		\node [style=kscalar] (0) at (0, 0) {$l$};
		\node [style=none] (1) at (-1.5, 0) {};
		\node [style=none] (2) at (1.5, 0) {};
	\end{pgfonlayer}
	\begin{pgfonlayer}{edgelayer}
		\draw (1.center) to (2.center);
	\end{pgfonlayer}
\end{tikzpicture}
\right] &= \left(k[x] = k[x] \xleftarrow{y \mapsto lx} k[y]\right) \\
            &= \left(k[x] \xrightarrow{x \mapsto l^{-1}y} k[y] = k[y]\right) \\
            &= \left[\begin{tikzpicture}
	\begin{pgfonlayer}{nodelayer}
		\node [style=kscalarop] (0) at (0, 0) {$l^{-1}$};
		\node [style=none] (1) at (-1.5, 0) {};
		\node [style=none] (2) at (1.5, 0) {};
	\end{pgfonlayer}
	\begin{pgfonlayer}{edgelayer}
		\draw (1.center) to (2.center);
	\end{pgfonlayer}
\end{tikzpicture}
\right]
        \end{split}
    \end{equation}
    where the second equality is the span isomorphism
    % https://q.uiver.app/#q=WzAsNCxbMSwwLCJrW3hdIl0sWzAsMSwia1t4XSJdLFsyLDEsImtbeV0iXSxbMSwyLCJrW3ldIl0sWzAsMSwiIiwwLHsic3R5bGUiOnsiaGVhZCI6eyJuYW1lIjoibm9uZSJ9fX1dLFsxLDMsImxeey0xfXkiLDJdLFszLDIsIiIsMCx7InN0eWxlIjp7ImhlYWQiOnsibmFtZSI6Im5vbmUifX19XSxbMiwwLCJseCIsMl0sWzAsMywibF57LTF9eSIsMV1d
    \[\begin{tikzcd}
        & {k[x]} \\
        {k[x]} && {k[y]} \\
        & {k[y]}
        \arrow[equals, from=1-2, to=2-1]
        \arrow["{l^{-1}y}"{description}, from=1-2, to=3-2]
        \arrow["{l^{-1}y}"', from=2-1, to=3-2]
        \arrow["lx"', from=2-3, to=1-2]
        \arrow[equals, from=3-2, to=2-3]
    \end{tikzcd}\]
    Next to prove \xbone:
    \begin{equation}
        \begin{split}
            \left[\begin{tikzpicture}
	\begin{pgfonlayer}{nodelayer}
		\node [style=X dot] (0) at (-0.5, 0) {};
		\node [style=X dot] (1) at (0.5, 0) {};
	\end{pgfonlayer}
	\begin{pgfonlayer}{edgelayer}
		\draw (0) to (1);
	\end{pgfonlayer}
\end{tikzpicture}
\right] &= \left(k = k \xleftarrow{x\mapsto 0} k[x] \xrightarrow{x\mapsto 0} k = k\right) \\ &= \left(k = k = k\right) = \left[\begin{tikzpicture}
	\begin{pgfonlayer}{nodelayer}
		\node [style=empty diagram] (0) at (0, 0) {};
	\end{pgfonlayer}
\end{tikzpicture}
\right]
        \end{split}
    \end{equation}
    here the second equality is the pushout
    % https://q.uiver.app/#q=WzAsNCxbMSwwLCJrIl0sWzAsMSwiayJdLFsyLDEsImsiXSxbMSwyLCJrW3hdIl0sWzMsMSwieFxcbWFwc3RvIDAiXSxbMCwxLCIiLDAseyJzdHlsZSI6eyJoZWFkIjp7Im5hbWUiOiJub25lIn19fV0sWzAsMiwiIiwyLHsic3R5bGUiOnsiaGVhZCI6eyJuYW1lIjoibm9uZSJ9fX1dLFswLDMsIiIsMSx7InN0eWxlIjp7Im5hbWUiOiJjb3JuZXIifX1dLFszLDIsInhcXG1hcHN0byAwIiwyXV0=
    \[\begin{tikzcd}
        & k \\
        k && k \\
        & {k[x]}
        \arrow[equals, from=1-2, to=2-1]
        \arrow[equals, from=1-2, to=2-3]
        \arrow["\lrcorner"{anchor=center, pos=0.125, rotate=-45}, draw=none, from=1-2, to=3-2]
        \arrow["{x\mapsto 0}", from=3-2, to=2-1]
        \arrow["{x\mapsto 0}"', from=3-2, to=2-3]
    \end{tikzcd}\]

    Next to prove\zfrob:
    \begin{equation}
        \begin{split}
            &\left[\input{./figures/z-frob-lhs.tikz}\right] \\
            &= \left(k[x_1, x_2] \xleftarrow{\substack{x_1 \leftmapsto y_1 \\ x_1 \leftmapsto y_2 \\ x_2 \leftmapsto y_3}} k[y_1, y_2, y_3] \xrightarrow{\substack{y_1 \mapsto z_1 \\ y_2 \mapsto z_2 \\ y_3 \mapsto z_2}} k[z_1, z_2] \right) \\
        &= \left(k[x_1, x_2] \rightarrow \frac{k[x_1, x_2, z_1, z_2]}{(x_1 - z_1, x_1 - z_2, x_2 - z_2)} \leftarrow k[z_1, z_2]\right)\\
        &= \left(k[x_1, x_2] \xrightarrow{\substack{x_1 \mapsto y \\ x_2 \mapsto y}} k[y] \xleftarrow{\substack{y \leftmapsto z_1 \\ y \leftmapsto z_2}} k[z_1, z_2]\right) = \left[\input{./figures/z-frob-mid.tikz}\right] \\
        &= \left(k[x_1, x_2] \xleftarrow{\substack{x_1 \leftmapsto y_1 \\ x_2 \leftmapsto y_2 \\ x_2 \leftmapsto y_3}} k[y_1, y_2, y_3] \xrightarrow{\substack{y_1 \mapsto z_1 \\ y_2 \mapsto z_1 \\ y_3 \mapsto z_2}} k[z_1, z_2]\right) \\
        &= \left[\input{./figures/z-frob-rhs.tikz}\right]\\
        \end{split}
    \end{equation}

    Next for the soundness of \xfrob: first, by a pushout we obtain that
    \begin{equation}
        \label{eq:xfrob-soundness-1}
        \begin{split}
            &\left[\input{./figures/x-frob-mid.tikz}\right] \\
            &= \left(k[x_1, x_2] \xleftarrow{x_1 + x_2 \leftmapsto y} k[y] \xrightarrow{y \mapsto z_1 + z_2} k[z_1, z_2]\right) \\
            &= \left(k[x_1, x_2] \rightarrow \frac{k[x_1, x_2, z_1, z_2]}{(x_1 + x_2 - z_1 - z_2)} \leftarrow k[z_1, z_2]\right)
        \end{split}
    \end{equation}
    Now consider the apex of the cospan in the final line of \ref{eq:xfrob-soundness-1}. We can parametrize it in two ways:
    % https://q.uiver.app/#q=WzAsNSxbMSwwLCJrW3VfMSwgdV8yLCB1XzNdIl0sWzEsMSwiXFxmcmFje2tbeF8xLCB4XzIsIHpfMSwgel8yXX17KHhfMSArIHhfMiAtIHpfMSAtIHpfMil9Il0sWzAsMSwia1t4XzEsIHhfMl0iXSxbMiwxLCJrW3pfMSwgel8yXSJdLFsxLDIsImtbdl8xLCB2XzIsIHZfM10iXSxbMSwwLCJcXGFscGhhIl0sWzEsNCwiXFxiZXRhIiwyXSxbMiwwLCJcXHN1YnN0YWNre3hfMSBcXG1hcHN0byB1XzEgKyB1XzIgXFxcXCB4XzIgXFxtYXBzdG8gdV8zfSJdLFsyLDFdLFsyLDQsIlxcc3Vic3RhY2t7eF8xIFxcbWFwc3RvIHZfMSBcXFxcIHhfMiBcXG1hcHN0byB2XzIgKyB2XzN9IiwyXSxbMywxXSxbMywwLCJcXHN1YnN0YWNre3pfMSBcXG1hcHN0byB1XzEgXFxcXCB6XzIgXFxtYXBzdG8gdV8yICsgdV8zfSIsMl0sWzMsNCwiXFxzdWJzdGFja3t6XzEgXFxtYXBzdG8gdl8xICsgdl8yIFxcXFwgel8yIFxcbWFwc3RvIHZfM30iXV0=
    \[\begin{tikzcd}
        & {k[u_1, u_2, u_3]} \\
        {k[x_1, x_2]} & {\frac{k[x_1, x_2, z_1, z_2]}{(x_1 + x_2 - z_1 - z_2)}} & {k[z_1, z_2]} \\
        & {k[v_1, v_2, v_3]}
        \arrow["\begin{array}{c} \substack{x_1 \mapsto u_1 + u_2 \\ x_2 \mapsto u_3} \end{array}", from=2-1, to=1-2]
        \arrow[from=2-1, to=2-2]
        \arrow["\begin{array}{c} \substack{x_1 \mapsto v_1 \\ x_2 \mapsto v_2 + v_3} \end{array}"', from=2-1, to=3-2]
        \arrow["\alpha", from=2-2, to=1-2]
        \arrow["\beta"', from=2-2, to=3-2]
        \arrow["\begin{array}{c} \substack{z_1 \mapsto u_1 \\ z_2 \mapsto u_2 + u_3} \end{array}"', from=2-3, to=1-2]
        \arrow[from=2-3, to=2-2]
        \arrow["\begin{array}{c} \substack{z_1 \mapsto v_1 + v_2 \\ z_2 \mapsto v_3} \end{array}", from=2-3, to=3-2]
    \end{tikzcd}\]
    where $\alpha, \beta$ are the obvious maps that make the above diagram commute: 
    \begin{equation}
        \alpha:
        \begin{split}
            &x_1 \mapsto u_1 + u_2 \\
            &x_2 \mapsto u_3 \\
            &z_1 \mapsto u_1 \\
            &z_2 \mapsto u_2 + u_3
        \end{split}
        \qquad
        \beta:
        \begin{split}
            &x_1 \mapsto v_1 \\
            &x_2 \mapsto v_2 + v_3 \\
            &z_1 \mapsto v_1 + v_2 \\
            &z_2 \mapsto v_3
        \end{split}
    \end{equation}
    Which are both isomorphisms of algebras. This gives:
    \begin{equation}
        \begin{split}
            &\left[\input{./figures/x-frob-mid.tikz}\right] \\
            &= \left(k[x_1, x_2] \rightarrow \frac{k[x_1, x_2, z_1, z_2]}{(x_1 + x_2 - z_1 - z_2)} \leftarrow k[z_1, z_2]\right) \\
            &= \left(k[x_1, x_2] \xrightarrow{\substack{x_1 \mapsto u_1 + u_2 \\ x_2 \mapsto u_3}} k[u_1, u_2, u_3] \xleftarrow{\substack{u_1 \leftmapsto z_1 \\ u_2 + u_3 \leftmapsto z_2}} k[z_1, z_2] \right) \\
            &= \left[\input{./figures/x-frob-rhs.tikz}\right] \\
            &= \left(k[x_1, x_2] \xrightarrow{\substack{x_1 \mapsto v_1 \\ x_2 \mapsto v_2 + v_3}} k[v_1, v_2, v_3] \xleftarrow{\substack{v_1 + v_2 \leftmapsto z_1 \\ v_3 \leftmapsto z_2}} k[z_1, z_2]\right) \\
            &= \left[\input{./figures/x-frob-lhs.tikz}\right]
        \end{split}
    \end{equation}

    Now for the soundness of \cccup:
    \begin{equation}
        \begin{split}
            &\left[\input{./figures/x-cup.tikz}\right] \\
            &= \left(k \xleftarrow{0 \leftmapsto x} k[x] \xrightarrow{x \mapsto y_1 - y_2} k[y_1, y_2]\right) \\
            &= \left(k \rightarrow \frac{k[y_1, y_2]}{(y_1 - y_2)} \leftarrow k[y_1, y_2]\right)\\
            &= \left(k \rightarrow k[x] \xleftarrow{\substack{x \leftmapsto y_1\\ x \leftmapsto y_2}} k[y_1, y_2]\right) \\
            &= \left[\input{./figures/z-cup.tikz}\right]
        \end{split}
    \end{equation}

    The proof for the soundness of \cccap is almost identical:
    \begin{equation}
        \begin{split}
            &\left[\input{./figures/x-cap.tikz}\right] \\
            &= \left(k[x_1, x_2] \xleftarrow{x_1 - x_2 \leftmapsto y} k[y] \xrightarrow{y \mapsto 0} k\right) \\
            &=\left(k[x_1, x_2] \rightarrow \frac{k[x_1, x_2]}{(x_1 - x_2)} \leftarrow k\right) \\
            &= \left(k[x_1, x_2] \xrightarrow{\substack{x_1 \mapsto y \\ x_2 \mapsto y}} k[y] \leftarrow k\right) = \left[\input{./figures/z-cap.tikz}\right]
        \end{split}
    \end{equation}

    Now for the soundness of \ktrp:
    \begin{equation}
        \begin{split}
            &\left[\input{./figures/k-trans.tikz}\right] \\
            &= \left(k[x_2] \rightarrow k\left[\substack{x_1\\x_2}\right] \xleftarrow{\substack{x_1 \leftmapsto y_1 \\ \kappa x_1 \leftmapsto y_2 \\ x_2 \leftmapsto y_3}} k\left[\substack{y_1\\y_2\\y_3}\right] \xrightarrow{\substack{y_1 \mapsto z_1 \\ y_2 \mapsto z_2 \\ y_3 \mapsto z_2}} k\left[\substack{z_1\\z_2}\right] \leftarrow k[z_1]\right) \\
            &= \left(k[x_2] \rightarrow \frac{k[x_1, x_2, z_1, z_2]}{(x_1 - z_1, \kappa x_1 - z_2, x_2 - z_2)} \leftarrow k[z_1] \right) \\
            &= \left(k[x_2] \xrightarrow{x_2 \mapsto \kappa z_1} k[z_1] \leftarrow k[z_1]\right) = \left[\begin{tikzpicture}
	\begin{pgfonlayer}{nodelayer}
		\node [style=kscalarop] (0) at (0, 0) {$\kappa$};
		\node [style=none] (1) at (-1.5, 0) {};
		\node [style=none] (2) at (1.5, 0) {};
	\end{pgfonlayer}
	\begin{pgfonlayer}{edgelayer}
		\draw (1.center) to (2.center);
	\end{pgfonlayer}
\end{tikzpicture}
\right]
        \end{split}
    \end{equation}

    For the soundness of \mtrp:
    \begin{equation}
        \begin{split}
            &\left[\input{./figures/mult-trans.tikz}\right] \\
            &= \left(k[x_3] \rightarrow k\left[\substack{x_1 \\x_2 \\x_3}\right] \xleftarrow{\substack{x_1 \leftmapsto y_1 \\ x_1 x_2 \leftmapsto y_2 \\ x_3 \leftmapsto y_3 \\ x_2 \leftmapsto y_4}} k \left[\substack{y_1 \\ y_2 \\y_3 \\ y_4}\right] \xrightarrow{\substack{y_1 \mapsto z_1 \\ y_2 \mapsto z_2 \\ y_3 \mapsto z_2 \\ y_4 \mapsto z_3}} k\left[\substack{z_1 \\ z_2 \\ z_3}\right] \leftarrow k\left[\substack{z_1 \\ z_3}\right]\right) \\
            &= \left(k[x_3] \rightarrow \frac{k[x_1, x_2, x_3, z_1, z_2, z_3]}{(x_1 - z_1, x_1 x_2 - z_2, x_3 - z_2, x_2 - z_3)} \leftarrow k\left[\substack{z_1 \\ z_3}\right]\right) \\
            &= \left(k[x_3] \xrightarrow{x_3 \mapsto z_1 z_3} k[z_1, z_3] \leftarrow k[z_1, z_3]\right) \\
            &= \left[\input{./figures/mult-op.tikz}\right]
        \end{split}
    \end{equation}

    For the soundness of \onetrp:
    \begin{equation}
        \begin{split}
            &\left[\input{./figures/one-trans.tikz}\right] \\
            &= \left(k[x] \xleftarrow{\substack{x \leftmapsto y_1 \\ 1 \leftmapsto y_2}} k[y_1, y_2] \xrightarrow{\substack{y_1 \mapsto z \\ y_2 \mapsto z}} k[z] \leftarrow k\right) \\
            &= \left(k[x] \rightarrow \frac{k[x, z]}{(x - z, 1 - z)} \leftarrow k\right) \\
            &= \left(k[x] \xrightarrow{x \mapsto 1} k \leftarrow k\right) = \left[\begin{tikzpicture}
	\begin{pgfonlayer}{nodelayer}
		\node [style=mult] (0) at (0.5, 0) {};
		\node [style=none] (1) at (-0.5, 0) {};
	\end{pgfonlayer}
	\begin{pgfonlayer}{edgelayer}
		\draw (1.center) to (0);
	\end{pgfonlayer}
\end{tikzpicture}
\right]
        \end{split}
    \end{equation}

    % And finally, for the soundness of \ref{ax:reduced}:
    % \begin{equation}
    % \begin{split}
    %     \left[\tikzfig{nilpotent-lhs}\right] &= \left(k = k \overset{x^2}{\longrightarrow} k \overset{0}{\longleftarrow} *= *\right)\\
    %     &= \left(k \overset{x}{\longleftarrow} \{x\in k:x^2 = 0\} \overset{!}{\longrightarrow} *\right) \\
    %     &= \left(k \overset{0}{\longleftarrow} * = *\right) = \left[\tikzfig{zero-op}\right]
    % \end{split}
    % \end{equation}
    This completes the proof that every rewrite rule contained in Figure~\ref{fig:rules-for-csp} is sound with respect to the semantics functor $\sem: \LangCsp \to \FrCsp$.
\end{proof}

\subsection{Proofs of Derived Rules}

\derivedRulesOfCsp*

\begin{proof}
    First we show how to derive \ref{eq:z-special}. The reasoning in this proof is almost identical to the reasoning in the proofs contained in Appendix A.2.2 of \cite{zanasi_interacting_2018}. We first make the following derivation: for any $\kappa \in k$,
    \begin{equation}
        \label{eq:knot-yank}
        \begin{split}
            &\input{./figures/knot-yank-lhs.tikz} \overset{\ktrp}{=} \input{./figures/knot-yank-1.tikz} \\
            &\overset{\substack{\zfrob\\ \zunit}}{=} \input{./figures/knot-yank-2.tikz} \overset{\cccup}{=} \input{./figures/knot-yank-3.tikz} \\
            &\overset{\kassoc}{=} \input{./figures/knot-yank-4.tikz} \overset{\xfrob}{=} \input{./figures/knot-yank-rhs.tikz}
        \end{split}
    \end{equation}
    and with this we can make the desired derivation:
    \begin{equation}
        \begin{split}
            &\begin{tikzpicture}
	\begin{pgfonlayer}{nodelayer}
		\node [style=none] (0) at (-1, 0) {};
		\node [style=none] (1) at (1, 0) {};
	\end{pgfonlayer}
	\begin{pgfonlayer}{edgelayer}
		\draw (0.center) to (1.center);
	\end{pgfonlayer}
\end{tikzpicture}
 \overset{\xbone}{=} \begin{tikzpicture}
	\begin{pgfonlayer}{nodelayer}
		\node [style=none] (0) at (-1, 0.5) {};
		\node [style=none] (1) at (1, 0.5) {};
		\node [style=X dot] (2) at (-0.5, -0.5) {};
		\node [style=X dot] (3) at (0.5, -0.5) {};
	\end{pgfonlayer}
	\begin{pgfonlayer}{edgelayer}
		\draw (0.center) to (1.center);
		\draw (2) to (3);
	\end{pgfonlayer}
\end{tikzpicture}
 \overset{\zunit}{=} \input{./figures/z-special-2.tikz} \\
            &\overset{\zeroInK}{=} \input{./figures/z-special-3.tikz} \overset{\kadd}{=} \input{./figures/z-special-4.tikz} \\
            &\overset{\zassoc}{=} \input{./figures/z-special-5.tikz} \\
            &\overset{\ref{eq:knot-yank}}{=} \input{./figures/z-special-6.tikz} \\
            &\overset{\substack{\zassoc \\ \xassoc}}{=} \input{./figures/z-special-7.tikz} \\
            &\overset{\substack{\kadd\\ \zeroInK}}{=} \input{./figures/z-special-8.tikz} \\
            &\overset{\substack{\zunit \\ \xunit}}{=} \input{./figures/z-special-lhs.tikz}
        \end{split}
    \end{equation}
    as required.

    Next we show how to derive \ref{eq:poly-trans} by structural induction on $\mathbb L_{CAlg_k}$. First we show that the proposition holds for $\begin{tikzpicture}
	\begin{pgfonlayer}{nodelayer}
		\node [style=rmat] (0) at (0, 0) {$f$};
		\node [style=none] (1) at (-1, 0) {};
		\node [style=none] (2) at (1, 0) {};
	\end{pgfonlayer}
	\begin{pgfonlayer}{edgelayer}
		\draw (1.center) to (2.center);
	\end{pgfonlayer}
\end{tikzpicture}
 = \begin{tikzpicture}
	\begin{pgfonlayer}{nodelayer}
		\node [style=none] (0) at (-1, -0.25) {};
		\node [style=none] (1) at (1, -0.25) {};
		\node [style=none] (2) at (-1, 0.25) {$n$};
	\end{pgfonlayer}
	\begin{pgfonlayer}{edgelayer}
		\draw (0.center) to (1.center);
	\end{pgfonlayer}
\end{tikzpicture}
$. When $n = 1$ this is easily derived as:
    \begin{equation}
        \label{eq:id-yank}
        LHS = \input{./figures/id-yank-lhs.tikz} \overset{\zfrob}{=} \input{./figures/id-yank-1.tikz} \overset{\zunit}{=} \begin{tikzpicture}
	\begin{pgfonlayer}{nodelayer}
		\node [style=none] (0) at (-1, 0) {};
		\node [style=none] (1) at (1, 0) {};
	\end{pgfonlayer}
	\begin{pgfonlayer}{edgelayer}
		\draw (0.center) to (1.center);
	\end{pgfonlayer}
\end{tikzpicture}
 = RHS
    \end{equation}
    whereas for $n > 1$ we can induct on $n$:
    \begin{equation}
        \label{eq:n-id-yank}
        \begin{split}
            &LHS = \input{./figures/n-id-yank-lhs.tikz} \overset{\ref{fig:scalable-notation}}{=} \input{./figures/n-id-yank-1.tikz} \\
            &= \input{./figures/n-id-yank-2.tikz} \\
            &\overset{\text{I.H.}}{=} \input{./figures/n-id-yank-3.tikz} \\
            &= \input{./figures/n-id-yank-4.tikz} \\
            &\overset{\ref{fig:scalable-notation}}{=} \begin{tikzpicture}
	\begin{pgfonlayer}{nodelayer}
		\node [style=none] (0) at (-1, -0.25) {};
		\node [style=none] (1) at (1, -0.25) {};
		\node [style=none] (2) at (-1, 0.25) {$n$};
	\end{pgfonlayer}
	\begin{pgfonlayer}{edgelayer}
		\draw (0.center) to (1.center);
	\end{pgfonlayer}
\end{tikzpicture}
 = RHS
        \end{split}
    \end{equation}
    where the equality labelled ``I.H.'' is an application of the inductive hypothesis.

    We then check that the proposition holds for $f$ ranging over every generator of $\mathbb L_{CAlg_k}$.

    If $\begin{tikzpicture}
	\begin{pgfonlayer}{nodelayer}
		\node [style=rmat] (0) at (0, 0) {$f$};
		\node [style=none] (1) at (-1, 0) {};
		\node [style=none] (2) at (1, 0) {};
	\end{pgfonlayer}
	\begin{pgfonlayer}{edgelayer}
		\draw (1.center) to (2.center);
	\end{pgfonlayer}
\end{tikzpicture}
 = \input{./figures/copy.tikz}$, then
    \begin{equation}
        \begin{split}
            &LHS = \input{./figures/copy-trans-lhs.tikz} \overset{\substack{\zfrob \\ \zunit}}{=} \input{./figures/copy-trans-1.tikz} \\
            &\overset{\substack{\zfrob \\ \zunit}}{=} \input{./figures/copy-trans-2.tikz} \overset{\substack{\zfrob \\ \zunit}}{=} \input{./figures/copy-trans-3.tikz} \\
            &\overset{\zsymm^\dag}{=} \input{./figures/copy-op.tikz} = RHS
        \end{split}
    \end{equation}
    If $\begin{tikzpicture}
	\begin{pgfonlayer}{nodelayer}
		\node [style=rmat] (0) at (0, 0) {$f$};
		\node [style=none] (1) at (-1, 0) {};
		\node [style=none] (2) at (1, 0) {};
	\end{pgfonlayer}
	\begin{pgfonlayer}{edgelayer}
		\draw (1.center) to (2.center);
	\end{pgfonlayer}
\end{tikzpicture}
 = \input{./figures/add.tikz}$, the derivation is similar:
    \begin{equation}
        \begin{split}
            &LHS = \input{./figures/add-trans-lhs.tikz} \overset{\substack{\cccup \\ \cccap \\ \kinv}}{=} \input{./figures/add-trans-1.tikz} \\
            &\overset{\substack{\xfrob \\ \xunit}}{=} \input{./figures/add-trans-2.tikz} \overset{\xsymm}{=} \input{./figures/add-trans-3.tikz} \\
            &\overset{\kadd}{=} \input{./figures/add-trans-4.tikz} \overset{\kassoc}{=} \input{./figures/add-op.tikz}
        \end{split}
    \end{equation}
    For $\begin{tikzpicture}
	\begin{pgfonlayer}{nodelayer}
		\node [style=rmat] (0) at (0, 0) {$f$};
		\node [style=none] (1) at (-1, 0) {};
		\node [style=none] (2) at (1, 0) {};
	\end{pgfonlayer}
	\begin{pgfonlayer}{edgelayer}
		\draw (1.center) to (2.center);
	\end{pgfonlayer}
\end{tikzpicture}
 = \begin{tikzpicture}
	\begin{pgfonlayer}{nodelayer}
		\node [style=Z dot] (0) at (0.5, 0) {};
		\node [style=none] (1) at (-0.5, 0) {};
	\end{pgfonlayer}
	\begin{pgfonlayer}{edgelayer}
		\draw (1.center) to (0);
	\end{pgfonlayer}
\end{tikzpicture}
$ or $\begin{tikzpicture}
	\begin{pgfonlayer}{nodelayer}
		\node [style=X dot] (0) at (-0.5, 0) {};
		\node [style=none] (3) at (0.5, 0) {};
	\end{pgfonlayer}
	\begin{pgfonlayer}{edgelayer}
		\draw (0) to (3.center);
	\end{pgfonlayer}
\end{tikzpicture}
$, the proposition is a straightforward application of their respective (co)unitality laws \zunit and \xunit. Finally, for $\begin{tikzpicture}
	\begin{pgfonlayer}{nodelayer}
		\node [style=rmat] (0) at (0, 0) {$f$};
		\node [style=none] (1) at (-1, 0) {};
		\node [style=none] (2) at (1, 0) {};
	\end{pgfonlayer}
	\begin{pgfonlayer}{edgelayer}
		\draw (1.center) to (2.center);
	\end{pgfonlayer}
\end{tikzpicture}
 = \input{./figures/mult.tikz}$, $\begin{tikzpicture}
	\begin{pgfonlayer}{nodelayer}
		\node [style=multop] (0) at (-0.5, 0) {};
		\node [style=none] (3) at (0.5, 0) {};
	\end{pgfonlayer}
	\begin{pgfonlayer}{edgelayer}
		\draw (0) to (3.center);
	\end{pgfonlayer}
\end{tikzpicture}
$, or $\begin{tikzpicture}
	\begin{pgfonlayer}{nodelayer}
		\node [style=kscalar] (1) at (0, 0) {$\kappa$};
		\node [style=none] (2) at (-1, 0) {};
		\node [style=none] (3) at (1, 0) {};
	\end{pgfonlayer}
	\begin{pgfonlayer}{edgelayer}
		\draw (2.center) to (3.center);
	\end{pgfonlayer}
\end{tikzpicture}
$, the proposition is verbatim identical to the laws \mtrp, \onetrp, or \ktrp respectively.

    Now for the inductive step: suppose the proposition holds for $\begin{tikzpicture}
	\begin{pgfonlayer}{nodelayer}
		\node [style=rmat] (0) at (0, 0) {$f$};
		\node [style=none] (1) at (-1, 0) {};
		\node [style=none] (2) at (1, 0) {};
		\node [style=none] (3) at (-1.5, 0) {$n$};
		\node [style=none] (4) at (1.5, 0) {$m$};
	\end{pgfonlayer}
	\begin{pgfonlayer}{edgelayer}
		\draw (1.center) to (2.center);
	\end{pgfonlayer}
\end{tikzpicture}
$ and $\begin{tikzpicture}
	\begin{pgfonlayer}{nodelayer}
		\node [style=rmat] (0) at (0, 0) {$g$};
		\node [style=none] (1) at (-1, 0) {};
		\node [style=none] (2) at (1, 0) {};
		\node [style=none] (3) at (-1, 0.5) {$n'$};
		\node [style=none] (4) at (1, 0.5) {$m'$};
	\end{pgfonlayer}
	\begin{pgfonlayer}{edgelayer}
		\draw (1.center) to (2.center);
	\end{pgfonlayer}
\end{tikzpicture}
$. Then, in the case $m = n'$, consider the diagram $\begin{tikzpicture}
	\begin{pgfonlayer}{nodelayer}
		\node [style=rmat] (0) at (-1, 0) {$f$};
		\node [style=rmat] (1) at (1, 0) {$g$};
		\node [style=none] (2) at (-2, 0) {};
		\node [style=none] (3) at (2, 0) {};
	\end{pgfonlayer}
	\begin{pgfonlayer}{edgelayer}
		\draw (2.center) to (3.center);
	\end{pgfonlayer}
\end{tikzpicture}
$. we have
    \begin{equation}
        \begin{split}
            &LHS = \input{./figures/poly-comp-trans-lhs.tikz} \\
            &\overset{\ref{eq:n-id-yank}}{=} \input{./figures/poly-comp-trans-1.tikz} \overset{I.H.}{=} \begin{tikzpicture}
	\begin{pgfonlayer}{nodelayer}
		\node [style=lmat] (0) at (1, 0) {$f$};
		\node [style=lmat] (1) at (-1, 0) {$g$};
		\node [style=none] (2) at (2, 0) {};
		\node [style=none] (3) at (-2, 0) {};
	\end{pgfonlayer}
	\begin{pgfonlayer}{edgelayer}
		\draw (2.center) to (3.center);
	\end{pgfonlayer}
\end{tikzpicture}
 = RHS
        \end{split}
    \end{equation}
    And for the diagram $\input{./figures/poly-tensor.tikz}$, we have
    \begin{equation}
        \begin{split}
            &LHS = \input{./figures/poly-tensor-trans-lhs.tikz} = \input{./figures/poly-tensor-trans-1.tikz}  \\
            &\overset{I.H.}{=} \input{./figures/poly-tensor-trans-2.tikz} = \input{./figures/poly-tensor-op.tikz} = RHS
        \end{split}
    \end{equation}
    Where the second equality follows from the naturality of the braidings, and the fourth equality follows from both the naturality and the involutivity of the braidings. This completes the structural induction and proof.
\end{proof}

\subsection{Proofs about Ideal Gadgets}
\label{append:ideal-lemmas}

In this appendix, we prove that the following lemmas hold in $\LangCsp$.

\principleidealflex*

\begin{proof}
    \begin{align*}
        LHS &= \input{./figures/polyzero-going-left.tikz} \\
        &\overset{\ref{eq:n-id-yank}}{=} \input{./figures/ideal-flex-1.tikz} \\
        &\overset{\ref{eq:poly-trans}, \ref{eq:poly-trans}^\dag}{=} \input{./figures/ideal-flex-2.tikz} \\
        &\overset{\zfrob}{=} \input{./figures/polyzero-going-right.tikz} = RHS
    \end{align*}
\end{proof}

\principleidealscommute*

\begin{proof}
    \begin{equation}
        \begin{split}
            &LHS = \input{./figures/polyzero-commute-lhs.tikz} \overset{\zassoc}{=} \input{./figures/polyzero-commute-1.tikz} \\
            &\overset{\zsymm}{=} \input{./figures/polyzero-commute-2.tikz} \overset{S.M.C.}{=} \input{./figures/polyzero-commute-rhs.tikz} = RHS.
        \end{split}
    \end{equation}
    Here the equality labelled S.M.C. follows from the naturality of the braidings in a symmetric monoidal category.
\end{proof}

\idealbasis*

\begin{proof}
    First, by the completeness of $\mathbb L_{CAlg_k}$, we know that the following equation is derivable:
    \begin{equation}
        \label{eq:mult-by-zero}
        \input{./figures/mult-by-zero-lhs.tikz} = \begin{tikzpicture}
	\begin{pgfonlayer}{nodelayer}
		\node [style=kscalar] (0) at (0, 0) {$0$};
		\node [style=none] (1) at (-1, 0) {};
		\node [style=none] (2) at (1, 0) {};
	\end{pgfonlayer}
	\begin{pgfonlayer}{edgelayer}
		\draw (1.center) to (2.center);
	\end{pgfonlayer}
\end{tikzpicture}
 = \begin{tikzpicture}
	\begin{pgfonlayer}{nodelayer}
		\node [style=Z dot] (0) at (-0.5, 0) {};
		\node [style=X dot] (1) at (0.5, 0) {};
		\node [style=none] (2) at (-1.5, 0) {};
		\node [style=none] (3) at (1.5, 0) {};
	\end{pgfonlayer}
	\begin{pgfonlayer}{edgelayer}
		\draw (2.center) to (0);
		\draw (1) to (3.center);
	\end{pgfonlayer}
\end{tikzpicture}

    \end{equation}
    Note, the second equality here is verbatim rule \zeroInK; while the first equality can be easily derived by substituting
    \begin{align*}
        \begin{tikzpicture}
	\begin{pgfonlayer}{nodelayer}
		\node [style=X dot] (0) at (-0.5, 0) {};
		\node [style=none] (3) at (0.5, 0) {};
	\end{pgfonlayer}
	\begin{pgfonlayer}{edgelayer}
		\draw (0) to (3.center);
	\end{pgfonlayer}
\end{tikzpicture}
 \overset{\onedel}{=} \begin{tikzpicture}
	\begin{pgfonlayer}{nodelayer}
		\node [style=multop] (0) at (-0.5, 0) {};
		\node [style=Z dot] (1) at (0.5, 0) {};
	\end{pgfonlayer}
	\begin{pgfonlayer}{edgelayer}
		\draw (0) to (1);
	\end{pgfonlayer}
\end{tikzpicture}
 \,\, \begin{tikzpicture}
	\begin{pgfonlayer}{nodelayer}
		\node [style=X dot] (0) at (-0.5, 0) {};
		\node [style=none] (3) at (0.5, 0) {};
	\end{pgfonlayer}
	\begin{pgfonlayer}{edgelayer}
		\draw (0) to (3.center);
	\end{pgfonlayer}
\end{tikzpicture}
 \overset{\zeroInK}{=} \begin{tikzpicture}
	\begin{pgfonlayer}{nodelayer}
		\node [style=multop] (0) at (-1, 0) {};
		\node [style=kscalar] (1) at (0, 0) {$0$};
		\node [style=none] (2) at (1, 0) {};
	\end{pgfonlayer}
	\begin{pgfonlayer}{edgelayer}
		\draw (0) to (2.center);
	\end{pgfonlayer}
\end{tikzpicture}

    \end{align*}

    Now, since $f \in (f_1, \dots, f_m)$, there exist polynomials $g_1, \dots, g_m \in k[x_1, \dots, x_n]$ such that $f = g_1 f_1 + \dots + g_m f_m$. Then by the completeness of $\mathbb L_{CAlg_k}$, we know that the following is derivable by its rewrite rules:
    \begin{equation}
        \label{eq:ideal-element}
        \begin{tikzpicture}
	\begin{pgfonlayer}{nodelayer}
		\node [style=rmat] (0) at (0, 0) {$f$};
		\node [style=none] (1) at (-1, 0) {};
		\node [style=none] (2) at (1, 0) {};
	\end{pgfonlayer}
	\begin{pgfonlayer}{edgelayer}
		\draw (1.center) to (2.center);
	\end{pgfonlayer}
\end{tikzpicture}
 = \input{./figures/ideal-element-constr.tikz}
    \end{equation}
    Thus we can rewrite the right-hand-side of \eqref{eq:ideal-basis-with-redundancy} as follows:
    \begin{multline}
            RHS = \input{./figures/polyzero-ideal-basis-with-redundancy.tikz} \\
            \overset{\ref{eq:ideal-element}}{=} \input{./figures/polyzero-ideal-basis-1.tikz} \\
            \overset{\zfrob}{=} \input{./figures/polyzero-ideal-basis-2.tikz} \\
            \overset{\polycp}{=} \input{./figures/polyzero-ideal-basis-3.tikz} \\
            \overset{\zxbialg}{=} \input{./figures/polyzero-ideal-basis-4.tikz} \\
    %     \end{split}
    % \end{equation}
    % This allows us to apply equation \ref{eq:mult-by-zero} to obtain
    % \begin{equation}
    %     \begin{split}
            \overset{\ref{eq:mult-by-zero}}{=} \input{./figures/polyzero-ideal-basis-5.tikz} \\
        \overset{\polydel}{=} \input{./figures/polyzero-ideal-basis-6.tikz} \\
        \overset{\zunit}{=} \input{./figures/polyzero-ideal-basis-7.tikz}
        \overset{\xunit, \xbone}{=} \input{./figures/polyzero-ideal-basis-8.tikz}  \\
        \overset{\zassoc}{=} \input{./figures/polyzero-ideal-basis.tikz} = LHS
    \end{multline}
\end{proof}

% \begin{prop}
%     (Lemma \ref{lem:diagram-reduced-ring}) Let $I \subset k[x_1, \dots, x_n]$ be an ideal, and let $f \in (k[x_1, \dots, x_n])^m$ be a polynomial such that $f \in (\sqrt{I})^m$. Then
%     $$\tikzfig{ring-quot-gadget} = \tikzfig{ring-quot-gadget-with-radical-redundancy}$$
% \end{prop}

% \begin{proof}
%     Since $f \in (\sqrt{I})^m$, for each component $f_j$, there exists some $g_j \in I$ and some $n_j \in \mathbb N$ such that $g_j^{n_j} = f_j$. Thus by the completeness of $\mathbb L_{CAlg_k}$, the following is derivable:
%     \begin{equation}
%         \label{eq:fj-is-gj-pow}
%         \tikzfig{fj-mat} = \tikzfig{gj-pow}
%     \end{equation}
%     So we have
%     \begin{align*}
%         RHS &= \tikzfig{ring-quot-gadget-with-radical-redundancy}
%         \overset{\zassoc}{=} \tikzfig{reduced-1}
%         \overset{\ref{eq:fj-is-gj-pow}}{=} \tikzfig{reduced-2} \\
%         &\overset{\ref{ax:reduced}}{=} \tikzfig{reduced-3}
%         \overset{\ref{lem:ideal-basis}}{=} \tikzfig{ring-quot-gadget} = LHS
%     \end{align*}
% \end{proof}

\quotientgadgets*

\begin{proof}
    We first show this for the case where $I$ is a principal ideal: say $I = (f)$. We view the gadget as a composition of two diagrams like so:
    \begin{equation}
        \input{./figures/polyzero.tikz} = \input{./figures/polyzero-as-spans.tikz}
    \end{equation}
    Note that the left-half is a diagram in $\mathbb L_{CAlg_k}$, while the right-half is a diagram in $\mathbb L_{CAlg_k}^{op}$. So
    \begin{multline}
        \label{deriv:semantic-of-principal-ideal}
        \left[\input{./figures/polyzero.tikz}\right] \\
        = \left(k[x_1, \dots, x_n] \xleftarrow{\substack{f(\mathbf x) \leftmapsto y_0 \\ x_j \leftmapsto y_j}} k[y_0, \dots, y_n] \xrightarrow{\substack{y_0 \mapsto 0 \\ y_j \mapsto z_j}} k[z_1, \dots, z_n]\right) \\
        = \left(k[x_1, \dots, x_n] \rightarrow \frac{k[x_1, \dots, x_n, z_1, \dots, z_n]}{(f(\mathbf x), x_j - z_j)} \leftarrow k[z_1, \dots, z_n]\right) \\
        = \left(k[x_1, \dots, x_n] \rightarrow \frac{k[x_1, \dots, x_n]}{(f(\mathbf x))} \leftarrow k[x_1, \dots, x_n]\right)
    \end{multline}
    where, in the derivation \ref{deriv:semantic-of-principal-ideal}, the index $j$ is always quantified universally over $1, \dots, n$.

    Now we perform the inductive step: suppose that $I, J \subset k[x_1, \dots, x_n]$ are ideals for which the proposition holds. Then the composition
    $\input{./figures/ring-quot-gadget-comp.tikz}$
    is the gadget corresponding to the ideal generated by the union of the elements in $I$ and $J$ (since we have taken a union of the generators). In other words it is the gadget corresponding to $I+J$. We compute:
    \begin{equation}
        \begin{split}
            &\left[\input{./figures/ring-quot-gadget-comp.tikz}\right] \\
        &= \left[\begin{tikzpicture}
	\begin{pgfonlayer}{nodelayer}
		\node [style=Z dot] (1) at (0, -0.5) {};
		\node [style=none] (2) at (1, -0.5) {};
		\node [style=none] (3) at (-1, -0.5) {};
		\node [style=idealket] (5) at (0, 0.5) {$J$};
	\end{pgfonlayer}
	\begin{pgfonlayer}{edgelayer}
		\draw (3.center) to (2.center);
		\draw (5) to (1);
	\end{pgfonlayer}
\end{tikzpicture}
\right] \circ \left[\begin{tikzpicture}
	\begin{pgfonlayer}{nodelayer}
		\node [style=none] (0) at (-1, -0.5) {};
		\node [style=none] (1) at (1, -0.5) {};
		\node [style=Z dot] (2) at (0, -0.5) {};
		\node [style=idealket] (3) at (0, 0.5) {$I$};
	\end{pgfonlayer}
	\begin{pgfonlayer}{edgelayer}
		\draw (0.center) to (1.center);
		\draw (3) to (2);
	\end{pgfonlayer}
\end{tikzpicture}
\right] \\
        &= \left(k\left[\substack{x_1 \\ \vdots \\ x_n}\right] \twoheadrightarrow k\left[\substack{x_1 \\ \vdots \\ x_n}\right]/I \twoheadleftarrow k\left[\substack{x_1 \\ \vdots \\ x_n}\right] \twoheadrightarrow k\left[\substack{x_1 \\ \vdots \\ x_n}\right]/J \twoheadleftarrow k\left[\substack{x_1 \\ \vdots \\ x_n}\right]\right) \\
        &= \left(k\left[\substack{x_1 \\ \vdots \\ x_n}\right] \twoheadrightarrow \frac{k[x_1, \dots, x_n]}{I + J} \twoheadleftarrow k\left[\substack{x_1 \\ \vdots \\ x_n}\right]\right)
        \end{split}
    \end{equation}
    Recall that, since $k[x_1, \dots, x_n]$ is a Noetherian ring, every ideal is a finite sum of principal ideals. This completes the induction and the proof.
\end{proof}

\subsection{Proofs about Universality and Completeness}
\cspUniversal*

\begin{proof}
    Recall from proposition \ref{prop:quotient-gadget} that the semantic of the quotient gadget in the middle of the diagram on the left-hand-side is
    \begin{equation}
        \left[\begin{tikzpicture}
	\begin{pgfonlayer}{nodelayer}
		\node [style=none] (0) at (-1, -0.5) {};
		\node [style=none] (1) at (1, -0.5) {};
		\node [style=Z dot] (2) at (0, -0.5) {};
		\node [style=idealket] (3) at (0, 0.5) {$I$};
	\end{pgfonlayer}
	\begin{pgfonlayer}{edgelayer}
		\draw (0.center) to (1.center);
		\draw (3) to (2);
	\end{pgfonlayer}
\end{tikzpicture}
\right] = \left(S_k(n) \twoheadrightarrow S_k(n)/I \twoheadleftarrow S_k(n)\right)
    \end{equation}
    Thus we have
    \begin{equation}
        \begin{split}
            &\left[\input{./figures/span-form.tikz}\right] \\
            &= \left(S_k(m) \xrightarrow{f^\#}S_k(n) \twoheadrightarrow S_k(n)/I \twoheadleftarrow S_k(n) \xleftarrow{g^\#} S_k(m')\right) \\
            &= \left(S_k(m) \xrightarrow{f^\#_I} A  \xleftarrow{g^\#_I} S_k(m')\right)
        \end{split}
    \end{equation}
    as required.
\end{proof}

\cospanForm*

\begin{proof}
    We prove by structural induction. First we show that every generator of $\LangCsp$ can be rewritten into cospan form. By $\dag$-symmetry, we only need to show this for the generators coming from $\mathbb L_{CAlg_k}$.

    But this is immediate: for any $\begin{tikzpicture}
	\begin{pgfonlayer}{nodelayer}
		\node [style=rmat] (0) at (0, 0) {$f$};
		\node [style=none] (1) at (-1, 0) {};
		\node [style=none] (2) at (1, 0) {};
	\end{pgfonlayer}
	\begin{pgfonlayer}{edgelayer}
		\draw (1.center) to (2.center);
	\end{pgfonlayer}
\end{tikzpicture}
 \in \mathbb L_{CAlg_k}$, we may rewrite it as
    \begin{equation}
        \label{eq:graphical-graph-embedding}
        \begin{tikzpicture}
	\begin{pgfonlayer}{nodelayer}
		\node [style=rmat] (0) at (0, 0) {$f$};
		\node [style=none] (1) at (-1, 0) {};
		\node [style=none] (2) at (1, 0) {};
	\end{pgfonlayer}
	\begin{pgfonlayer}{edgelayer}
		\draw (1.center) to (2.center);
	\end{pgfonlayer}
\end{tikzpicture}
 \overset{\zunit, \xbone}{=} \input{./figures/f-mat-span-constr.tikz} \overset{\ref{def:ring-quot-gadget}}{=} \input{./figures/f-mat-span.tikz}
    \end{equation}

    Thus all the generators of $\LangCsp$ can be rewritten into cospan form. It remains now to show that compositions and monoidal products of cospan form diagrams can be rewritten into cospan form.

    Take any pair of composable diagrams in cospan form. Their composition can be rewritten thus:
    \begin{equation}
        \begin{split}
            &\input{./figures/span-compose-1.tikz}  \\
        &\overset{\zfusion}{=} \input{./figures/span-compose-2.tikz} \\
        &\overset{\cccap, \ref{eq:poly-trans}}{=} \input{./figures/span-compose-3.tikz} \\
        &\overset{\text{S.M.C.}}{=} \input{./figures/span-compose-4.tikz}  \\
        &\overset{\ref{def:ring-quot-gadget}}= \input{./figures/span-compose-5.tikz}
        \end{split}
    \end{equation}
    Where the equality labelled ``S.M.C.'' is an application of the coherence axioms of symmetric monoidal categories. Note that this corresponds precisely to the pushout operation taken in the composition of cospans.

    Finally, take any pair of arbitrary diagrams in cospan form. Their monoidal product can be rewritten as:
    \begin{align*}
        \input{./figures/span-tensor-1.tikz} \overset{S.M.C.}{=} \input{./figures/span-tensor-2.tikz}
    \end{align*}
    which is again in cospan form. Thus by induction all diagrams in $\LangCsp$ can be rewritten into cospan form.
\end{proof}

\diagramSheafRing*

\begin{proof}
    Suppose $h = f - g \in I^m$. By Proposition~\ref{prop:forward-fragment-universal-complete}
    \begin{equation}
    \begin{tikzpicture}
	\begin{pgfonlayer}{nodelayer}
		\node [style=rmat] (0) at (0, 0) {$f$};
		\node [style=none] (1) at (-1, 0) {};
		\node [style=none] (2) at (1, 0) {};
	\end{pgfonlayer}
	\begin{pgfonlayer}{edgelayer}
		\draw (1.center) to (2.center);
	\end{pgfonlayer}
\end{tikzpicture}
 = \input{./figures/g-plus-h-mat.tikz}
    \end{equation}
    is a rewrite in $\mathbb L_{CAlg_k}$. So we compute:
    \begin{equation}
        \begin{split}
        &\input{./figures/ring-on-variety-lhs.tikz} = \input{./figures/ring-on-variety-1.tikz} \\
        &\overset{\ref{lem:ideal-basis}}{=} \input{./figures/ring-on-variety-2.tikz}=
        \overset{\zassoc}{=} \input{./figures/ring-on-variety-3.tikz} \\
        &\overset{\polycp}{=} \input{./figures/ring-on-variety-4.tikz}
        \overset{\zerocp}{=} \input{./figures/ring-on-variety-5.tikz} \\
        &\overset{\xunit}{=} \input{./figures/ring-on-variety-6.tikz}
        \overset{\ref{lem:ideal-basis}}{=} \input{./figures/ring-on-variety-rhs.tikz}
        \end{split}
    \end{equation}
\end{proof}

\diagramRespectsKernel*

\begin{proof}
    Suppose $g_1, \dots, g_m$ is a basis $J = (g_1, \dots, g_m)$. Then for each $g_j$, we have $f^\#_I(g_j) = 0$. Using the definition of the $\#$, we have in other words $(g_j \circ f)^\#_I = 0$. So $g_j \circ f \in I$. Thus by lemma \ref{lem:diagram-sheaf-ring}, we have a rewrite
    \begin{equation}
        \label{eq:regular-map-comp-nil}
        \input{./figures/regular-map-comp-nil-lhs.tikz} = \input{./figures/regular-map-comp-nil-rhs.tikz}
    \end{equation}
    So we can derive:
    \begin{equation}
        \begin{split}
            &\input{./figures/regular-map.tikz} \\
            &\overset{\ref{def:ring-quot-gadget}}= \input{./figures/regular-map-1.tikz} \\
        &\overset{\polycp}= \input{./figures/regular-map-2.tikz} \\
        &\overset{\zfrob, \ref{cor:idem}}= \input{./figures/regular-map-3.tikz} \\
        &\overset{\ref{eq:regular-map-comp-nil}}= \input{./figures/regular-map-4.tikz} \\
        &\overset{\zfrob, \ref{cor:idem}}= \input{./figures/regular-map-5.tikz} \\
        &\overset{\zunit, \xbone} = \input{./figures/regular-map-without-im.tikz}
        \end{split}
    \end{equation}
\end{proof}

\isoInverseIsOp*

\begin{proof}
    We know that
    \begin{equation}
        \begin{split}
            \pi_I \circ \sigma^\# &= \alpha \circ \pi_J \\
            \pi_J \circ \tau^\# &= \alpha^{-1} \circ \pi_I
        \end{split}
    \end{equation}
    from which we can easily derive that
    \begin{equation}
        \begin{split}
            \pi_I \circ \sigma^\# \circ \tau^\# &= \pi_I \\
            \pi_J \circ \tau^\# \circ \sigma^\# &= \pi_J
        \end{split}
    \end{equation}
    which is to say,
    \begin{equation}
        \begin{split}
            (\tau \circ \sigma)^\#_I &= (id)^\#_I \\
            (\sigma \circ \tau)^\#_J &= (id)^\#_J
        \end{split}
    \end{equation}
    So by lemma \ref{lem:diagram-sheaf-ring}, the following are derivable in $\LangCsp$:
    \begin{equation}
        \label{eq:regular-map-comp-id}
        \begin{split}
            &\input{./figures/regular-map-comp-id-forward.tikz} = \begin{tikzpicture}
	\begin{pgfonlayer}{nodelayer}
		\node [style=none] (0) at (-1, -0.5) {};
		\node [style=none] (1) at (1, -0.5) {};
		\node [style=Z dot] (2) at (0, -0.5) {};
		\node [style=idealket] (3) at (0, 0.5) {$I$};
	\end{pgfonlayer}
	\begin{pgfonlayer}{edgelayer}
		\draw (0.center) to (1.center);
		\draw (3) to (2);
	\end{pgfonlayer}
\end{tikzpicture}
 \\
            &\input{./figures/regular-map-comp-id-backward.tikz} = \begin{tikzpicture}
	\begin{pgfonlayer}{nodelayer}
		\node [style=Z dot] (1) at (0, -0.5) {};
		\node [style=none] (2) at (1, -0.5) {};
		\node [style=none] (3) at (-1, -0.5) {};
		\node [style=idealket] (5) at (0, 0.5) {$J$};
	\end{pgfonlayer}
	\begin{pgfonlayer}{edgelayer}
		\draw (3.center) to (2.center);
		\draw (5) to (1);
	\end{pgfonlayer}
\end{tikzpicture}

        \end{split}
    \end{equation}
    Now we can show that $\input{./figures/sigma-op.tikz}$ is the left inverse of $\sigma$:
    \begin{align}
        \label{eq:trans-is-left-inv}
        \begin{split}
        &\begin{tikzpicture}
	\begin{pgfonlayer}{nodelayer}
		\node [style=none] (0) at (-1, -0.5) {};
		\node [style=none] (1) at (1, -0.5) {};
		\node [style=Z dot] (2) at (0, -0.5) {};
		\node [style=idealket] (3) at (0, 0.5) {$I$};
	\end{pgfonlayer}
	\begin{pgfonlayer}{edgelayer}
		\draw (0.center) to (1.center);
		\draw (3) to (2);
	\end{pgfonlayer}
\end{tikzpicture}
 \overset{\zunit, \polydel}= \input{./figures/inverse-1.tikz} \\
        &\overset{\ref{eq:z-special}}= \input{./figures/inverse-2.tikz}
        \overset{\polycp}= \input{./figures/inverse-3.tikz}  \\
        &\overset{\zfrob, \ref{eq:poly-trans}}= \input{./figures/inverse-4.tikz}
        \overset{\ref{cor:idem}, \zfrob}= \input{./figures/inverse-5.tikz} \\
        &\overset{\ref{eq:regular-map-comp-id}}= \input{./figures/inverse-6.tikz} \\
        &\overset{\zfrob, \polycp}= \input{./figures/inverse-7.tikz} \\
        &\overset{\zfrob, \ref{eq:poly-trans}}= \input{./figures/inverse-8.tikz} \\
        &\overset{\polycp}= \input{./figures/inverse-9.tikz} \\
        &\overset{\ref{eq:z-special}}= \input{./figures/inverse-10.tikz} \\
        &\overset{\polydel, \zunit}= \input{./figures/inverse-11.tikz}
        \end{split}
    \end{align}
    And consequently we have the rewrite we need:
    \begin{align*}
        &\input{./figures/sigma-inverse.tikz} \overset{\ref{lem:diagram-respects-kernel}}= \input{./figures/sigma-inverse-with-im.tikz} \\
        &\overset{\ref{eq:trans-is-left-inv}}= \input{./figures/sigma-inverse-with-pair-1.tikz} \\
        &\overset{\ref{lem:diagram-respects-kernel}}= \input{./figures/sigma-inverse-with-pair-2.tikz} \\
        &\overset{\ref{eq:regular-map-comp-id}}= \input{./figures/sigma-inverse-with-pair-3.tikz} \\
        &\overset{\ref{lem:diagram-respects-kernel}^{op}}= \input{./figures/sigma-op.tikz}
    \end{align*}
\end{proof}

\conjByIso*

\begin{proof}
    \begin{equation}
        \begin{split}
            \input{./figures/ring-quot-conj-sigma.tikz} &\overset{\ref{lem:iso-inverse-is-op}}{=} \input{./figures/ring-quot-conj-1.tikz} \\
            &\overset{\ref{eq:regular-map-comp-id}}= \begin{tikzpicture}
	\begin{pgfonlayer}{nodelayer}
		\node [style=Z dot] (1) at (0, -0.5) {};
		\node [style=none] (2) at (1, -0.5) {};
		\node [style=none] (3) at (-1, -0.5) {};
		\node [style=idealket] (5) at (0, 0.5) {$J$};
	\end{pgfonlayer}
	\begin{pgfonlayer}{edgelayer}
		\draw (3.center) to (2.center);
		\draw (5) to (1);
	\end{pgfonlayer}
\end{tikzpicture}

        \end{split}
    \end{equation}
\end{proof}

\cspComplete*

\begin{proof}
    This is in other words to say that our rewrite rules are complete. By proposition \ref{prop:cospan-form} any diagram in $\LangCsp$ can be rewritten into cospan form. So it suffices to prove that, whenever two diagrams in cospan form are mapped to isomorphic cospans by $\sem$, then they are rewritable into each other. Pick any two diagrams in cospan form
    $$\input{./figures/span-form.tikz} \text{ and } \input{./figures/span-form-alt.tikz}$$
    with the same wire types on both ends. Suppose there exists a cospan isomorphism
    % https://q.uiver.app/#q=WzAsNCxbMSwwLCJWKEkpIl0sWzAsMSwia15uIl0sWzIsMSwia15tIl0sWzEsMiwiVihKKSJdLFswLDEsImYiLDJdLFswLDIsImciXSxbMywxLCJmJyJdLFszLDIsImcnIiwyXSxbMCwzLCJcXHNpZ21hIiwxXV0=
    \begin{equation}
        \label{eq:csp-iso}
        \begin{tikzcd}
            & {S_k(u)/I} \\
            {k^n} && {k^m} \\
            & {S_k(v)/J}
            \arrow["f^\#_I"', from=2-1, to=1-2]
            \arrow["g^\#_I", from=2-3, to=1-2]
            \arrow["\alpha"{description}, from=3-2, to=1-2]
            \arrow["{f'^\#_J}", from=2-1, to=3-2]
            \arrow["{g'^\#_J}"', from=2-3, to=3-2]
        \end{tikzcd}
    \end{equation}
    between their images under the semantics functor.

    Since $\alpha: S_k(v)/J \to S_k(u)/I$ is an isomorphism, corollary \ref{cor:conj-by-iso} gives a rewrite
    $$\begin{tikzpicture}
	\begin{pgfonlayer}{nodelayer}
		\node [style=Z dot] (1) at (0, -0.5) {};
		\node [style=none] (2) at (1, -0.5) {};
		\node [style=none] (3) at (-1, -0.5) {};
		\node [style=idealket] (5) at (0, 0.5) {$J$};
	\end{pgfonlayer}
	\begin{pgfonlayer}{edgelayer}
		\draw (3.center) to (2.center);
		\draw (5) to (1);
	\end{pgfonlayer}
\end{tikzpicture}
 = \input{./figures/ring-quot-conj-sigma.tikz}$$
    where $\sigma$ is some tuple of polynomials such that $\sigma^\#_I = \alpha \circ \pi_J$ (say, pick $\sigma = (\alpha \circ \pi_J)^\flat_I$). Moreover, since diagram \ref{eq:csp-iso} commutes,
    \begin{equation}
        f^\#_I = \alpha \circ f'^\#_J = \alpha \circ \pi_J \circ f'^\# = \sigma^\#_I \circ f'^\# = (f' \circ \sigma)^\#_I
    \end{equation}
    and similarly we may derive $g^\#_I = (g' \circ \sigma)^\#_I$.
    So we have the following rewrites:
    \begin{equation}
        \begin{split}
            \input{./figures/span-form-alt.tikz} &\overset{\ref{cor:conj-by-iso}}= \input{./figures/completeness-1.tikz} \\&\overset{\ref{lem:diagram-sheaf-ring}}= \input{./figures/span-form.tikz}
        \end{split}
    \end{equation}
\end{proof}

\section{Proofs about $\LangSpk$ and $\LangSpq$} 
\label{sec:append-sp} 

\subsection{Proofs about $\LangSpk$} 

\reductionSound* 

\begin{proof}
    We compute directly: 
    \begin{equation}
        \begin{split}
            \Spec &\left[\input{./figures/nilpotent-lhs.tikz}\right]_{red} = \Spec \left(k[x] \xleftarrow{x^2} k[y] \xrightarrow{0} k\right)_{red} \\ 
            &\overset{pushout}{=} \Spec \left(k[x] \rightarrow k[x]/(x^2) \leftarrow k\right)_{red} \\
            &\overset{reduce}{=} \Spec\left(k[x] \xrightarrow{0} k = k\right) \\ 
            &= \Spec \left[\begin{tikzpicture}
	\begin{pgfonlayer}{nodelayer}
		\node [style=X dot] (0) at (0.5, 0) {};
		\node [style=none] (1) at (-0.5, 0) {};
	\end{pgfonlayer}
	\begin{pgfonlayer}{edgelayer}
		\draw (1.center) to (0);
	\end{pgfonlayer}
\end{tikzpicture}
\right]_{red}
        \end{split}
    \end{equation}    
    Thus rule \ref{ax:reduced} is sound with respect to the map $\Spec \sem_{red}$, and so the map factors through $\LangSpk$. 
\end{proof}

\diagramReduction* 

\begin{proof}
    Since $\sqrt{I}$ will have a finite basis, it suffices to prove that the following rewrite exists in $\LangSpk$ 
    \begin{equation}
        \begin{tikzpicture}
	\begin{pgfonlayer}{nodelayer}
		\node [style=none] (0) at (-1, -0.5) {};
		\node [style=none] (1) at (1, -0.5) {};
		\node [style=Z dot] (2) at (0, -0.5) {};
		\node [style=idealket] (3) at (0, 0.5) {$I$};
	\end{pgfonlayer}
	\begin{pgfonlayer}{edgelayer}
		\draw (0.center) to (1.center);
		\draw (3) to (2);
	\end{pgfonlayer}
\end{tikzpicture}
 = \input{./figures/ring-quot-gadget-with-radical-redundancy.tikz}
    \end{equation}
    for any $f \in \sqrt{I}$.  

    Since $f \in \sqrt{I}$, there exists some $g \in I$ and some $n \in \mathbb N$ such that $g^{n} = f$. Thus by the completeness of $\mathbb L_{CAlg_k}$ alone, the following is derivable: 
    \begin{equation}
        \label{eq:f-is-g-pow}
        \begin{tikzpicture}
	\begin{pgfonlayer}{nodelayer}
		\node [style=rmat] (0) at (0, 0) {$f$};
		\node [style=none] (1) at (-1, 0) {};
		\node [style=none] (2) at (1, 0) {};
	\end{pgfonlayer}
	\begin{pgfonlayer}{edgelayer}
		\draw (1.center) to (2.center);
	\end{pgfonlayer}
\end{tikzpicture}
 = \input{./figures/g-pow.tikz}
    \end{equation}
    So we have 
    \begin{align*}
        RHS &= \input{./figures/ring-quot-gadget-with-radical-redundancy.tikz} 
        \overset{\ref{eq:f-is-g-pow}}{=} \input{./figures/reduced-2.tikz} \\ 
        &\overset{\ref{ax:reduced}}{=} \input{./figures/reduced-3.tikz} 
        \overset{\ref{lem:ideal-basis}}{=} \begin{tikzpicture}
	\begin{pgfonlayer}{nodelayer}
		\node [style=none] (0) at (-1, -0.5) {};
		\node [style=none] (1) at (1, -0.5) {};
		\node [style=Z dot] (2) at (0, -0.5) {};
		\node [style=idealket] (3) at (0, 0.5) {$I$};
	\end{pgfonlayer}
	\begin{pgfonlayer}{edgelayer}
		\draw (0.center) to (1.center);
		\draw (3) to (2);
	\end{pgfonlayer}
\end{tikzpicture}
 = LHS
    \end{align*}
\end{proof}

\spkComplete*

\begin{proof}
    Pick any ideals $I \subset S_k(u)$ and $J \subset S_k(v)$, and any tuples of polynomials $f \in (k[x_1, \dots, x_n])^u$, $f' \in (k[x_1, \dots, x_n])^{v}$, $g \in (k[y_1, \dots, y_m])^{u}$, and $g' \in (k[y_1, \dots, y_m])^n$. Suppose that the following diagrams are sent to the same semantic by $\sem_k$: 
    \begin{equation}
        \left[\input{./figures/span-form.tikz}\right]_k = \left[\input{./figures/span-form-alt.tikz}\right]_k
    \end{equation}
    that is, there is a span isomorphism of varieties
    \begin{equation}
        \begin{split}
            \Spec &\left(S_k(n) \xrightarrow{f^\#_I} S_k(u)/I \xleftarrow{g^\#_I} S_k(m)\right)_{red} \\ 
            &= \Spec \left(S_k(n) \xrightarrow{f'^\#_J} S_k(v)/J \xleftarrow{g'^\#_J} S_k(m)\right)_{red}
        \end{split}
    \end{equation}
    Since $\Spec: \CAlgfgredk \to \Affk^{op}$ is an equivalence of categories (Proposition \ref{prop:alg-closed-anti-equiv}), this means 
    \begin{equation}
        \begin{split}
           \left(S_k(n) \xrightarrow{f^\#_I} S_k(u)/I \xleftarrow{g^\#_I} S_k(m)\right)_{red} \\ 
            = \left(S_k(n) \xrightarrow{f'^\#_J} S_k(v)/J \xleftarrow{g'^\#_J} S_k(m)\right)_{red}
        \end{split}
    \end{equation}
    which is to say, there is an isomorphism of cospans of algebras 
    \begin{equation}
        \begin{split}
           \left(S_k(n) \xrightarrow{f^\#_{\sqrt{I}}} S_k(u)/\sqrt{I} \xleftarrow{g^\#_{\sqrt{I}}} S_k(m)\right) \\ 
            = \left(S_k(n) \xrightarrow{f'^\#_{\sqrt{J}}} S_k(v)/\sqrt{J} \xleftarrow{g'^\#_{\sqrt{J}}} S_k(m)\right)
        \end{split}
    \end{equation}
    Thus, by the completeness of $\LangCsp$ for $\FrCsp$ (Proposition \ref{prop:csp-complete}), the following rewrite exists under just the rules of $\LangCsp$: 
    \begin{equation}
        \input{./figures/span-form-radical.tikz} = \input{./figures/span-form-radical-alt.tikz}
    \end{equation}
    Thus with lemma \ref{lem:diagram-reduction} we obtain the series of rewrites 
    \begin{equation}
        \begin{split}
            \input{./figures/span-form.tikz} &= \input{./figures/span-form-radical.tikz} \\ 
            &= \input{./figures/span-form-radical-alt.tikz} \\ 
            &= \input{./figures/span-form-alt.tikz}
        \end{split}
    \end{equation}
\end{proof}

\qredSound*

\begin{proof}
    We compute directly: 
    \begin{equation}
        \begin{split}
            &\Spec\left[\input{./figures/qred-lhs.tikz}\right]_{qred} \\ 
        &= \Spec\left(k[x] = k[x] \xleftarrow{x^q} k[y]\right)_{qred} \\ 
        &= \Spec\left(k[x]/(x^q - x) = k[x]/(x_q - x) \xleftarrow{x} k[y]/(y^q - y)\right) \\ 
        &= \Spec\left(k[x]/(x^q - x) = k[x]/(x_q - x) = k[x]/(x^q - x)\right) \\ 
        &= \Spec\left(k[x] = k[x] = k[x]\right)_{qred}\\
        &= \Spec\left[\begin{tikzpicture}
	\begin{pgfonlayer}{nodelayer}
		\node [style=none] (0) at (-1, 0) {};
		\node [style=none] (1) at (1, 0) {};
	\end{pgfonlayer}
	\begin{pgfonlayer}{edgelayer}
		\draw (0.center) to (1.center);
	\end{pgfonlayer}
\end{tikzpicture}
\right]_{qred}
        \end{split}
    \end{equation}
\end{proof}

\diagramQReduction* 

\begin{proof}
    Recall from \ref{rem:basis-of-qradical} that $\Gamma_q^n \subset k[x_1, \dots, x_n]$ is the ideal generated by $x_1^q - x_1, \dots, x_n^q - x_n$. Thus 
    \begin{equation}
        \begin{split}
            \begin{tikzpicture}
	\begin{pgfonlayer}{nodelayer}
		\node [style=Z dot] (0) at (0, -0.5) {};
		\node [style=idealket] (1) at (0, 0.5) {$\Gamma_q^n$};
		\node [style=none] (2) at (-1, -0.5) {};
		\node [style=none] (3) at (1, -0.5) {};
	\end{pgfonlayer}
	\begin{pgfonlayer}{edgelayer}
		\draw (2.center) to (3.center);
		\draw (0) to (1);
	\end{pgfonlayer}
\end{tikzpicture}
 &= \input{./figures/qradical-gadget-constr.tikz} \\ 
            &\overset{\ref{ax:q-reduced}}{=} \input{./figures/qreduction-1.tikz} \\ 
            &\overset{\kadd}{=} \begin{tikzpicture}
	\begin{pgfonlayer}{nodelayer}
		\node [style=Z dot] (0) at (0, -1) {};
		\node [style=X dot] (4) at (0, 1) {};
		\node [style=none] (5) at (-2, -1) {};
		\node [style=none] (6) at (2, -1) {};
		\node [style=ukscalar] (7) at (0, 0) {$0$};
	\end{pgfonlayer}
	\begin{pgfonlayer}{edgelayer}
		\draw (5.center) to (6.center);
		\draw (0) to (4);
	\end{pgfonlayer}
\end{tikzpicture}
 \\ 
            &\overset{\zeroInK}{=} \input{./figures/qreduction-3.tikz} \\ 
            &\overset{\zunit, \xbone}{=} \ \begin{tikzpicture}
	\begin{pgfonlayer}{nodelayer}
		\node [style=none] (0) at (-1, 0) {};
		\node [style=none] (1) at (1, 0) {};
	\end{pgfonlayer}
	\begin{pgfonlayer}{edgelayer}
		\draw (0.center) to (1.center);
	\end{pgfonlayer}
\end{tikzpicture}

        \end{split}
    \end{equation}
\end{proof}

\spqComplete* 

\begin{proof}
    For ideals $I \subset S_k(u)$ and $J \subset S_k(v)$, and tuples of polynomials $f \in (k[x_1, \dots, x_n])^u$, $f' \in (k[x_1, \dots, x_n])^{v}$, $g \in (k[y_1, \dots, y_m])^{u}$, and $g' \in (k[y_1, \dots, y_m])^n$. Suppose that the following diagrams are sent to the same semantic by $\sem_q$: 
    \begin{equation}
        \left[\input{./figures/span-form.tikz}\right]_q = \left[\input{./figures/span-form-alt.tikz}\right]_q
    \end{equation}
    that is, there is a span isomorphism of varieties
    \begin{equation}
        \begin{split}
            \Spec &\left(S_{\Fq}(n) \xrightarrow{f^\#_I} S_{\Fq}(u)/I \xleftarrow{g^\#_I} S_{\Fq}(m)\right)_{qred} \\ 
            &= \Spec \left(S_{\Fq}(n) \xrightarrow{f'^\#_J} S_{\Fq}(v)/J \xleftarrow{g'^\#_J} S_{\Fq}(m)\right)_{qred}
        \end{split}
    \end{equation}
    Since $\Spec: \CAlgfgqredq \to {\Affqq}^{op}$ is an equivalence of categories (Proposition \ref{prop:finite-field-anti-equiv}), this means
    \begin{equation}
        \begin{split}
            \left(S_{\Fq}(n) \xrightarrow{f^\#_I} S_{\Fq}(u)/I \xleftarrow{g^\#_I} S_{\Fq}(m)\right)_{qred} \\
            = \left(S_{\Fq}(n) \xrightarrow{f'^\#_J} S_{\Fq}(v)/J \xleftarrow{g'^\#_J} S_{\Fq}(m)\right)_{qred}
        \end{split}
    \end{equation}
    which is to say, there is an isomorphism of cospans of algebras 
    \begin{equation}
        \begin{split}
            \left(\frac{S_{\Fq}(n)}{\Gamma_q^n} \xrightarrow{(f^\#_I)_{qred}} \frac{S_{\Fq}(u)}{I+\Gamma_q^u} \xleftarrow{(g^\#_I)_{qred}} \frac{S_{\Fq}(m)}{\Gamma_q^m}\right) \\ 
            = \left(\frac{S_{\Fq}(n)}{\Gamma_q^n} \xrightarrow{(f'^\#_J)_{qred}} \frac{S_{\Fq}(v)}{J+\Gamma_q^v} \xleftarrow{(g'^\#_J)_{qred}} \frac{S_{\Fq}(m)}{\Gamma_q^m}\right)
        \end{split}
    \end{equation}
    Composing both sides of the equation with projections maps we obtain isomorphisms of cospans 
    \begin{equation}
        \label{eq:qreduced-cospan-iso}
        \begin{split}
            \left(S_{\Fq}(n) \twoheadrightarrow \frac{S_{\Fq}(n)}{\Gamma_q^n} \xrightarrow{(f^\#_I)_{qred}} \frac{S_{\Fq}(u)}{I+\Gamma_q^u} \xleftarrow{(g^\#_I)_{qred}} \frac{S_{\Fq}(m)}{\Gamma_q^m} \twoheadleftarrow S_{\Fq}(m)\right) \\ 
            = \left(S_{\Fq}(n) \twoheadrightarrow \frac{S_{\Fq}(n)}{\Gamma_q^n} \xrightarrow{(f'^\#_J)_{qred}} \frac{S_{\Fq}(v)}{J+\Gamma_q^v} \xleftarrow{(g'^\#_J)_{qred}} \frac{S_{\Fq}(m)}{\Gamma_q^m} \twoheadleftarrow S_{\Fq}(m)\right)
        \end{split}
    \end{equation}
    Now recall, that for an algebra morphism $\varphi: A \to B$, the morphism $\varphi_{qred}: A/\Gamma_q(A) \to B/\Gamma_q(B)$ is the unique morphism making the following commute: 
    % https://q.uiver.app/#q=WzAsNCxbMCwwLCJBIl0sWzAsMSwiQS9cXEdhbW1hX3EoQSkiXSxbMSwwLCJCIl0sWzEsMSwiQi9cXEdhbW1hX3EoQikiXSxbMCwyLCJcXHZhcnBoaSJdLFsyLDMsIiIsMCx7InN0eWxlIjp7ImhlYWQiOnsibmFtZSI6ImVwaSJ9fX1dLFswLDEsIiIsMix7InN0eWxlIjp7ImhlYWQiOnsibmFtZSI6ImVwaSJ9fX1dLFsxLDMsIlxcdmFycGhpX3txcmVkfSIsMl1d
    \[\begin{tikzcd}
        A & B \\
        {A/\Gamma_q(A)} & {B/\Gamma_q(B)}
        \arrow["\varphi", from=1-1, to=1-2]
        \arrow[two heads, from=1-1, to=2-1]
        \arrow[two heads, from=1-2, to=2-2]
        \arrow["{\varphi_{qred}}"', from=2-1, to=2-2]
    \end{tikzcd}\]
    so in particular, equation \ref{eq:qreduced-cospan-iso} becomes 
    \begin{equation}
        \begin{split}
            \left(S_{\Fq}(n) \xrightarrow{f^\#_{I + \Gamma_q^u}} \frac{S_{\Fq}(u)}{I+\Gamma_q^u} \xleftarrow{g^\#_{I + \Gamma_q^u}} S_{\Fq}(m)\right) \\ 
            = \left(S_{\Fq}(n) \xrightarrow{f'^\#_{J + \Gamma_q^v}} \frac{S_{\Fq}(v)}{J+\Gamma_q^v} \xleftarrow{g'^\#_{J + \Gamma_q^v}} S_{\Fq}(m)\right)
        \end{split}
    \end{equation}
    Thus by the completeness of $\LangCsp$ for $\FrCsp$, the following rewrite can be derived from just the rewrite rules of $\LangCsp$: 
    \begin{equation}
        \input{./figures/span-form-qradical.tikz} = \input{./figures/span-form-qradical-alt.tikz}
    \end{equation}
    Thus, with lemma \ref{lem:diagram-qreduction}, we obtain the following series of rewrites: 
    \begin{equation}
        \begin{split}
            \input{./figures/span-form.tikz} &= \input{./figures/span-form-qradical.tikz} \\ 
            &= \input{./figures/span-form-qradical-alt.tikz} \\ 
            &= \input{./figures/span-form-alt.tikz}
        \end{split}
    \end{equation}
    Thus any two diagrams with the same semantic under $\sem_q$ are rewritable into each other with the rules of $\LangSpq$. So $\sem_q$ is faithful. 
\end{proof}

\subsection{Proofs about $\Mat_{\mathbb N}$ semantics for $\LangSpq$} 
\label{append:spq-matn}

\MSound* 

\begin{proof}
    First consider the identity span $\Fq^n \xleftarrow{id} \Fq^n \xrightarrow{id} \Fq^n$. For each pair $\mathbf x, \mathbf y \in \Fq^n$, it is clear that 
    \begin{equation}
        M_q\left(\xleftarrow{id} \xrightarrow{id}\right)_{\mathbf x \mathbf y} = \delta_{\mathbf x, \mathbf y}
    \end{equation}
    and so $M_q\left(\xleftarrow{id} \xrightarrow{id}\right)$ is the identity matrix. 

    Next suppose we have a composition of two structured spans of rational loci over $\Fq$: 
    % https://q.uiver.app/#q=WzAsNixbMiwwLCJWXFx1bmRlcnNldHtcXHBzaSwgXFxwaGknfSBcXHRpbWVzIFciXSxbMSwxLCJWIl0sWzMsMSwiVyJdLFswLDIsIlxcRnFebiJdLFsyLDIsIlxcRnFebSJdLFs0LDIsIlxcRnFee24nfSJdLFswLDEsIlxccGlfViIsMl0sWzAsMiwiXFxwaV9XIl0sWzEsMywiXFxwaGkiLDJdLFsxLDQsIlxccHNpIl0sWzIsNCwiXFxwaGknIiwyXSxbMiw1LCJcXHBzaSciXV0=
    \[\begin{tikzcd}[ampersand replacement=\&]
        \&\& {V\underset{\psi, \phi'} \times W} \\
        \& V \&\& W \\
        {\Fq^n} \&\& {\Fq^m} \&\& {\Fq^{n'}}
        \arrow["{\pi_V}"', from=1-3, to=2-2]
        \arrow["{\pi_W}", from=1-3, to=2-4]
        \arrow["\phi"', from=2-2, to=3-1]
        \arrow["\psi", from=2-2, to=3-3]
        \arrow["{\phi'}"', from=2-4, to=3-3]
        \arrow["{\psi'}", from=2-4, to=3-5]
    \end{tikzcd}\]
    The set of points in the pullback $V \underset{\psi, \psi'}\times W$ is the set of pairs $(\mathbf v, \mathbf w) \in V \times W$ such that $\psi \mathbf v = \psi' \mathbf w$. So we compute, for $\mathbf x \in \Fq^n, \mathbf z \in \Fq^{n'}$: 
    \begin{equation}
        \begin{split}
            &M \left(\xleftarrow{\phi}\xrightarrow{\psi}\xleftarrow{\phi'}\xrightarrow{\psi'}\right)_{\mathbf x \mathbf z} \\ 
            &= \left|\left\{(\mathbf v, \mathbf w) \in V \times W : \psi \mathbf v = \phi' \mathbf w, \phi \mathbf v = \mathbf x, \psi' \mathbf w = \mathbf z\right\}\right| \\ 
            &= \left|\bigcup_{\mathbf y \in \Fq^m}\left\{(\mathbf v, \mathbf w) \in V \times W : \substack{\phi \mathbf v = \mathbf x \\ \psi \mathbf v = \mathbf y}, \substack{\phi' \mathbf w = \mathbf y \\ \psi' \mathbf w = \mathbf z}\right\}\right|\\ 
            &= \sum_{\mathbf y \in \Fq^m} \left|\left\{\mathbf v \in V : \substack{\phi \mathbf v = \mathbf x \\ \psi \mathbf v = \mathbf y}\right\} \times \left\{\mathbf w \in W : \substack{\phi' \mathbf w = \mathbf y \\ \psi' \mathbf w = \mathbf z}\right\}\right| \\ 
            &= \sum_{\mathbf y \in \Fq^m} M_q\left(\xleftarrow{\phi} \xrightarrow{\psi}\right)_{\mathbf x \mathbf y} M_q\left(\xleftarrow{\phi'} \xrightarrow{\psi'}\right)_{\mathbf y \mathbf z}
        \end{split}
    \end{equation} 
    and therefore, 
    \begin{equation}
        M \left(\xleftarrow{\phi}\xrightarrow{\psi}\xleftarrow{\phi'}\xrightarrow{\psi'}\right) = M_q\left(\xleftarrow{\phi} \xrightarrow{\psi}\right) M_q\left(\xleftarrow{\phi'} \xrightarrow{\psi'}\right)
    \end{equation}
    As required. 

    We now prove fullness and faithfulness separately. 
\end{proof}

\begin{lemma}
    The functor $M_q$ is full. 
\end{lemma}

\begin{proof}
    It's clear that the objects in the image of $M_q$ are $q^n$, with $n \in \mathbb N$. Let $A$ be an $\mathbb N$-matrix of dimensions $q^n \times q^m$, which we interpret as a function $A: \Fq^n \times \Fq^m \to \mathbb N$. Pick a number $N \in \mathbb N$ such that 
    \begin{equation}
        q^N > \max_{\substack{\mathbf x \in \Fq^n \\ \mathbf y \in \Fq^m}} A_{\mathbf x \mathbf y}
    \end{equation}
    Then for $\mathbf x \in \Fq^n, \mathbf y \in \Fq^m$, we can find some subset $V_{\mathbf x \mathbf y} \subset \Fq^N$ such that $|V_{\mathbf x \mathbf y}| = A_{\mathbf x \mathbf y}$. Then consider the set $V \subset \Fq^n \times \Fq^m \times \Fq^N$ defined by 
    \begin{equation}
        V = \bigcup_{\substack{\mathbf x \in \Fq^n\\ \mathbf y \in \Fq^m}} \{\mathbf x\} \times \{\mathbf y\} \times V_{\mathbf x \mathbf y}
    \end{equation}
    Since every subset of affine spaces over $\Fq$ is an affine variety (this is a well-known elementary fact, see for example \cite{lidl_finite_1997,glynn_classification_1995}), $V$ is an affine variety. Then consider the structured span 
    \begin{equation}
        \Fq^n \xleftarrow{\pi_x} V \xrightarrow{\pi_y} \Fq^m 
    \end{equation}
    where $\pi_x, \pi_y$ are the restrictions of the projection maps 
    \begin{equation}
        \Fq^n \leftarrow \Fq^n \times \Fq^m \times \Fq^N \rightarrow \Fq^m.
    \end{equation}
    This is a morphism in $\FrSpq$, and  
    \begin{equation}
        M_q\left(\xleftarrow{\pi_x} \xrightarrow{\pi_y}\right)_{\mathbf x \mathbf y} = |V_{\mathbf x \mathbf y}| = A_{\mathbf x \mathbf y}
    \end{equation}
    Thus $A$ is in the image of $M_q$. 
\end{proof}

\begin{lemma}
    The functor $M_q$ is faithful. 
\end{lemma}

\begin{proof} 
    Suppose two structured spans of rational loci over $\Fq$ which are mapped by $M_q$ to the same $\mathbb N$-matrix: 
    \begin{equation}
        M_q\left(\Fq^n \xleftarrow{\phi} V \xrightarrow{\psi} \Fq^m\right) = M_q\left(\Fq^n \xleftarrow{\phi'} W \xrightarrow{\psi'} \Fq^m\right) 
    \end{equation}
    Then for each pair $\mathbf x \in \Fq^n, \mathbf y \in \Fq^m$, there exists some set-theoretic bijection
    \begin{equation}
        \sigma_{\mathbf x \mathbf y}: \phi^{-1}(\mathbf x) \cap \psi^{-1}(\mathbf y) \to \phi'^{-1}(\mathbf x) \cap \psi'^{-1}(\mathbf y)
    \end{equation}
    These bijections can be glued into 
    \begin{equation}
        \begin{split}
            \sigma: &V \to W \\ 
            &\mathbf v \mapsto \sigma_{\phi \mathbf v, \psi \mathbf v} (\mathbf v)
        \end{split}
    \end{equation}
    Since every function between rational loci over $\Fq$ is a polynomial mapping \cite{glynn_classification_1995}, $\sigma$ is a polynomial mapping. Moreover, by construction the following commutes:
    % https://q.uiver.app/#q=WzAsNCxbMSwwLCJWIl0sWzEsMiwiVyJdLFswLDEsIlxcRnFebiJdLFsyLDEsIlxcRnFebSJdLFswLDIsIlxccGhpIiwyXSxbMCwzLCJcXHBzaSJdLFsxLDIsIlxccGhpJyJdLFsxLDMsIlxccHNpJyIsMl0sWzAsMSwiXFxzaWdtYSIsMV1d
    \[\begin{tikzcd}[ampersand replacement=\&]
        \& V \\
        {\Fq^n} \&\& {\Fq^m} \\
        \& W
        \arrow["\phi"', from=1-2, to=2-1]
        \arrow["\psi", from=1-2, to=2-3]
        \arrow["\sigma"{description}, from=1-2, to=3-2]
        \arrow["{\phi'}", from=3-2, to=2-1]
        \arrow["{\psi'}"', from=3-2, to=2-3]
    \end{tikzcd}\]
    and therefore $\sigma$ is a structured span isomorphism. 
\end{proof}

\section{Proof of GAG as CSP} 

\GAGasCSP*

\begin{proof}
    We know from proposition \ref{prop:quotient-gadget} that the \textit{cospan semantic} of the ideal gadget is 
    \begin{equation}
        \left[\input{./figures/ring-quot-gadget-numbered.tikz}\right] = \left(S_k(n) \twoheadrightarrow S_k(n)/I \twoheadleftarrow S_k(n)\right)
    \end{equation}
    Hitting the cospan of algebras with the functor $\Spec(-)_{\mathsf{qred}}: \CAlgfgq \to \Affqq$ yields equation \ref{eq:ideal-gadget-as-csp-instance}. Then equation \ref{eq:count-csp} follows immediately by composing on both left and right legs with the projection $\Fq^n \xrightarrow{!} *$ where $*$ is the singleton. 
\end{proof}
\section{Proofs about ZH calculus as an extension of GAG}
\label{append:zh}

\iotaZHSound*
\begin{proof}
    We check that the image of each generator of $\ZH$ under $\iota_{\ZH}$ is mapped to the same matrix as under $\ZHsem$.\\
    Z-spider:
    \begin{equation*}
        \semZOne[\input{./figures/Z-spider.tikz}] = \sum_{i \in \Fq} \ket{i}^{\otimes m} \bra{i}^{\otimes n} = \ZHsem[\input{./figures/Z-spider.tikz}]
    \end{equation*}
    H-spider:
    \begin{align}
        &\semZOne[\input{./figures/H-spider-in-gag.tikz}] \\
        =& \frac{1}{\sqrt{q}}\left( \sum_{\substack{j_1, ..., j_m \in \Fq \\ i_1, ..., i_n \in \Fq}} \!\!\!  \ket{j_1} ... \ket{j_m} \bra{i_1} ... \bra{i_n} \otimes \bra{j_1...j_m i_1...i_n} \right)
        \left(\sum_{k \in \Fq} \omega^{\tr(k)} \ket{k} \right)\\
        =&\ \frac{1}{\sqrt{q}}\sum_{\substack{j_1, ..., j_m \in \Fq \\ i_1, ..., i_n \in \Fq}} \!\!\! \omega^{\tr(j_1...j_m i_1...i_n)} \ket{j_1} ... \ket{j_m} \bra{i_1} ... \bra{i_n} \\
        =&\ \ZHsem[\input{./figures/H-spider.tikz}]
    \end{align}
    X-basis state:
    \begin{align}
        \semZOne[\input{./figures/x-basis-state-in-gag.tikz}]
        \!\!\!= q^{-\frac{1}{2}}q \left( \sum_{i \in \Fq} \ket{ji} \bra{i} \right) \ket{1}
        = \sqrt{q}\ket{j}
        = \ZHsem[\begin{tikzpicture}
	\begin{pgfonlayer}{nodelayer}
		\node [style=X phase dot] (0) at (-0.75, 0) {$j$};
		\node [style=none] (1) at (0.5, 0) {};
	\end{pgfonlayer}
	\begin{pgfonlayer}{edgelayer}
		\draw (0) to (1.center);
	\end{pgfonlayer}
\end{tikzpicture}
]
    \end{align}
\end{proof}

\copyDeleteZOneSound*
\begin{proof}
    First we check the \ZOneCopy rule:
    \begin{align}
        \semZOne[L.H.S.] = & \left(\sum_{j,k \in \Fq} \ket{j}\ket{k}\bra{j+k} \right) \left( \sum_{i \in \Fq} \omega^{\tr(i)} \ket{i} \right)\\
        = & \sum_{j,k \in \Fq} \omega^{\tr(j+k)} \ket{j}\ket{k} \\
        = & \left( \sum_{j \in \Fq} \omega^{\tr(j)} \ket{j} \right) \otimes \left( \sum_{k \in \Fq} \omega^{\tr(k)} \ket{k} \right)\\
        = & \semZOne[R.H.S.]
    \end{align}
    Next we check the \ZOneDel rule:
    \begin{equation}
        \semZOne[L.H.S.]
        = \bra{0} \left( \sum_{i \in \Fq} \omega^{\tr(i)} \ket{i} \right)
        = \omega^{\tr(0)}
        = 1
        = \semZOne[R.H.S.]
    \end{equation}
\end{proof}

\pullingZOneOut*
\begin{proof}
    For all $\begin{tikzpicture}
	\begin{pgfonlayer}{nodelayer}
		\node [style=scalar] (0) at (0.5, 0) {$q^{-\frac{1}{2}}$};
	\end{pgfonlayer}
\end{tikzpicture}
$ in $D$, we can pull them out to the front of the diagram using the monoidal structure.
    Then we only need to consider the instances of the $\begin{tikzpicture}
	\begin{pgfonlayer}{nodelayer}
		\node [style=Z phase dot] (0) at (-0.75, 0) {$1$};
		\node [style=none] (1) at (0.5, 0) {};
	\end{pgfonlayer}
	\begin{pgfonlayer}{edgelayer}
		\draw (0) to (1.center);
	\end{pgfonlayer}
\end{tikzpicture}
$ generator.
    This holds trivially if the diagram $D$ contains a single instance of the $\begin{tikzpicture}
	\begin{pgfonlayer}{nodelayer}
		\node [style=Z phase dot] (0) at (-0.75, 0) {$1$};
		\node [style=none] (1) at (0.5, 0) {};
	\end{pgfonlayer}
	\begin{pgfonlayer}{edgelayer}
		\draw (0) to (1.center);
	\end{pgfonlayer}
\end{tikzpicture}
$.
    If there are two or more instances, we can use the \ZOneCopy rule to merge multiple instances into one.
    When there are no instances of $\begin{tikzpicture}
	\begin{pgfonlayer}{nodelayer}
		\node [style=Z phase dot] (0) at (-0.75, 0) {$1$};
		\node [style=none] (1) at (0.5, 0) {};
	\end{pgfonlayer}
	\begin{pgfonlayer}{edgelayer}
		\draw (0) to (1.center);
	\end{pgfonlayer}
\end{tikzpicture}
$, we can use the \ZOneDel rule in reverse to introduce one copy.
\end{proof}

\zhAmplitudeViaGAG*
\begin{proof}
    From Proposition~\ref{prop:iota-zh-sound}, we have that $\ZHsem[D] = \semZOne[\iota_{\ZH}(D)]$.
    \begin{align*}
        \semZOne[\ \input{./figures/zh-gag-oracle-1.tikz}\ ]
        = \semZOne[\input{./figures/zh-gag-oracle-2.tikz}]
        = q^{-\frac{k}{2}} \sum_{j \in \Fq} \semZOne[\input{./figures/zh-gag-oracle-3.tikz}]\\
        = q^{-\frac{k}{2}} \sum_{j \in \Fq} \semqMat[\input{./figures/zh-gag-oracle-3.tikz}]
    \end{align*}
    Thus, we can compute $\ZHsem[D]$ using $q$ queries to the oracle $\mathcal{O}$ that computes $\semqMat$ on scalar diagrams in $\LangSpq$.
\end{proof}

\end{document}